\documentclass[10pt,a4paper]{article}
\usepackage[utf8]{inputenc}
\usepackage{setspace}
\usepackage{amsmath}
\usepackage{amsfonts}
\usepackage{amssymb}
\usepackage{amsthm}
\usepackage{graphicx}
\usepackage{enumitem}
\usepackage{algpseudocode}
\usepackage{algorithm}

\usepackage[left=1in, right=1in]{geometry}
\usepackage[title]{appendix}
\usepackage{authblk}
\usepackage{xcolor}

\usepackage[mathscr]{euscript}
\usepackage{bbm}

\newtheorem{theorem}{Theorem}
\newtheorem{lemma}{Lemma}

\theoremstyle{definition}
\newtheorem{conditionA}{Condition}
\newtheorem{conditionB}{Condition}
\newtheorem{conditionC}{Condition}

\newtheorem{remark}{Remark}

\usepackage[width=.9\textwidth]{caption}

\usepackage[square,numbers]{natbib}
\title{Estimating the Efficiency Gain of Covariate-Adjusted Analyses in Future Clinical Trials Using External Data}
\author[*]{Xiudi Li}
\author[*]{Sijia Li}
\author[$\dagger$]{Alex Luedtke}
\affil[*]{Department of Biostatistics, University of Washington}
\affil[$\dagger$]{Department of Statistics, University of Washington}

\date{}

\begin{document}
\maketitle

\begin{abstract}
We present a general framework for using existing data to estimate the efficiency gain from using a covariate-adjusted estimator of a marginal treatment effect in a future randomized trial. We describe conditions under which it is possible to define a mapping from the distribution that generated the existing external data to the relative efficiency of a covariate-adjusted estimator compared to an unadjusted estimator. Under conditions, these relative efficiencies approximate the ratio of sample size needed to achieve a desired power. We consider two situations where the outcome is either fully or partially observed and several treatment effect estimands that are of particular interest in most trials. For each such estimand, we develop a semiparametrically efficient estimator of the relative efficiency that allows for the application of flexible statistical learning tools to estimate the nuisance functions and an analytic form of a corresponding Wald-type confidence interval.
We also propose a double bootstrap scheme for constructing confidence intervals. We demonstrate the performance of the proposed methods through simulation studies and apply these methods to data to estimate the relative efficiency of using covariate adjustment in Covid-19 therapeutic trials.
\end{abstract}

\section{Introduction}
The aim of most clinical trials is to estimate a marginal treatment effect that contrasts outcomes in a treatment group to those in a control group. In addition to the treatment assignment and outcome, data on prognostic baseline covariates are often available. In the case of continuous outcomes, the U.S. Food and Drug Administration \citep{FDA2019} recommends adjusting for these baseline covariates through ANCOVA or linear regression models. However, such covariate-adjusted analyses are often underutilized in practice, especially with ordinal or time-to-event data \citep{austin2010substantial}. 

Compared with unadjusted analyses, analyses that adjust for baseline covariates have several benefits. First, adjusted analyses can lead to consistent estimators of the treatment effect under weaker assumptions. One such example arises when right-censoring is present in a trial with a time-to-event outcome. Adjusted analyses often give consistent estimates provided that the censoring and survival times are independent given treatment and covariates \citep{moore2009increasing}. This condition is more plausible in many trial settings than is the requirement made in unadjusted analyses that the censoring and survival times are independent given treatment alone. Second, adjusting for covariates that are predictive of the outcome can improve precision, and thus a smaller sample size can be required to achieve a desired power. Such precision gain is generally expected when the outcome is fully observed, and also applies in certain cases where the outcome is only partially observed --- some exceptions occur, for example, when the covariates are highly predictive of the censoring time but are only weakly predictive of the survival time. Despite these potential benefits, covariate adjustment is underutilized in analyzing clinical trial data. This is partly because, at the trial planning stage, there is typically little prior knowledge about the amount of precision gain that should be expected to result from using covariate adjustment.

To address this problem, many previous works have aimed to estimate this precision gain using external datasets. In particular, some works have demonstrated the potential precision gain of covariate adjustment in clinical trial settings by comparing the standard errors of adjusted and unadjusted estimators on existing clinical trial datasets \cite[e.g.,][]{steingrimsson2017improving,wang2019analysis}. When the data-generating mechanism that gave rise to one of these existing datasets is reflective of the corresponding mechanism that is anticipated in an upcoming trial, these point estimates may yield a reasonable estimate of the precision gain anticipated in these future trials. It is worth noting, however, that the sampling variability in the existing trial dataset induces uncertainty in this point estimate. Other works have used an existing trial dataset to design a simulation study that can be used to estimate the precision gain \citep[e.g.,][]{kahan2014risks,colantuoni2015leveraging,steingrimsson2017improving}. However, even when these simulation studies involve many repetitions, so that the Monte Carlo error is negligible, there is still uncertainty associated with these precision gain estimates that arises due to sampling variability in the existing trial dataset. In many cases, there may not be data available from a clinical trial that is reflective of the upcoming trial. An alternative approach, which does not require access to such data but can leverage it when it is available, is to conduct a simulation based on an external dataset that may be reflective of the covariate and outcome distributions that will be seen in the control arm of the upcoming trial \citep{benkeser2020improving}. This data may, for example, be derived from a pilot study or an observational study. Treatment arm data can then be simulated under a sharp null of no effect or as a user-defined shift of the conditional distribution of the outcome given covariates in the pilot study. As for the earlier simulation approach, a point estimate of the precision gain can be easily obtained, but there is still uncertainty in this point estimate that arises from the sampling variability in the dataset upon which the simulation is based. 

It can be challenging to be confident that a favorable estimated precision gain is not due to random noise, especially when the external dataset is small. Consequently, some clinical trialists may be cautious when making decisions about using covariate adjustment in future clinical trials based on a point estimate alone, even if the external dataset upon which it is based is known to be reflective of the data that will be seen in the trial. 
In other statistical estimation problems, the lack of interpretability of point estimates is often addressed by reporting a confidence interval alongside each point estimate. However, to the best of our knowledge, the problem of making statistical inferences about the precision gains of covariate-adjusted estimators has not been formally investigated. In this work, we aim to fill this critical knowledge gap. When doing so, we focus on the most general case described above, namely that data from an external study that is reflective of the covariate and control arm outcome distributions are available. Special cases of this setting include the case where data are available from a previous trial and the control arm data are used for the external study, and also the case where covariate and outcome data are available from an observational study.

We primarily consider treatment effect estimands that can be written as contrasts of the distributions of the outcomes within each treatment arm. Most commonly, investigators perform an unadjusted analysis that uses the empirical distribution to estimate these two arm-specific distributions. One approach to covariate adjustment involves instead estimating these distributions with possibly-misspecified working models. Specifically, this involves fitting a working parametric model within each arm that conditions on covariates, and then marginalizing over the arm-pooled empirical distribution of the covariates \citep{moore2011robust,benkeser2020improving}. 
In many cases, this approach can result in consistent and asymptotically normal estimators of the marginal effect of interest, even if the working model is misspecified.
Nevertheless, these approaches are typically inefficient when the model is not correct, in the sense that they fail to achieve the asymptotic efficiency bound within a model that only imposes that treatment is randomized \citep{bickel1993efficient}. 
In contrast, many covariate-adjusted estimators have been proposed recently that achieve the efficiency bound under appropriate regularity conditions (see, for example, \citep{vermeulen2015increasing,diaz2016enhanced} for ordinal outcomes; and \citep{moore2009increasing,stitelman2012general,diaz2019improved} for survival outcomes). These approaches usually involve estimating nuisance parameters such as the treatment mechanism and the outcome regression functions. 
While being more efficient, these estimators are often more difficult for practitioners to understand because they cannot typically be framed as corresponding to a commonly used estimator within a parametric working model. In this paper, we consider both of the above-described strategies for covariate adjustment, which we refer to as working-model-based approaches and fully adjusted approaches, respectively.

Our main contributions are as follows:
\begin{enumerate}
    \item we provide a framework for using external data to identify the efficiency gain in terms of percentage reduction in sample size needed to achieve a desired power from using covariate-adjusted rather than unadjusted estimation methods on future clinical trial data;
    \item we introduce efficient estimators of this quantity that allow for the incorporation of flexible statistical learning tools to estimate the needed nuisance functions;
    \item we present statistical inference procedures to accompany the proposed estimators, namely a Wald-type procedure that requires knowledge of their influence functions but is widely applicable and a bootstrap procedure that applies only to working-model-based estimators but is easy to implement; and
    \item we evaluate the performance of the proposed methods in a simulation study and an application to a dataset of Covid-19 patients hospitalized at the University of Washington Medical Center.
\end{enumerate}

This paper is organized as follows. In Section \ref{sec:background}, we provide background on efficient estimation in semiparametric models and describe the relevance of the relative efficiency and local alternatives to clinical trial settings. In Section \ref{sec:ordinal}, we introduce the framework to identify and estimate the efficiency gain when the outcome is fully observed. We also propose efficient estimators and develop analytical and bootstrap inference procedures for estimands that are of frequent interest in the cases of continuous and ordinal outcomes. In Section \ref{sec:partiallyobserved}, we study the case where the outcome is partially observed and consider time-to-event outcomes with right-censoring as an example. In Section \ref{sec:experiment}, we demonstrate the performance of the proposed methods through simulation experiments and an analysis of a real dataset. Section \ref{sec:discussion} concludes with a discussion.

\section{Review of efficiency theory and its relevance to clinical trial settings}\label{sec:background}

\subsection{Pathwise differentiability and regular and asymptotically linear estimators}

The theory of efficient estimation in nonparametric and semiparametric models was described in \citet{pfanzagl1990estimation} and \citet{bickel1993efficient}. Here we give a brief review of the relevant concepts. Let $X$ denote a generic data unit with distribution $P$ and let $L_0^2(P) := \{f: E_P[f(X)]=0, \textnormal{var}_P[f(X)]<\infty\}$. 
Let $\mathcal{M}$ denote a statistical model, that is, a collection of distributions of $X$. 
We suppose that $\mathcal{M}$ contains $P$. Let $\mathscr{M}(P)$ denote the collection of all one-dimensional submodels $\{P_\epsilon : \epsilon\in\mathbb{R}\}\subseteq\mathcal{M}$ that are quadratic mean differentiable \citep{van2000asymptotic} at $\epsilon=0$ and are such that $P_{\epsilon=0}=P$. Let $S_{\mathcal{M}}(P)$ denote the collection of all functions $s : \mathcal{X}\rightarrow\mathbb{R}$ for which $s$ is a score function at $\epsilon=0$ for some submodel contained in $\mathscr{M}(P)$, and let $T_\mathcal{M}(P)$ denote the $L_0^2(P)$-closure of the linear span of $S_{\mathcal{M}}(P)$. The subspace $T_{\mathcal{M}}(P)$ of $L_0^2(P)$ is referred to as the tangent space. A parameter $\psi : \mathcal{M}\rightarrow\mathbb{R}$ is called pathwise differentiable at $P$ in $\mathcal{M}$ if there exists a function $D \in L_0^2(P)$ such that, for all submodels $\{P_\epsilon : \epsilon\in\mathbb{R}\}\in \mathscr{M}(P)$, it holds that $\frac{\partial}{\partial \epsilon}\psi(P_\epsilon)|_{\epsilon=0} = E_P[ D(X)s(X)]$, where $s$ is the score function of $\{P_\epsilon : \epsilon\in\mathbb{R}\}$ at $\epsilon=0$. Any such function $D$ is called a gradient of $\psi$ with respect to $\mathcal{M}$ at $P$. The canonical gradient $D^*$ is the gradient that lies in the tangent space $T_{\mathcal{M}}(P)$ --- it can be shown that this gradient is unique. We note that the $\mathcal{X}\rightarrow\mathbb{R}$ functions $D$ and $D^*$ both depend on $P$.

We refer to an estimator $\hat{\psi}$ of $\psi$ as a random variable that is a function of an independent and identically distributed (iid) sample $\boldsymbol{X} := \{X_1,\ldots,X_n\}$ drawn from some distribution. An estimator $\hat{\psi}$ is called regular if there exists a real-valued probability distribution $\mathcal{L}$ such that, for all submodels $\{P_{\epsilon} : \epsilon\in\mathbb{R}\}$ in $\mathscr{M}(P)$ and all $c\in\mathbb{R}$,
\begin{equation}\label{regularity}
    \sqrt{n}\left\{\hat\psi - \psi\left(P_{cn^{-1/2}}\right)\right\} \xrightarrow[]{P_{cn^{-1/2}}} \mathcal{L}.
\end{equation}
Importantly, note that, if an estimator is regular, then the distribution $\mathcal{L}$ above does not depend on the choice of submodel in $\mathscr{M}(P)$. 

An estimator $\hat\psi$ of $\psi(P)$ is called asymptotically linear if there exists some function $\xi_P\in L_0^2(P)$ such that
\begin{equation*}
    \hat\psi - \psi(P) = \frac{1}{n}\sum_{i=1}^{n}\xi_P(X_i) + o_P(n^{-1/2}).
\end{equation*}
The function $\xi_P$ is referred to as the influence function of $\hat\psi$. Asymptotically linear estimators are consistent and asymptotically normal, in the sense that
\begin{equation}\label{normality}
    \sqrt{n}\{\hat\psi - \psi(P)\} \xrightarrow[]{d} N(0,\sigma^2_P),
\end{equation}
where $\sigma^2_P$ is the variance of $\xi_P(X)$ when $X \sim P$. If $\hat\psi$ is asymptotically linear and $\psi$ is pathwise differentiable, then $\hat\psi$ is regular if, and only if, $\xi_P$ is a gradient of $\psi$ with respect to $\mathcal{M}$ at $P$. Among the collection of gradients, the canonical gradient $D^*$ has the smallest variance, and is also called the efficient influence function (EIF). This variance characterizes the efficiency bound of estimating $\psi$ given the model $\mathcal{M}$ with a regular and asymptotically linear (RAL) estimator. An estimator is called efficient if it is RAL and its influence function is the same as the EIF.

Suppose that we have available an initial estimator $\hat P$ of the distribution. A plug-in estimator is defined as $\psi(\hat P)$. However, such estimators may not be $\sqrt{n}$-consistent due to the potential bias in the initial estimators. One way to construct a RAL estimator with influence function $D$ is through one-step estimation \citep{ibragimov1981statistical,bickel1982adaptive,pfanzagl1985contributions}, which corrects for this bias by using $\hat\psi = \psi(\hat P)+\mathbb{P}_n D(\hat P)$ where $\mathbb{P}_n(\cdot)$ is the empirical mean. Estimating equations \citep{van2003unified,tsiatis2007semiparametric} and targeted minimum loss-based estimation \citep{van2006targeted} are alternative approaches.
These techniques are often used to construct efficient covariate-adjusted estimators of a treatment effect. Later, we will also use them to estimate the relative efficiency of two estimators based on external data.

\subsection{Local alternatives, relative efficiency, and their relevance to clinical trial settings}\label{sec:localAlts}

In the context that we consider in this paper, the treatment effect measure that will be estimated with the future clinical trial data will often correspond to an evaluation $\psi(P)$ of a pathwise differentiable parameter $\psi$. In many cases, a primary objective of the forthcoming trial will be to test the null hypothesis that this quantity is equal to zero against a one-sided alternative, for example, that this quantity is positive. Suppose that a level $\alpha$ Wald test is performed, which corresponds to evaluating whether zero is smaller than $\hat\psi - n^{-1/2}z_{1-\alpha}\hat\sigma_P$ based on a RAL estimator $\hat\psi$, where $z_{1-\alpha}$ is the $(1-\alpha)$-quantile of a standard normal distribution and $\hat\sigma_P$ is a consistent estimator of $\sigma_P$, as defined in \eqref{normality}. Fix an arbitrary $c\not=0$. As $\psi$ is pathwise differentiable, for any $\{P_{\epsilon} : \epsilon\in\mathbb{R}\}$ in $\mathscr{M}(P)$ with score function $s$ at $\epsilon=0$, it holds that $\psi(P_{cn^{-1/2}}) = \psi(P) + cn^{-1/2}\mu_{P,s} + o(n^{-1/2})$, where $\mu_{P,s}:=E_P[D^*(X)s(X)]$. If $P$ is such that the null that $\psi(P)=0$ holds, then this shows that $n^{1/2}\psi(P_{cn^{-1/2}})\rightarrow c\mu_{P,s}$ as $n\rightarrow\infty$. If $s$ is such that $\mu_{P,s} > 0$, then we call $\{P_{cn^{-1/2}}\}_{n=1}^\infty$ a sequence of local alternatives --- this name is natural given that $\psi(P_{cn^{-1/2}})>0$ for all $n$ large enough (and so the alternative holds for all $n$ large enough), while also $\psi(P_{cn^{-1/2}})\rightarrow 0$ as $n\rightarrow\infty$ (and so these alternatives are local to the null hypothesis). Because $\hat\psi$ is RAL, combining \eqref{regularity} and \eqref{normality} with Slutsky's theorem implies that
\begin{align}\label{eq:localNormal}
    \sqrt{n}\hat\psi \xrightarrow[]{P_{cn^{-1/2}}} N\left(c\mu_{P,s},\sigma^2_P\right),
\end{align}
where $\sigma_P^2$ is as defined below \eqref{normality}. Let $\beta$ denote the power for rejecting the null that $\psi(P)=0$ under sampling from $P_{cn^{-1/2}}$ --- that is, let
\begin{align}
    \beta:= \lim_{n\rightarrow\infty} P_{cn^{-1/2}}\left\{0 < \hat\psi - n^{-1/2}z_{1-\alpha}\hat\sigma_P\right\}. \label{eq:beta}
\end{align}
Letting $\Phi(\cdot)$ denote the cumulative distribution function (CDF) of the standard normal distribution, \eqref{eq:localNormal} implies that $\beta = 1-\Phi(z_{1-\alpha}-c\mu_{P,s}/\sigma_P)$ which lies in $(\alpha,1)$ when the shift in the mean of the limit normal distribution $c\mu_{P,s}$ is positive. This gives us a way to quantify the power of the test in a range of settings where the effect size is small. Hence, effect sizes scaling as $n^{-1/2}$ are interesting in general testing problems, given that it is exactly at these effect sizes that the problem is neither asymptotically trivial (power converging to $1$) or impossible (power converging to $\alpha$). Nevertheless, in many statistical problems, there may be no \textit{a priori} reason to believe that the effect size will be of the order $n^{-1/2}$.

The setup is quite different in randomized trials. Indeed, these local alternatives are natural to think about in these settings because, under sampling from such a sequence of alternatives, the asymptotic power takes some intermediate value between $\alpha$ and $1$, which reflects the fact that the sample size in most trials is specified so that a test of the null will have a chosen power $\beta^*\in(\alpha,1)$. To be more concrete, suppose that $\{\psi^{(k)}\}_{k=1}^\infty$ is a decreasing sequence of effect sizes that satisfy the alternative hypothesis, that is, that are such that $\psi^{(k)}\downarrow 0$ as $k\rightarrow\infty$. We suppose that these effect sizes arise from some sequence of distributions $\{P^{(k)}\}_{k=1}^\infty$ that belong to some submodel $\mathcal{M}_1:= \{P_\epsilon : \epsilon\in\mathbb{R}\}\in\mathscr{M}(P)$, so that $\psi(P^{(k)})=\psi^{(k)}$ for each $k$. Our objective is to establish an expression for the sequence of sample sizes $\{n^{(k)}\}_{k=1}^\infty$ so that, as $k\rightarrow\infty$, the power for rejecting the null hypothesis converges to $\beta^*$ when $n^{(k)}$ iid observations are drawn from $P^{(k)}$.  Let $s$ denote the score of $\epsilon$ at $0$ in the submodel $\mathcal{M}_1$, and suppose that $\mu_{P,s}\not=0$. To derive the sequence $\{n^{(k)}\}_{k=1}^\infty$, it will be helpful to first find a $c$ such that, when $n$ iid observations are drawn from $P_{cn^{-1/2}}$, the power converges to the desired $\beta^*$ as $n\rightarrow\infty$. Because $\{P^{(k)}\}_{k=1}^\infty\subseteq \mathcal{M}_1$, it will also be possible to find an $n^{(k)}$ such that $P^{(k)}\approx P_{c/\sqrt{n^{(k)}}}$. As $\{n^{(k)}\}_{k=1}^\infty$ is a subsequence of $\{n\}_{n=1}^\infty$, it will then be reasonable to expect that, when a sequence of tests is conducted based on $n^{(k)}$ iid observations sampled from each $P^{(k)}$, the power for rejecting the null hypothesis will converge to $\beta^*$ as $k\rightarrow\infty$.

We now find the expressions for $c$ and $n^{(k)}$ that were described in the last paragraph. Recalling \eqref{eq:beta} and the alternative expression for the power given below that display, we see that, when $c= \sigma_P(z_{1-\alpha} - z_{1-\beta^*})/\mu_{P,s}$, it holds that
\begin{align*}
\lim_{n\rightarrow\infty} P_{cn^{-1/2}}\left\{0 < \hat\psi - n^{-1/2}z_{1-\alpha}\hat\sigma_P\right\}&= 1-\Phi(z_{1-\alpha}-c\mu_{P,s}/\sigma_P) = \beta^*.
\end{align*}
To find the expression for $n^{(k)}$, we note that, as $\psi$ is pathwise differentiable, $\psi(P_{cn^{-1/2}})=cn^{-1/2} \mu_{P,s} + o(n^{-1/2})$ when $n$ is large --- here, the little-oh term describes behavior as $n\rightarrow\infty$. Hence, $P^{(k)}\approx P_{c/\sqrt{n^{(k)}}}$, where $n^{(k)}=\lceil(c\mu_{P,s}/\psi^{(k)})^2\rceil = \lceil \sigma_P^2(z_{1-\alpha} - z_{1-\beta^*})^2/(\psi^{(k)})^2 \rceil$. Thus, to achieve asymptotic power $\beta^*$ when sampling $n^{(k)}$ iid observations from $P^{(k)}$, $n^{(k)}$ must scale proportionally with $\sigma_P^2$. These calculations also provide a means to compare the sample sizes needed to achieve the same power based on two different RAL estimators. In particular, suppose that a second RAL estimator is available and its asymptotic variance is equal to $\tilde{\sigma}_P^2\le \sigma_P^2$. In this case, the proportional reduction in sample size needed to achieve power $\beta^*$ when using this estimator rather than the estimator with variance $\sigma_P^2$ is approximately equal to $1-\tilde\sigma^2_P/\sigma^2_P$ when $k$ is large. In fact, $\tilde{\sigma}_P^2/\sigma_P^2$ is often referred to as the relative efficiency of the RAL estimator with variance $\tilde{\sigma}_P^2$ versus the RAL estimator with variance $\sigma_P^2$, and so this proportional reduction is exactly equal to one minus the relative efficiency of these two estimators.

\section{When the outcome is fully observed}\label{sec:ordinal}

\subsection{Framework to identify the relative efficiency}\label{subsec:setup}

We now propose a general framework to identify the relative efficiency of covariate-adjusted estimators.

We start by defining notation that we will use to describe the data that will arise in the future clinical trial. Let $A$ denote the binary treatment, $W$ denote the $d$-dimensional covariate vector and $Y$ denote the outcome. We will use superscript $t$ to denote random variables in a future clinical trial. Let $P_1$ be the conditional distribution of $Y^t|A=1$ and $P_0$ be the conditional distribution of $Y^t|A=0$. The treatment effect is often a functional $f$ of these distributions, i.e., $\psi = f(P_1,P_0)$ for some $f$. An example is the average treatment effect for continuous outcome, where $f(P_1,P_0) = E[Y^t|A=1]- E[Y^t|A=0]$. The observation unit in a trial is $X^t = (Y^t,A,W^t)$. Randomization implies that $A$ and $W^t$ are independent, and thus that the distribution of $X^t$, denoted by $\nu$, is determined by the marginal distribution of $A$, the conditional distribution of $Y^t$ given $(A=0,W^t=w)$, the conditional distribution of $Y^t$ given $(A=1,W^t=w)$, and the marginal distribution of $W^t$ --- we denote these distributions by $\Pi$, $P_0^t$, $P_1^t$, and $P_W$, respectively. Let $\mathcal{M}_X$ consist of all distributions for which $A$ and $W^t$ are independent. In the randomized trial settings that we consider, $\nu\in\mathcal{M}_X$. The adjusted analysis uses $X^t$, while the unadjusted analysis ignores the covariate $W^t$. Let $\hat\psi_u, \hat\psi_a$ and $\hat\psi_{m}$ be a specified unadjusted estimator, fully adjusted estimator and working-model-based adjusted estimator of $\psi$, respectively. Further suppose that these estimators are regular and asymptotically linear with influence functions $D_u$, $D_a$ and $D_{m}$, respectively. We note that these influence functions all depend on the underlying distribution $\nu$, but we will omit this dependence when it is clear from the context. 

We assume the following regularity condition holds throughout, which guarantees that the treatment effects of interest can be estimated using strategies typically employed in randomized trial settings. Although some of the conditions can in principle be relaxed, they cover most realistic clinical trial settings. 
\begin{conditionA}\label{cond2bounded}
    Treatment is independent of covariates $(A \perp W^t)$, the treatment probability $\Pi(A=1)$ falls in $(0,1)$, and $Y$ and $W^t$ both have bounded support under sampling from $\nu$. 
\end{conditionA}
Let $\mathcal{W} \subseteq \mathbb{R}^d$ and $\mathcal{Y} \subseteq \mathbb{R}$ be bounded and convex sets that contain the support of $W^t$ and $Y$, respectively.

Our objective is to quantify the relative precision of the specified adjusted and unadjusted estimators. To do this, we will consider the relative efficiency of these two estimators, defined as the ratio between the asymptotic variances of the adjusted estimator and the unadjusted estimator under a sharp null distribution. Though we will focus on the sharp null when introducing these relative efficiencies, these quantities also correspond to the relative efficiencies under local alternatives (see Section \ref{sec:localAlts}). Consequently, our apparent restriction to the sharp null setting will in fact not be restrictive at all. Indeed, from this sharp null setting, we can typically approximate the reduction in sample size needed to achieve a desired power at all local alternatives that are consistent with the design alternative used to size the trial (ibid.).

We now define these relative efficiencies. Consider a trial where the sharp null holds, that is, the treatment has no effect and $P_1^t = P_0^t$. In this case, let $P$ denote the joint distribution of $(Y^t,W^t)$, determined by the pair $(P_1^t,P_W)$. Under Condition~\ref{cond2bounded}, $\nu$ is equal to the product measure $P\Pi$ in this sharp null setting. The relative efficiency of the fully adjusted estimator compared to that of the unadjusted estimator is defined as
\begin{equation}\label{refully}
    \phi_{a}(P) := \frac{E_{P\Pi}[D_a(P\Pi)(X^t)^2]}{E_{P\Pi}[D_u(P\Pi)(X^t)^2]}=\frac{E_{\nu}[D_a(\nu)(X^t)^2]}{E_{\nu}[D_u(\nu)(X^t)^2]},
\end{equation}
whereas the analogous quantity for the working-model-based estimator is defined as
\begin{equation}\label{reworking}
    \phi_{m}(P) := \frac{E_{P\Pi}[D_{m}(P\Pi)(X^t)^2]}{E_{P\Pi}[D_u(P\Pi)(X^t)^2]} =\frac{E_{\nu}[D_{m}(\nu)(X^t)^2]}{E_{\nu}[D_u(\nu)(X^t)^2]}.
\end{equation}
Although in general $\nu$ depends on both $\Pi$ and $P$, we define relative efficiencies as functions of $P$ only. As we will show, in many cases that are of interest in practice, the relative efficiency does not depend on $\Pi$. Even in cases where it does depend on $\Pi$, the investigator in a trial would have control over the treatment distribution $\Pi$ in the trial setting, and the only unknown component would still be $P$.

The local alternatives we consider allow for a variety of perturbations to the underlying distribution.
The direction of these perturbations is described by their scores, which belong to the tangent space $T_{\mathcal{M}_X}(\nu)$. This tangent space decomposes into three subspaces, corresponding to the marginal distribution of $W^t$, the distribution of treatment $A$, and the conditional distribution of $Y^t|A,W^t$. The score in a smooth submodel can lie in one or more subspaces, which means that the local alternative can perturb one or more of the four components of $\nu$, namely $\Pi$, $P_0^t$, $P_1^t$, and $P_W$. For an example of how these perturbations may impact $\nu$, consider the special case of the average treatment effect $\psi = E[Y^t|A=1]-E[Y^t|A=0]$ that was introduced at the beginning of this section. We note that $\psi = E_{P_W}[c(W^t)]$, where $c(w):= E[Y^t|A=1,W^t=w] - E[Y^t|A=0,W^t=w]$ is the conditional average treatment effect function. Here, $\psi$ will be zero when this function is zero for all values of the covariates. Now, for any $L^2(P_W)$ integrable function $f$, there exists a sequence of distributions $\{\nu_n\}_{n=1}^\infty$ along a smooth submodel whose score perturbs the conditional distribution of $Y^t|A,W^t$ in such a way that the conditional average treatment effect function of $\nu_n$ is equal to $n^{-1/2}f(w)$. This sequence of distributions will constitute a local alternative whenever $E_{P_W}[f(W^t)]>0$. There are also local alternatives that perturb the covariate distribution. For example, consider the case where the conditional average treatment effect $c(w)$ is not everywhere zero but is such that $E_{P_W}[c(W^t)]=0$. In this case, there are local alternatives that perturb the marginal distribution $P_W$ but do not modify the conditional average treatment effect.

We now describe conditions that we use to identify the relative efficiency $\phi_a(P)$ and $\phi_m(P)$ using the external data that are available at the trial planning stage. 
When doing this, we assume that $P \in \mathcal{M}$, where $\mathcal{M}$ is a locally nonparametric model of all distributions of $(Y^t,W^t)$, that is, a model where the tangent space at $P$ is $L^2_0(P)$. 
Let $X = (Y,W)$ be the data unit in the external dataset, which we assume consists of $n$ iid draws from some distribution. The identifiability condition that we consider imposes that the external data should accurately reflect the distribution of covariate and outcomes in future trials where treatment has no effect.
\begin{conditionA}\label{cond1identifiability}
    A random variate $X = (Y,W)$ from the external dataset has distribution $P$.
\end{conditionA}
Under this condition, the relative efficiencies in \eqref{refully} and \eqref{reworking} can be estimated based on the external data.

\subsection{Estimating the relative efficiency}\label{subsec:ordinalestimation}

We now consider estimating relative efficiency for certain treatment effect estimands that are of particular interest in many clinical trials. We focus primarily on continuous and ordinal outcomes in this section. For the examples we consider, the asymptotic variances of the adjusted and unadjusted estimators factorize into a product of two terms, one that depends on $\Pi$ only and another that depends on $P$ only. Moreover, the term that depends only on $\Pi$ is the same for both the adjusted and unadjusted estimators, and the relative efficiency is a function of $P$ only. In particular, \eqref{refully} and \eqref{reworking} now take the following forms,
\begin{equation}\label{master}
    \phi_{a}(P) = \sigma_a^2(P)/\sigma_u^2(P), \ \phi_{m}(P) = \sigma_m^2(P)/\sigma_u^2(P).
\end{equation}
The exact forms of $\sigma_a, \sigma_m$ and $\sigma_u$ depend on the treatment effect estimand and the specified estimators but do not depend on $\Pi$. Some specific examples are presented in the remainder of this section.

\subsubsection{Continuous outcomes and average treatment effect}
To illustrate the idea, we start with a simple example where the outcome is continuous and we are interested in the average treatment effect, defined as $\psi = E[Y^t|A=1]- E[Y^t|A=0]$. Let $n^t$ denote the sample size of the future trial dataset.

The unadjusted estimator that we consider corresponds to the difference between the arm-specific means, namely $$\hat\psi_u = \sum_{i=1}^{n^t}A_iY_i^t/\sum_{i=1}^{n^t}A_i - \sum_{i=1}^{n^t}(1-A_i)Y_i^t/\sum_{i=1}^{n^t}(1-A_i).$$
For the fully adjusted estimator, we consider the augmented inverse probability weighted (AIPW) estimator, namely
\begin{equation}\label{AIPW}
    \hat\psi_a = \frac{1}{n^t}\sum_{i=1}^{n^t}\left[\frac{A_i\{Y_i^t - \hat r_1(W_i^t)\}}{\hat\pi(W_i^t)} -\frac{(1-A_i)\{Y^t_i - \hat r_0(W_i^t)\}}{1-\hat\pi(W_i^t)}  + \hat r_1(W_i^t)-\hat r_0(W_i^t)\right],
\end{equation}
where $\hat r_a(w)$ is an estimator of the conditional mean function $r_a(w) := E[Y^t|A=a,W^t=w]$ and $\hat\pi$ is an estimator of the treatment mechanism $\pi(w):= P(A=1|W^t=w)$. In randomized trials, the treatment mechanism can be estimated with $\hat\pi$, the empirical marginal of $A$, which is $\sqrt{n}$-consistent, and the AIPW estimator is efficient provided that $\hat r_a$ is consistent and satisfies appropriate conditions. The estimator will be consistent and asymptotically normal even if $\hat r_a$ is inconsistent but has an appropriately defined limit.

For the working-model-based adjusted estimator, we consider linear models, which are commonly used by practitioners for continuous outcomes \citep{FDA2019}. Specifically, we fit an arm-specific linear model for the outcome regression, which assumes that
$$E[Y^t|A=a, W^t=w] = \alpha_a + \beta_a^\top w,$$ and denote by $\hat\alpha_a, \hat\beta_a$ the fitted coefficients. To estimate the average treatment effect, we marginalize the fitted values over all covariates and take the difference between treatment arms,
$$\hat\psi_m = \hat\alpha_1 -\hat\alpha_0 + (\hat\beta_1 - \hat\beta_0)^\top \sum_{i=1}^{n^t}W_i^t/n^t.$$ 
We note again that the consistency and asymptotic normality of this estimator does not rely on the arm-specific linear models being correct.

The following lemma gives the forms of the relevant variances in the definition of relative efficiency in this setting.
\begin{lemma}\label{reATE}
Let $(Y,W) \sim P$. Suppose that the appropriate regularity conditions hold such that the AIPW estimator $\hat\phi_a$ is efficient. Then, for the above $\hat\psi_u, \hat\psi_a$ and $\hat\psi_m$, we have that $\sigma_u^2(P) = \textnormal{var}_P(Y), \sigma_a^2(P) = {E}_P[\textnormal{var}_P(Y|W)]$ and $\sigma^2_m = {E}_P\left[(Y-\alpha^*-W^\top \beta^*)^2\right]$, where $(\alpha^*,\beta^*)$ is the minimizer of ${E}_P[(Y-\alpha-W^\top \beta)^2]$ over $(\alpha,\beta)\in \mathbb{R}\times \mathbb{R}^d$. 
\end{lemma}

We now estimate these variances using external data $\{(Y_i,W_i),i=1,\ldots,n\}$. Specifically, we use the sample variance for the unadjusted variance, $\hat\sigma^2_u = \sum_{i=1}^{n}(Y_i - \bar{Y})^2/n$, where $\bar{Y}$ is the overall mean of the outcome. Let $\hat r(w)$ be an estimator of $r(w) := E_P[Y|W=w]$, then we estimate the adjusted variance by $\hat\sigma^2_a = \sum_{i=1}^{n}\{Y_i -\hat r(W_i)\}^2/n$. Finally, a natural plug-in estimator of $\sigma_m^2$ is
$\hat\sigma_m^2 = \sum_{i=1}^{n}(Y_i-\hat\alpha-W_i^\top \hat\beta )^2/n$, where $(\hat\alpha,\hat\beta)$ are the coefficients in the linear regression of $Y$ on $W$. We estimate the relative efficiencies by $\hat\phi_a = \hat\sigma^2_a/\hat\sigma^2_u$ and $\hat\phi_m = \hat\sigma^2_m/\hat\sigma^2_u$.

For a generic $\mathcal{W}\rightarrow\mathbb{R}$ function $f$, define its (squared) $L^2(P_W)$ norm as $\|f\|_{L^2(P_W)}^2 := \int f(w)^2dP_W(w)$. The following theorem establishes the asymptotic linearity of $\hat\phi_a$ and $\hat\phi_m$ under appropriate conditions.
\begin{theorem}\label{ATE}
Suppose that Conditions~\ref{cond2bounded} and \ref{cond1identifiability} hold. Suppose, in addition, that the random function $\hat{r} : \mathcal{W}\rightarrow\mathcal{Y}$ is such that $\|\hat r - r\|_{L^2(P_W)} = o_P(n^{-1/4})$ and belongs to some fixed $P$-Donsker class $\mathcal{F}$ of functions with probability tending to one. Then, $\hat\phi_a$ is an efficient estimator of $\phi_a$ and $\hat\phi_m$ is an efficient estimator of $\phi_m$.
\end{theorem}

\subsubsection{Ordinal outcomes}
Now suppose that the outcome is ordinal and takes value in $\{1,2,\ldots,K\}$. Dichotomous outcomes correspond to the special case where $K=2$. Let $F_a(k) := P(Y^t \leq k|A=a)$ denote the treatment-specific CDF. The treatment effect estimands we consider can all be written as $\psi = g\left(\{F_0(k),F_1(k)\}_{k=1}^{K-1}\right)$ for some real-valued function $g$. 

The proportional odds model \citep{mccullagh1980regression} is a commonly used parametric model for ordinal outcomes. Here we use a treatment-specific proportional odds model as our working parametric model. For $k \in \{1,\ldots,K-1\}$, the model assumes that $P(Y^t \leq k| A=a, W^t=w) =\theta_{\alpha_a,\beta_a}(k,w)$, where
\begin{equation*}
    \text{logit }\theta_{\alpha_a,\beta_a}(k,w) = \alpha_a(k) + \beta_a^\top  w.
\end{equation*}
The above reduces to a logistic regression when the outcome is dichotomous. Let $(\hat\alpha_a,\hat\beta_a)$ be the coefficients, fitted by minimizing the following empirical risk function: 
\begin{multline}\label{propoddsloss}
    L_{n,a}(\alpha,\beta) = -\sum_{k=1}^{K-1}\sum_{i=1}^{n^t} I\{A_i = a\}\Big[I\{Y_i^t\leq k\}\log\left\{\theta_{\alpha,\beta}(k,W_i^t)\right\} \\
    + I\{Y_i^t > k\}\log\left\{1-\theta_{\alpha,\beta}(k,W_i^t)\right\}\Big],
\end{multline}
where $I\{\cdot\}$ is the indicator function. In the special case that $\sum_{i=1}^{n^t}I\{A_i=a,Y_i^t\leq k\}=0$ for some $k$, we let $\hat\alpha_a(k) = -\infty$. Similarly, in the case that $\sum_{i=1}^{n^t}I\{A_i=a,Y_i^t> k\}=0$ for some $k$, we let $\hat\alpha_a(k) = \infty$. For such cases we use the conventions that $\text{logit}^{-1}(-\infty)=0$, $\text{logit}^{-1}(\infty)=1$, and $0\log(0)=0$. 
The treatment-arm-specific CDFs are estimated by $\hat F_a(k) = \sum_{i=1}^{n^t}\theta_{\hat\alpha_a,\hat\beta_a}(k,W_i^t)/n^t$. In addition, we define $(\alpha_a^*,\beta_a^*)$ as the minimizer of $E[L_{n,a}(\alpha,\beta)]$ over $\mathbb{R}^{K-1} \times \mathbb{R}^d$.

We first establish the RAL property of $\hat F_a(k)$, which holds even when the proportional odds model is misspecified. Let $\theta_a(k,w) := P(Y^t \leq k|A=a,W^t=w)$ be the true outcome regression, and let $\theta^*_a(k,w) = \theta_{\alpha^*_a,\beta^*_a}(k,w)$ be the best model approximation to the true outcome regression according to the population analogue of the risk in \eqref{propoddsloss}. 
Note that $\theta^*_a(k,w)$ can be different from $\theta_a(k,w)$ in the presence of misspecification. 

\begin{lemma}\label{ordinalworking}
Suppose that Condition~\ref{cond2bounded} holds and that $(\hat\alpha_a,\hat\beta_a)$ is estimated by minimizing \eqref{propoddsloss}, then $\hat F_a(k)$ is an asymptotically linear estimator of $F_a(k)$, for $k \in \{1,\ldots,K-1\}$. Its influence function is given by
\begin{equation*}
    \textnormal{IF}_{F_a(k)}(y^t,\tilde{a},w^t) = \frac{I\{\tilde{a}=a\}}{\Pi(A=a)}\left[I\{y^t \leq k\}-\theta^*_a(k,w^t)\right]+\theta^*_a(k,w^t)-F_a(k).
\end{equation*}
\end{lemma}

Following \cite{benkeser2020improving}, we focus on three treatment effect estimands that are often of interest.
\\
\\
\noindent\emph{Difference in mean (DIM)} is
defined as $\psi = E\left[u(Y^t)|A=1\right] - E\left[u(Y^t)|A=0\right]$ for a pre-specified monotone transformation $u(\cdot)$. This reduces to the average treatment effect when $u(\cdot)$ is the identity function. The unadjusted estimator is the difference between the arm-specific sample means, $$\hat\psi_u = \sum_{i=1}^{n^t}u(Y_i^t)A_i\Big/\sum_{i=1}^{n^t}A_i-\sum_{i=1}^{n^t}u(Y_i^t)(1-A_i)\Big/\sum_{i=1}^{n^t}(1-A_i).$$ Instead of using sample means, the adjusted estimator based on proportional odds model computes means with respect to the estimated CDFs $\hat F_a$.
\begin{equation*}
    \hat\psi_{m} = \sum_{k=1}^{K-1}\{u(k)-u(k+1)\}\{\hat F_1(k)-\hat F_0(k)\}.
\end{equation*}
Finally, we define an AIPW estimator similarly to \eqref{AIPW}, but with $Y^t$ replaced by $u(Y^t)$. We denote this estimator as $\hat\psi_a$. As in the previous section, this estimator achieves the semiparametric efficiency bound when the treatment mechanism is estimated with the marginal proportion of treatment and the outcome regression is consistently estimated. For the above three estimators, the variances in \eqref{master} are given in the following lemma.

\begin{lemma}\label{reDIM}
Let $(Y,W) \sim P$. Suppose that the appropriate regularity conditions hold such that the AIPW estimator $\hat\phi_a$ is efficient. Then, for the above $\hat\psi_u, \hat\psi_a$ and $\hat\psi_m$, we have that $\sigma_u^2(P) = \textnormal{var}_P[u(Y)]$, $\sigma_a^2(P) = {E}_P\left[\textnormal{var}_P(u(Y)|W)\right]$, and
\begin{equation*}
    \sigma^2_{m}(P) = {E}_P\left[\left(\sum_{k=1}^{K-1}\{u(k)-u(k+1)\}\left[I\{Y\leq k\}-\theta^*(k,W)\right]\right)^2\right],
\end{equation*}
where $\theta^*(k,w) = \theta_{\alpha^*,\beta^*}(k,w)$ and $(\alpha^*, \beta^*)$ maximizes the following objective:
\begin{equation}\label{propoddsEE}
    E_P\left[\sum_{k=1}^{K-1}I\{Y\leq k\}\log\left\{\theta_{\alpha,\beta}(k,W)\right\} + I\{Y > k\}\log\left\{1-\theta_{\alpha,\beta}(k,W)\right\}\right].
\end{equation}
\end{lemma}

\noindent We now propose estimators of these quantities for settings where external data are available. Let $\hat r(w)$ be an estimator of $r(w) := E_P[u(Y)|W=w]$, the conditional mean of $u(Y)$, and let $\bar{u}_n$ be the sample mean of $u(Y)$. We estimate the unconditional and conditional variances by
\begin{equation*}
    \hat{\sigma}_u^2 = \frac{1}{n}\sum_{i=1}^{n}\{u(Y_i)-\bar{u}_n\}^2, \ \hat{\sigma}_a^2 = \frac{1}{n}\sum_{i=1}^{n}\{u(Y_i)-\hat r(W_i)\}^2.
\end{equation*}
To estimate $\sigma_m^2$, we fit the proportional odds model by maximizing the empirical counterpart of \eqref{propoddsEE}, and let $(\hat\alpha,\hat\beta)$ denote the fitted coefficients. 
We then construct a plug-in estimator
\begin{equation*}
    \hat\sigma_{m}^2 = \frac{1}{n}\sum_{i=1}^{n}\left(\sum_{k=1}^{K-1}\{u(k)-u(k+1)\}\left[I\{Y_i\leq k\}-\theta_{\hat\alpha,\hat\beta}(k,W_i)\right]\right)^2.
\end{equation*}
Finally, we estimate the relative efficiency by $\hat\phi_a = \hat\sigma_a^2/\hat\sigma_u^2$ and $\hat\phi_{m} = \hat\sigma_{m}^2/\hat\sigma_u^2$. 

In the upcoming theorem, we let $u(\mathcal{Y})$ denote the convex hull of $\{u(y) : y\in\mathcal{Y}\}$.
\begin{theorem}\label{DIM}
Suppose that Conditions~\ref{cond2bounded} and \ref{cond1identifiability} hold. Suppose, in addition, that the random function $\hat{r} : \mathcal{W}\rightarrow u(\mathcal{Y})$ is such that $\|\hat r - r\|_{L^2(P_W)} = o_P(n^{-1/4})$ and belongs to some fixed $P$-Donsker class $\mathcal{F}$ of functions with probability tending to one. Then $\hat\phi_a$ is an efficient estimator of $\phi_a$. Moreover, $\hat\phi_m$ is an efficient estimator of $\phi_m$.
\end{theorem}
Exact forms of the influence functions of $\hat\phi_a$ and $\hat\phi_m$ are given in Appendix~\ref{ordinalproof}.
\\
\\
\noindent \emph{The Mann-Whitney estimand (MW)} is defined as $\psi = P(Y_1^t>\tilde Y_0^t) + P(Y_1^t = \tilde Y_0^t)/2$, for two independent variables $Y_1^t \sim P_1$ and $\tilde Y_0^t \sim P_0$.
It is the probability that a randomly chosen individual's outcome under treatment is larger than another randomly chosen individual's outcome under control plus one half times the probability that the two outcomes are equal. Define $h(x,y) = I\{x>y\} + I\{x=y\}/2$. Then the Mann-Whitney parameter can be alternatively written as
$\psi = \int\int h(x,y)dP_1(x)dP_0(y).$ This alternative definition suggests the following unadjusted estimator
$$\hat\psi_u = \sum_{i=1}^{n^t}\sum_{j=1}^{n^t}A_i(1-A_j)h(Y_i^t,Y_j^t)\Big/\left\{\left(\sum_{i=1}^{n^t}A_i\right)\left(n^t-\sum_{j=1}^{n^t}A_j\right)\right\},$$
and the following working-model-based estimator
$$\hat\psi_{m} = \int\int h(x,y)d\hat{P}_1(x)d\hat{P}_0(y),$$ where $\hat{P}_a$ is the distribution with CDF $\hat F_a(k)$. In addition, let $\hat\psi_a$ be the covariate-adjusted estimator in \citet{vermeulen2015increasing} for the MW parameter, which is efficient under appropriate regularity conditions. 

\begin{lemma}\label{reMW}
Let $(Y,W) \sim P$, and define $\eta_P(k) = P(Y<k) + P(Y=k)/2$ and $p_k = P(Y=k)$. Suppose that the appropriate regularity conditions hold such that $\hat\phi_a$ is an efficient estimator of the MW parameter. Then, for the above $\hat\psi_u, \hat\psi_a$ and $\hat\psi_m$, we have that
$\sigma_u^2 = \textnormal{var}_P[\eta_P(Y)] = (1-\sum_{k=1}^{K}p_k^3)/12$; $\sigma_a^2 = E_P\left[\textnormal{var}_P(\eta_P(Y)|W)\right]$; and
\begin{equation*}
    \sigma^2_{m}(P) =  E_P\left[\left(\sum_{k=1}^{K-1}\{\eta_P(k)-\eta_P(k+1)\}\left[I\{Y\leq k\}-\theta^*(k,W)\right]\right)^2\right].
\end{equation*}
\end{lemma}

We now propose estimators for these quantities. Unlike in the case of the DIM estimand, $\eta_P(\cdot)$ depends on the unknown marginal distribution of $Y$ and needs to be estimated from the external data. A simple estimator is based on the empirical distribution, $\hat p_k = \sum_{i=1}^{n}I\{Y_i = k\}/n$ and $\hat \eta(k) = \sum_{j=1}^{k}\hat p_j - \hat p_k/2$. The unadjusted variance can be estimated via the plug-in estimator $\hat\sigma_u^2 = (1-\sum_{k=1}^{K}\hat p_k^3)/12$. Let $\hat r(w)$ be an estimator of the conditional mean $r(w) := E_P[\eta_P(Y)|W=w]$, and we estimate the adjusted variance by $\hat\sigma_a^2 = \sum_{i=1}^{n}\{\hat \eta(Y_i)-\hat r(W_i)\}^2/n$. Finally, a natural plug-in estimator for $\hat\sigma_m^2$ is given by 
\begin{align*}
    \hat\sigma^2_{m} &= \sum_{i=1}^{n}\left(\sum_{k=1}^{K-1}\left\{[\hat\eta(k)-\hat\eta(k+1)]\left[I\{Y_i\leq k\}-\theta_{\hat\alpha,\hat\beta}(k,W_i)\right]\right\}\right)^2/n,
\end{align*}
where $(\hat\alpha,\hat\beta)$ is again the fitted coefficients from the proportional odds model, by maximizing the sample counterpart of \eqref{propoddsEE}. The relative efficiency can be estimated as $\hat\phi_a= \hat\sigma_a^2/\hat\sigma_u^2$ and $\hat\phi_m= \hat\sigma_m^2/\hat\sigma_u^2$.

The next theorem establishes the asymptotic properties of these estimators.

\begin{theorem}\label{MW}
Suppose that Conditions~\ref{cond2bounded} and \ref{cond1identifiability} hold. Suppose, in addition, that the random function $\hat{r} : \mathcal{W}\rightarrow \mathbb{R}$ is such that $\|\hat r - r\|_{L^2(P_W)} = o_P(n^{-1/4})$ and belongs to some fixed $P$-Donsker class $\mathcal{F}$ of functions with probability tending to one. Then, $\hat\phi_a$ is an efficient estimator of $\phi_a$ and $\hat\phi_m$ is an efficient estimator of $\phi_m$.
\end{theorem}
The influence functions of $\hat\psi_a$ and $\hat\psi_{m}$ are given in Appendix~\ref{ordinalproof}.
\\
\\
\noindent \emph{The log odds ratio (LOR)} is defined as $\psi = \sum_{k=1}^{K-1} \{\text{logit }F_1(k)-\text{logit }F_0(k)\}/(K-1)$, which is an average of the cumulative log odds ratios \citep{diaz2016enhanced}. In general, one can also consider a weighted average. Employing this definition for the LOR ensures that the LOR is well-defined even in settings where a proportional odds assumption fails.

The unadjusted estimator is given by
$$\hat\psi_u = \frac{1}{K-1}\sum_{k=1}^{K-1} \left\{\text{logit }\tilde F_1(k)-\text{logit }\tilde F_0(k)\right\},\textnormal{ where }   \tilde F_a(k) = \frac{\sum_{i=1}^{n^t}I\{Y_i^t \leq k,A_i=a\}}{\sum_{i=1}^{n^t}I\{A_i=a\}}.$$
The working-model-based adjusted estimator $\hat\psi_{m}$ replaces $\tilde F_a(k)$ with the proportional odds model-based estimator $\hat F_a(k)$ that was defined earlier. Finally, let $\hat\psi_a$ be the covariate-adjusted estimator proposed in \citet{diaz2016enhanced}, which achieves the semiparametric efficiency bound under regularity conditions. The following lemma gives the forms of the relevant variances.

\begin{lemma}\label{reLOR}
Let $(Y,W) \sim P$, $F(k) := P(Y\leq k)$, and
\begin{equation*}
    \zeta(Y) := \frac{1}{K-1}\sum_{k=1}^{K-1} \frac{I\{Y\leq k\}}{F(k)\{1-F(k)\}}.
\end{equation*}
Suppose that the appropriate regularity conditions hold such that $\hat\phi_a$ is an efficient estimator of the LOR. Then, for the above $\hat\psi_u$, $\hat\psi_a$ and $\hat\psi_m$, we have that
\begin{align*}
    \sigma_u^2(P) &= \textnormal{var}_P[\zeta(Y)], \ \sigma_a^2(P) = E_P\left[\textnormal{var}_P\left(\zeta(Y)|W\right)\right],  \\
    \sigma_{m}^2(P) &= E_P\left[\left(\zeta(Y) - \frac{1}{K-1}\sum_{k=1}^{K-1}\frac{\theta^*(k,W)}{F(k)\{1-F(k)\}}\right)^2\right].
\end{align*}
\end{lemma}

Let $\hat\theta(k,w)$ be an estimator of the true conditional distribution function $\theta(k,w) := P(Y\leq k |W=w)$ in the setting where external data are available. We estimate the relative efficiencies by $\hat\phi_a = \hat\sigma^2_a/\hat\sigma^2_u$ and $\hat\phi_m = \hat\sigma^2_m/\hat\sigma^2_u$, with the variance estimators all taking the following form with certain choice of the estimator $\hat\theta_c$:
\begin{equation*}
      \frac{1}{n}\sum_{i=1}^{n}\left[\sum_{k=1}^{K-1}\frac{I\{Y_i\leq k\}-\hat \theta_c(k,W_i)}{\tilde F(k)\{1-\tilde F(k)\}}\right]^2,
\end{equation*}
where $\tilde F(k) := \sum_{i=1}^{n}I\{Y_i \leq k\}/n$. Specifically, for $\hat\sigma^2_a$, $\hat\theta_c(k,w)$ is replaced with $\hat\theta(k,w)$; for $\hat\sigma_u^2$, we use $\tilde F(k)$; and, for $\hat\sigma^2_m$, we take $\hat\theta_c(k,w) =  \theta_{\hat\alpha,\hat\beta}(k,w)$.

\begin{theorem}\label{LOR}
Suppose that Conditions~\ref{cond2bounded} and \ref{cond1identifiability} hold and that there exists a constant $\delta>0$ such that $\delta<F(k)<1-\delta$ for all $k \in \{1,\ldots,K-1\}$. Suppose, in addition, that, for all $k \in \{1,\ldots,K-1\}$, the random function $\hat{\theta}(k,\cdot) : \mathcal{W}\rightarrow \mathbb{R}$ is such that $\|\hat\theta(k,\cdot) - \theta(k,\cdot)\|_{L^2(P_W)} = o_P(n^{-1/4})$ and belongs to some fixed $P$-Donsker class $\mathcal{F}$ of functions with probability tending to one.
Then, $\hat\phi_a$ is an efficient estimator of $\phi_a$. Moreover, $\hat\phi_m$ is an efficient estimator of $\phi_m$.
\end{theorem}

For all of the aforementioned treatment effect estimands, $\phi_a \in (0,1]$, since the fully adjusted estimator achieves the semiparametric efficiency bound. However, $\phi_{m}$ might be larger than 1 if the proportional odds model is far from the truth. 

Wald-type intervals are a standard approach for constructing confidence intervals when an asymptotically linear estimator $\hat\phi$ of $\phi$ is available. Specifically, a $(1-\alpha)$-CI is given by $\hat\phi \pm n^{-1/2}z_{1-\alpha/2}\mathbb{P}_n\hat\tau^2$, where $z_{1-\alpha/2}$ is the $(1-\alpha/2)$-quantile of a standard normal distribution and $\hat\tau$ is the influence function of $\hat\phi$ except that we replace unknown quantities with consistent estimates. However, there are certain cases where, even though such consistent estimates are used, the Wald-type confidence interval will not provide asymptotically valid coverage. A key time when this challenge arises occurs when the limiting distribution of a RAL estimator is degenerate because the influence function is almost everywhere zero, which will often occur in our setting when the relative efficiency is one. One such example arises in estimating the relative efficiency of the fully adjusted estimator to the unadjusted estimator for the ATE or DIM estimands. In this case, the influence function of $\hat\phi_a$ is almost everywhere zero when $E[u(Y)|W=w] = E[u(Y)]$ for almost all $w$.
While a Wald-type interval will typically achieve asymptotically valid coverage outside of these degenerate cases, there is no way to know in advance whether or not $P$ is such that degeneracy will occur. 
To overcome this challenge, we propose an alternative approach that yields a confidence set that achieves the desired coverage regardless of whether degeneracy occurs. Most importantly, the resulting confidence sets are valid regardless of whether or not the true relative efficiency is, in fact, one.

The proposed confidence set is constructed as follows. Suppose that we have available a valid level $\alpha$ test of the null hypothesis $H_0: \phi_a = 1$ ($\phi_m=1$) --- one such test based on sample splitting is given in Appendix~\ref{unionset}. Denote the Wald confidence interval by $I_{wald}$. We first test this null hypothesis. If we do not reject it, we take the $1-\alpha$ confidence set to be $I_{wald} \cup \{1\}$. If, instead, the null hypothesis is rejected, we take the confidence set to be $I_{wald}$. At first glance, it seems that the proposed approach may fail to achieve valid coverage given that it uses the same data twice --- once to test the null hypothesis and a second time to form the Wald-type interval. Nevertheless, as we show in Appendix~\ref{unionset}, the confidence set resulting from this procedure in fact has at least $1-\alpha$ asymptotic probability of covering the truth. It may happen that $I_{wald}$ and $\{1\}$ are disjoint. In this case, a disconnected confidence set $I_{wald}\cup \{1\}$ can be reported or, if this is considered undesirable, the convex hull of this confidence set can be taken to form a confidence interval.

\subsection{Bootstrap procedure for working-model-based estimators}\label{subsec:bootstrap}
The inferential procedures described in the preceding subsection are based on closed-form expressions for the relative efficiency parameter in several problems and knowledge of the corresponding efficient influence functions. On the one hand, now that these expressions have been calculated, the estimators that we have presented can be used in any problems in which the relative efficiency takes this form. On the other hand, if a new effect estimand or working parametric model is of interest in a future setting, then new analytical calculations will need to be conducted to derive the closed-form expression for the relative efficiency parameter and develop corresponding estimators and confidence intervals. Here, we propose an automated double bootstrap procedure that avoids the need to perform these potentially-tedious analytic calculations. When doing so, we focus on the case where the goal is to infer about the relative efficiency of a new working-model-based adjusted estimator, that is, we focus on $\phi_m$. The reason for this choice is discussed at the end of this subsection.

Before describing this procedure, we first investigate the applicability of a more traditional, one-layer bootstrap procedure. Suppose that the relative efficiency parameter $\Phi_m: \mathcal{M} \rightarrow \mathbb{R}^{+}$ is sufficiently smooth so that a plug-in estimator of $\Phi_m(P)$ based on the empirical distribution is asymptotically linear \citep[see, e.g., Theorem~20.8 in][]{van2000asymptotic}. 
In this case, we can construct a plug-in estimator based on the empirical distribution $\mathbb{P}_n$ of a sample of iid observations. We denote this plug-in estimator by $\Phi_m(\mathbb{P}_n)$ and note that all the estimators proposed so far for $\phi_m$ correspond to plug-in estimators of this form. In a traditional setting where the bootstrap would be applied, a closed-form expression for the functional $\Phi_m$ would be available and so would be the plug-in estimator $\Phi_m(\mathbb{P}_n)$, and the goal would be to derive a corresponding confidence interval. In particular, let the $n$ entries of $\boldsymbol{X}=(X_k)_{k=1}^n$ correspond to an iid sample of external data, where $X_k=(Y_k,W_k)$. We sample from $\boldsymbol{X}$ with replacement $B_1$ times, to form the bootstrap resamples $\boldsymbol{X}^*_i$ of size $n$ for $i=1,\ldots,B_1$. Letting $\mathbb{P}_{n,i}^*$ denote the empirical distribution of the observations in $\boldsymbol{X}_i^*$, we could then use the empirical standard deviation of $\Phi_m(\mathbb{P}_{n,i}^*)$, $i=1,\ldots,B_1$, as the standard error estimate used to construct a confidence interval centered around $\Phi_m(\mathbb{P}_n)$. Though this traditional bootstrap approach is useful in that it avoids explicitly computing the influence function of $\Phi_m(\mathbb{P}_n)$, it does not fully avoid the aforementioned analytic calculations. Indeed, in many cases, deriving the plug-in estimator will itself require deriving the explicit form of the relative efficiency parameter, which in turn relies on computing inefficient and efficient gradients of the treatment effect estimand in the model $\mathcal{M}_X$. Computing these gradients requires specialized calculations that are unfamiliar to many practitioners.

To avoid this challenge, we approximate the plug-in estimator with an alternative estimator $\tilde\phi$ that can be obtained in a fully automated fashion. Specifically, we propose to use an additional layer of resampling to approximate $\phi_m$. Evaluating the resulting estimation strategy only requires having access to the external data and the treatment effect estimator that will be used to analyze data from the future clinical trial. 

Again let $\{\boldsymbol{X}_i^*\}_{i=1}^{B_1}$ be the first layer resamples, and in addition define $\boldsymbol{X}_0^* := \boldsymbol{X}$. For each $i\ge 0$, we then let $\tilde{\boldsymbol{X}}_{ij}$, $j=1,\ldots,B_2$, denote an iid sample of size $N$ from the product measure $\mathbb{P}_{n,i}^*\Pi$, where $\Pi$ is the known distribution of treatment. 
To simulate $\tilde{\boldsymbol{X}}_{ij}$, we first draw an iid sample of size $N$ from the empirical distribution $\mathbb{P}_{n,i}^*$ and then append a random draw of the treatment vector, which is a tuple consisting of $N$ iid draws from a Bernoulli($\pi$) distribution.

For each $\boldsymbol{X}_i^*$, we will construct an estimator $\tilde\phi(\boldsymbol{X}_i^*)$ using the collection of second-layer resamples. Specifically, for each $(i,j)$, we compute the adjusted and unadjusted estimators based on the sample $\tilde{\boldsymbol{X}}_{ij}$, which we denote as $\hat\psi_m^{ij}$ and $\hat\psi_u^{ij}$, respectively. Define $\Bar\psi_m^i = \sum_{j=1}^{B_2} \hat\psi_m^{ij}/B_2$ and $\Bar\psi_u^i = \sum_{j=1}^{B_2} \hat\psi_u^{ij}/B_2$. A stochastic approximation of the parameter evaluation $\Phi_m(\mathbb{P}_{n,i}^*)$ is then given by
\begin{equation}
    \tilde\phi_{n}(\boldsymbol{X}_i^*) = \sum_{j=1}^{B_2}(\hat\psi_m^{ij} - \Bar\psi_m^i)^2 \Big/ \sum_{j=1}^{B_2}(\hat\psi_u^{ij} - \Bar\psi_u^i)^2.
\end{equation}
We note that $\tilde\phi_n(\boldsymbol{X}_i^*)$ depends on $\{\tilde{\boldsymbol{X}}_{ij}\}_{j=1}^{B_2}$, and therefore also on $B_2$ and $N$ --- we omit these dependencies in the notation. A Wald-type bootstrap confidence interval can be constructed by using the empirical standard deviation of $\tilde\phi_n(\boldsymbol{X}_i^*)$ over $i=1,\ldots,B_1$ as the standard error and $\tilde\phi_n(\boldsymbol{X}) := \tilde\phi_n(\boldsymbol{X}^*_0)$ as the center. This double bootstrap procedure is summarized in Algorithm \ref{algbootstrap} in Appendix~\ref{app:pseudocode}.

We now provide some intuition behind why the above-described double bootstrap procedure is expected to work. We then provide a theorem that formalizes these arguments. First, we observe that the double bootstrap procedure is analogous to a traditional single-layer bootstrap, except that we replace the plug-in estimator with an estimator $\tilde\phi_n$, which itself is defined through an additional layer of bootstrap. Intuitively, if $\tilde\phi_n(\boldsymbol{X}_i^*)$ is close enough to the plug-in estimator $\Phi_m(\mathbb{P}_{n,i}^*)$ on all $\boldsymbol{X}_i^*$, we would expect that using the stochastic approximation instead of the plug-in makes little difference and the procedure works similarly as the traditional bootstrap works. We now give heuristic arguments showing that $\tilde\phi_n(\boldsymbol{X}_i^*)$ and $\Phi_m(\mathbb{P}_{n,i}^*)$ should indeed be close, in the sense that 
\begin{align*}
    \tilde{\phi}_n(\boldsymbol{X}_i^*)-\Phi_m(\mathbb{P}_{n,i}^*)=o_{\mathbb{P}_{n,i}^*}(n^{-1/2}).
\end{align*}
To see this, for an arbitrary $j\in\{1,\ldots,B_2\}$, consider a general asymptotically linear estimator $\hat\psi$ of the treatment effect that satisfies
\begin{equation}\label{ALexpansion}
    \hat\psi(\tilde{\boldsymbol{X}}_{ij})-\psi(\mathbb{P}_{n,i}^*\Pi)=\frac{1}{N}\sum_{l=1}^{N}D(\mathbb{P}_{n,i}^*\Pi)(\tilde{\boldsymbol{X}}_{ij}^{l})+ \text{Rem}_i,
\end{equation}
where $\tilde{\boldsymbol{X}}_{ij}^{l}$ is the $l$-th observation in $\tilde{\boldsymbol{X}}_{ij}$. In our upcoming theorem, we will assume that $\text{Rem}$ is negligible in an appropriate sense. 
Suppose that we take sufficiently many samples from $\mathbb{P}_{n,i}^*\Pi$ --- that is, that $B_2$ is sufficiently large --- so that the Monte-Carlo error from the second bootstrap layer is negligible. We can then accurately approximate the sampling distribution of $\sqrt{N}\{\hat\psi(\tilde{\boldsymbol{X}}_{ij})-\psi(\mathbb{P}_{n,i}^*\Pi)\}$ under $\mathbb{P}_{n,i}^*$ by the empirical distribution of $\{\hat\psi(\tilde{\boldsymbol{X}}_{ij})\}_{j=1}^{B_2}$. Applying these arguments at $\hat\psi_u$ and $\hat\psi_m$ suggests that $\tilde\phi_n(\boldsymbol{X}_i^*)$ accurately approximates $\tilde\sigma_{m,i}^2/\tilde\sigma_{u,i}^2$, where $\tilde \sigma^2_{u,i} = \textnormal{var}_{\mathbb{P}_{n,i}^*}[\sqrt{N}\hat\psi_u(\tilde{\boldsymbol{X}}_{ij})]$ and $\tilde \sigma^2_{m,i} = \textnormal{var}_{\mathbb{P}_{n,i}^*}[\sqrt{N}\hat\psi_m(\tilde{\boldsymbol{X}}_{ij})]$ are the variances of the sampling distributions where $\tilde{\boldsymbol{X}}_{ij}$ is an iid sample from $\mathbb{P}_{n,i}^*\Pi$. In addition, provided that $N \gg n$ so that the remainders in the above linear expansion \eqref{ALexpansion} are sufficiently small when $\hat{\psi}$ is equal to $\hat{\psi}_u$ and $\hat{\psi}_m$, the ratio between these variances $\tilde\sigma_{m,i}^2/\tilde\sigma_{u,i}^2$ is approximately $E_{\mathbb{P}_{n,i}^*}[D_m^2(\mathbb{P}_{n,i}^*\Pi)]/E_{\mathbb{P}_{n,i}^*}[D_u^2(\mathbb{P}_{n,i}^*\Pi)] = \Phi_m(\mathbb{P}_{n,i}^*)$.  As a result, we expect $\tilde\phi_n(\boldsymbol{X}_i^*)$ to be reasonably close to the plug-in estimator $\Phi_m(\mathbb{P}_{n,i}^*)$.

The upcoming theorem formalizes the heuristic argument given in the previous paragraph. Before giving this result, we define a key differentiability concept that is useful for establishing theoretical guarantees for bootstrap procedures. Let $\mathbb{D}$ denote the space of c\`{a}dl\`{a}g $\mathbb{R}^{d+1}\rightarrow\mathbb{R}$ functions equipped with the uniform norm. Let $\rho$ be the operator that takes as input a CDF on $\mathbb{R}^{d+1}$ and outputs the corresponding distribution on $\mathbb{R}^{d+1}$. Also let $\mathbb{D}_{\mathcal{M}}:=\{\rho^{-1}(P) : P\in\mathcal{M}\}$, where $\rho^{-1}(P)$ denotes the CDF of $P$. In what follows, we will call a parameter $\phi : \mathcal{M}\rightarrow\mathbb{R}$ Hadamard differentiable if the composition $\phi\circ \rho : \mathbb{D}_{\mathcal{M}}\rightarrow\mathbb{R}$, defined on the subset $\mathbb{D}_{\mathcal{M}}$ of the normed space $\mathbb{D}$, is Hadamard differentiable in the sense defined in Chapter~20.2 of \cite{van2000asymptotic}.

We will assume that the following conditions hold:
\begin{conditionB}\label{Hadamardvar}
    Both $\sigma^2_u(\cdot)$ and $\sigma^2_m(\cdot)$ are Hadamard differentiable;
\end{conditionB}
\begin{conditionB}\label{Eremaindervar}
    There exists a $\gamma\in (1/2,\infty)$ such that the remainder $\textnormal{Rem}_1$ in Eq.~\ref{ALexpansion} is such that $E[\textnormal{var}_{\mathbb{P}_{n,1}^*}(N^{\gamma}\textnormal{Rem}_1)]$ is uniformly bounded in $n$, where the expectation is over the draw of the bootstrap sample $\mathbb{P}_{n,1}^*$ and $X_1,X_2,\ldots$;
\end{conditionB}
\begin{conditionB}\label{B2}
    $B_2$ grows with $n$ in such a way that $n^{1/2}\{\tilde\sigma^2_{m,1}/\tilde\sigma^2_{u,1} - \tilde\phi_n(\boldsymbol{X}_1^*)\}\overset{p}{\rightarrow} 0$ given $(X_1,X_2,\ldots)=(x_1,x_2,\ldots)$ for almost every $(x_1,x_2,\ldots)$;
\end{conditionB}
\begin{conditionB}\label{Nbiggern}
    $N\gg n^{1/(2\gamma-1)}$ in the sense that $n^{1/(2\gamma-1)}/N \rightarrow 0$ as $n\rightarrow\infty$.
\end{conditionB}

We are now ready to state the theorem. 

\begin{theorem}\label{bootstrap}
Under Conditions \ref{cond2bounded}-\ref{cond1identifiability} and Conditions \ref{Hadamardvar}-\ref{Nbiggern}, we have that $\sqrt{n}\{\tilde \phi_n(\boldsymbol{X}^*_1) - \Phi_m(\mathbb{P}_n)\}$ converges in distribution to $\Phi_m^{\prime}(\mathbb{G})$, given $X_1,X_2, \ldots$, in probability, where $\Phi_m^{\prime}$ is the G{\^a}teaux derivative of the functional $\Phi_m$ and $\mathbb{G}$ is a mean-zero Gaussian process with covariance $\textnormal{cov}(\mathbb{G}f_1,\mathbb{G}f_2) = P(f_1f_2)-Pf_1Pf_2$.
\end{theorem}

The proof is a modification of the proof of Theorem 23.9 in \citet{van2000asymptotic}, and is given in Appendix~\ref{ordinalproof}. To approximate the limiting distribution $\Phi_m^{\prime}(\mathbb{G})$ given in the above theorem, Algorithm~\ref{algbootstrap} uses the empirical distribution of $\sqrt{n}\{\tilde \phi_n(\boldsymbol{X}^*_i) - \tilde \phi_n(\boldsymbol{X})\}$ across the $B_1$ bootstrap replicates $\boldsymbol{X}_1^*,\ldots,\boldsymbol{X}_{B_1}^*$.

We now discuss the conditions of Theorem~\ref{bootstrap}. Condition~\ref{Hadamardvar} ensures the Hadamard differentiability of the relative efficiency parameter $\Phi_m(\cdot)$, which is used in most standard sets of sufficient conditions for the validity of bootstrap methods. We can establish these Hadamard differentiability conditions by noting that the variance is essentially the mean of a function indexed by nuisance parameters, which themselves are transformations of some population means. We use the Mann-Whitney estimand in the ordinal outcome case as an example. The variance of the adjusted estimator takes the form
\begin{equation*}
    E_P\left[\left(\sum_{k=1}^{K-1}[b_P(k)\{I\{Y\leq k\}-\theta^*_{\alpha_P,\beta_P}(k,W)\}]\right)^2\right].
\end{equation*}
Here $b, \alpha, \beta$ are nuisance parameters, which are defined, either explicitly or implicitly, with a set of population means. The mean functional is Hadamard differentiable, for example, when the support is bounded. One can then apply the chain rule of Hadamard differentiability as in \cite{hirose2016differentiability}. 

Condition~\ref{Eremaindervar} ensures that the remainder term in the asymptotic linear expansion is sufficiently small. For the examples we have considered, it is possible to show that we can take $\gamma = 1$ under mild conditions. Conditions~\ref{B2} and \ref{Nbiggern} require that the user selects sufficiently large values for $B_2$ and $N$. Condition~\ref{B2} places a restriction on the Monte Carlo approximation $\tilde{\phi}_n(\boldsymbol{X}_1^*)$ of $\tilde{\sigma}_{m,1}^2/\tilde{\sigma}_{u,1}^2$. In most cases, this condition will hold provided the number of second-layer bootstrap samples goes to infinity faster than does $n$, that is, so that $n/B_2\rightarrow 0$. Condition~\ref{Nbiggern} places a restriction on the sample size $N$ of each second-layer bootstrap sample. When $\gamma=1$, this condition requires that these samples be of a larger order than the original sample size $n$. Taken together, Conditions~\ref{B2} and \ref{Nbiggern} impose that sufficient computing power must be available to compute the estimator $\hat{\psi}$ approximately $B_1 B_2$ times on samples of size $N$ --- in contrast, the analytic method in the previous section only required fitting the estimator $\hat\phi_m$ (and estimating its standard error) once on a sample of size $n$.

We conclude by noting that we can define a double bootstrap procedure analogous to Algorithm~\ref{algbootstrap} for the estimation of the relative efficiency of a fully adjusted estimator $\phi_a$. However, our arguments cannot generally be used to establish the validity of double bootstrap confidence intervals for $\phi_a$. The problem arises because the asymptotic variance of the fully adjusted estimator often involves a regression function of the outcome against the covariate. Because the statistical model is nonparametric up to knowledge of the treatment probability, this dependence will often make it so that the parameter $\sigma_a^2(\cdot)$ is not Hadamard differentiable, and so the theoretical guarantee presented above for our double bootstrap procedure may not apply. It is therefore an open question as to whether the double bootstrap will yield valid confidence intervals for the relative efficiency of fully adjusted estimators.

\section{When the outcome is partially observed}\label{sec:partiallyobserved}
We now consider settings where the outcome in the trial is only partially observed. For this purpose, we use the notion of coarsening-at-random \citep{gill1997coarsening, heitjan1991ignorability}. 
Let $Z^t = (T^t,A,W^t)$ be the full data unit in the trial, $C^t$ be a coarsening variable, and $X^t = G_a(Z^t,C^t)$ be the observation unit in the trial where $G_a(\cdot,\cdot)$ is some many-to-one function. We further assume that, under $G_a$, the covariate $W^t$ is fully observed. The adjusted analysis estimates the treatment effect $\psi$ based on $X^t$. We can write $X^t$ as $(\tilde X^t, W^t)$, where $\tilde X^t$ represents the components in $X^t$ that are not covariates and $W^t$ is the covariate vector. Define a function $c(\cdot)$ such that $c(X^t) := \tilde X^t$, and write $G_u$ to denote the composition $c\circ G_a$. The unadjusted analysis ignores the covariate information, which is equivalent to working with the observation unit $X^t_u = G_u(Z^t,C^t)$ rather than with $X^t$. The relative efficiency, defined in terms of the variances of the unadjusted and adjusted estimators, is interesting only when both analyses give consistent estimators. Thus, we will assume that both conditional distributions $G_a(Z^t,C^t)|Z^t$ and $G_u(Z^t,C^t)|Z^t$ satisfy the coarsening-at-random assumption, so that both the unadjusted and adjusted analyses are asymptotically unbiased for the treatment effect.

Let $\nu$ denote the distribution of $X^t$. We again define the relative efficiencies by focusing on trials under the sharp null, that is, the conditional distributions $T^t|A=0, W^t$ and $T^t|A=1,W^t$ are the same. We let $G(\cdot|a,w)$ denote the conditional distribution of $C^t$ given that $(A,W^t)=(a,w)$. Under the sharp null, $\nu$ is fully characterized by the treatment distribution $\Pi$, the conditional distribution of $C^t$ characterized by $G$, and the joint distribution of $(T^t,W^t)$ denoted by $P$ --- when we wish to emphasize this dependence, we write $\nu_{\Pi,G,P}$. 
We define the relative efficiencies as
\begin{equation}
    \Phi_{a,\Pi,G}(P) = \frac{E_{\nu_{\Pi,G,P}}[D_a(\nu_{\Pi,G,P})(X^t)^2]}{E_{\nu_{\Pi,G,P}}[D_u(\nu_{\Pi,G,P})(X_u^t)^2]}, \ \Phi_{m,\Pi,G}(P) = \frac{E_{\nu_{\Pi,G,P}}[D_{m}(\nu_{\Pi,G,P})(X^t)^2]}{E_{\nu_{\Pi,G,P}}[D_u(\nu_{\Pi,G,P})(X^t_u)^2]},\label{eq:relEffsCAR}
\end{equation}
where $D_u$, $D_a$ and $D_m$ are the influence functions of the unadjusted, fully adjusted and working-model-based adjusted estimators, respectively. We will often suppress the dependence of these relative efficiencies on $\Pi$ and $G$ in the notation by writing $\Phi_{a}(P)$ and $\Phi_{m}(P)$. 

We aim to identify and estimate these relative efficiencies from external data available at the trial planning stage. Like the future trial data, the external data can be subject to coarsening. Let $C$ be the coarsening variable, and $\Gamma(\cdot,\cdot)$ be a many-to-one function. The full data unit in the external dataset is $Z = (T,W)$, and the observed data unit is $X = \Gamma(Z,C)$. Let $Q$ be the distribution of $X$, induced by the joint distribution of $(Z,C)$ and the many-to-one function $\Gamma$. To identify the relative efficiencies from the observed external data $X$, we assume the following condition holds throughout this section. This condition is similar to Condition~\ref{cond1identifiability} and assumes in addition that coarsening-at-random holds in the external data. 
\begin{conditionA}\label{cond1prime}
    A full data unit in the external data $Z = (T,W)$ has distribution $P$, and the conditional distribution $\Gamma(Z,C)\,|\,Z$ satisfies the coarsening-at-random assumption.
\end{conditionA}
Under this condition, it is possible to identify the relative efficiencies in \eqref{eq:relEffsCAR} as parameters of the distribution of the observed external data, and also to show that, under reasonable conditions, these parameters will be smooth enough so that it should be possible to develop regular and asymptotically linear estimators based on the external data \citep[Theorem~1.3 in][]{van2003unified}.

The external data might be obtained from various settings including observational studies, some of which are distinct from randomized clinical trials. Consequently, the reasons for coarsening can be much different from those in the future trial. For example, for time-to-event data, administrative censoring may account for a large proportion of right censoring in clinical trials, but a lesser proportion for observational data. Thus, it is often not plausible to assume that we can identify $G$ from the external data. To overcome this issue, we define the relative efficiencies for a particular $G$, and the user can choose a coarsening mechanism that is expected to reflect the setting of a future trial.

In Appendix~\ref{sec:survival}, we use the identifiability result stemming from Condition~\ref{cond1prime} to develop estimators and confidence intervals for the relative efficiency in settings where there are time-to-event outcomes with right censoring. In this case, $T^t$ is the time to some event of interest and $C^t$ is the censoring time in the trial. The full data unit $Z^t$ is $(T^t,A,W^t)$, and the observation unit $X^t$ is $(Y^t,\Delta^t,A,W^t)$, where $Y^t = \min\{T^t,C^t\}$ and $\Delta^t = I\{T^t\leq C^t\}$. 
The mapping that gives rise to this observation unit is given by $G_a(z^t,c^t)=(\min\{t^t,c^t\},I\{t^t\le c^t\},a,w^t)$. The validity of the unadjusted analysis relies on the condition that $T^t\perp C^t |A$, while the validity of the adjusted analysis relies on the condition that $T^t \perp C^t |(W^t, A)$. It is worth noting that, although the condition for the validity of the adjusted analysis can be more plausible in many settings, neither of these conditions implies the other --- this is a consequence of the fact that conditional independence does not imply marginal independence and marginal independence does not imply conditional independence. The external data consist of $X=(Y,\Delta,W)$ where $Y=\min\{T,C\}$ and $\Delta = I\{T\leq C\}$. Here $C$ is the censoring time in the external dataset. Letting $\Gamma(z,c)=(\min\{t,c\},I\{t\le c\},w)$, we see that the observed external data $X$ is equal to $\Gamma(Z,C)$. We consider three estimands of treatment effect, which are all functionals of the treatment-arm-specific survival function $S_a(t) := P(T^t>t|A=a)$. In particular, we develop estimators and confidence intervals for the risk difference (RD), the relative risk (RR), and the restricted mean survival time (RMST) --- see Appendix~\ref{sec:survival} for details.

\section{Experiments}\label{sec:experiment}
\subsection{Simulations}
For the ordinal outcome case, we generate data based on a CDC report describing the age distribution and probabilities of various outcomes within age groups for hospitalized Covid-19 patients \citep{cdc2020severe}, which are also presented in Table \ref{CDC}. The ordinal outcome is assigned the value 1, 2, or 3 for ``death", ``ICU and survived", or ``no ICU and survived", respectively. Age category is the only covariate we adjust for.

\begin{table}
 \caption{\label{CDC}Age distribution and probability of outcomes within age groups, among hospitalized Covid-19 patients \citep{cdc2020severe}. ``ICU" represents ICU admission.}
    \centering
    \begin{tabular}{ccccc}
    \hline
        age & P(age) & P(death $|$ age) & P(ICU $\&$ survived $|$ age) & P(no ICU $\&$ survived $|$ age)  \\
        \hline
        0-19 & 0.01 & 0.00 & 0.00 & 1.00\\
        20-44 & 0.09 & 0.01 & 0.18 & 0.81\\
        45-54 & 0.12 & 0.03 & 0.32 & 0.65\\
        55-64 & 0.13 & 0.08 & 0.31 & 0.61\\
        65-74 & 0.18 & 0.11 & 0.37 & 0.52\\
        75-84 & 0.22 & 0.17 & 0.47 & 0.36\\
        $\geq$ 85 & 0.25 & 0.37 & 0.35 & 0.28\\
        \hline
    \end{tabular}
\end{table}

We consider both the fully adjusted and working-model-based estimators. The relative efficiency of fully adjusted estimators is estimated with the analytical approach, while for the working-model-based estimators, we use both the analytical and the bootstrap approaches. We consider three estimands of the treatment effect: difference in mean, Mann-Whitney, and average log odds ratio. 

As the covariate is ordinal as well, the nuisance conditional mean functions are estimated by sample averages within each age group. In the analytical approach, we build Wald-type confidence intervals on the logit scale first and transform them. For the bootstrap, we take the number of bootstrap resamples in the two layers to be 100 and 500. Though these resample sizes are small compared to those used in typical applications of the bootstrap, we use them to reduce the computational cost in this Monte Carlo simulation. We do 1,000 replications for the analytical approach and 200 for the bootstrap.

The simulation results for sample size 1,000 are presented in Table \ref{simordinal}. We observe that despite the small number of resamples, the bootstrap procedure gives approximately $95\%$ coverage, but that this estimator has larger variance than does the analytical estimator across all settings considered. We expect the performance to improve as the number of resamples increases. The coverage of the analytical approach is close to the nominal level. Additional results for sample sizes 200 and 500 are given in Appendix~\ref{additionalsim}. We note that, as the true relative efficiency is strictly less than 1, the confidence sets constructed using the two-step approach detailed in Section~\ref{subsec:ordinalestimation} have the same coverage.

\begin{table}
\caption{\label{simordinal}Simulation results for ordinal outcome. We consider relative efficiency of fully adjusted and working-model-based estimators for DIM, MW and LOR. In the bootstrap approach, we take $B_1 = 100$ and $B_2 = 500$. Results are based on 1000 replications for analytic approach, 200 for bootstrap. ``F" stands for the fully adjusted estimator, and ``W" stands for the working-model-based estimator.}
    \centering
    \begin{tabular}{rrrrrrrr}
    \hline
         & truth & method & bias & MSE & $\%$RMSE & coverage & CI width \\
        \hline
        DIM (F) & 0.837 & analytic & 0.000 & 0.000 & 0.025 & 0.957 & 0.084 \\
        DIM (W) & 0.840 & analytic & -0.004 & 0.000 & 0.025 & 0.943 & 0.082 \\
        && bootstrap & 0.000 & 0.002 & 0.047 & 0.940 & 0.154 \\
        \hline
        MW (F) & 0.842 & analytic & 0.006 & 0.000 & 0.026 & 0.949 & 0.084 \\
        MW (W) & 0.845 & analytic & 0.002 & 0.000 & 0.025 & 0.957 & 0.083 \\
        && bootstrap & 0.001 & 0.002 & 0.048 & 0.935 & 0.160 \\
        \hline
        LOR (F) & 0.838 & analytic & 0.003 & 0.000 & 0.026 & 0.954 & 0.085 \\
        LOR (W) & 0.842 & analytic & -0.000 & 0.000 & 0.024 & 0.958 & 0.081 \\
        && bootstrap & 0.000 & 0.001 & 0.045 & 0.945 & 0.147 \\
        \hline
    \end{tabular}
\end{table}

For survival outcomes, we only consider the relative efficiency of the fully adjusted estimators. We generate a univariate covariate $W \sim \text{Uniform}(0,1)$, and the survival time follows an exponential distribution $Y|W \sim \text{Exp}\{(1+9W)/10\}$. The censoring time in the external data $C$ is generated from an $\text{Exp}(0.1)$ distribution independent of $W$. The user-specified censoring mechanism in the trial is taken to be the same, that is, $\text{Exp}(0.1)$. We again consider three estimands: risk difference (RD), relative risk (RR), and restricted mean survival time (RMST). The relative efficiency for RD and RR are the same under the null. For RMST, we discretize time with a 0.2 interval to reduce computation time and also mimic a setting where there are fixed follow-up times. 

With continuous time, the nuisance functions are estimated using a sequence of Cox proportional hazard models with polynomials of the covariate. 
We select the best model based on BIC. For discrete time, we use a proportional odds model instead, which slightly outperforms the Cox model in the simulations. Results for sample size 1,000 are presented in Table \ref{simsurvival}. The coverage of the confidence intervals is close to the nominal level across all settings. The uncertainty in the estimates becomes larger as time ($t$) increases, due to the reduced size of the risk set.

The R scripts for all the simulation experiments are available as supplementary files.

\begin{table}
   
    \caption{\label{simsurvival}Simulation results for survival outcome. We consider relative efficiency of fully adjusted estimators for RD at time 1, 2, and 3 (the relative efficiency is the same for RR) and RMST at time 3. Results are based on 1000 replications.}
    \centering
    \begin{tabular}{rrrrrrr}
    \hline
     & truth & bias & MSE & $\%$RMSE & coverage & mean width \\
    \hline
    RD ($t=1$) & 0.903 & 0.000 & 0.000 & 0.020 & 0.949 & 0.071\\
    RD ($t=2$) & 0.847 & 0.000 & 0.001 & 0.027 & 0.954 & 0.091 \\
    RD ($t=3$) & 0.819 & 0.002 & 0.001 & 0.034 & 0.941 & 0.106 \\
    RMST ($t=3$) & 0.820 & -0.002 & 0.001 & 0.028 & 0.952 & 0.091\\
    \hline
    \end{tabular}
\end{table}

\subsection{Application to Covid-19 data}

We apply the proposed methods to assess the efficiency gain of covariate-adjustment using Covid-19 data. The data contains information on 345 non-pregnant patients ($\ge$ 18 years old) admitted to University of Washington Medical Center through 6/15/2020. Among these patients, 40 were admitted twice and 3 were admitted three times. The following demographic and clinical features were measured at baseline: gender, age at admission, race (White, Asian, Black or African American, American Indian or Alaska Native and Native Hawaiian or other Pacific Islander), body mass index (kg/m$^2$), type I diabetes (yes/no), type II diabetes (yes/no), cardiovascular disease (CVD) (yes/no), hypertension (HTN) (yes/no), chronic kidney disease (yes/no), whether are on cholesterol medications (yes/no) and whether are on HTN medications (yes/no). Since only 4 patients have type I diabetes, we combine type I and type II diabetes as one single baseline feature and therefore have 10 baseline covariates in total. We discretize age into 7 groups ( $< $30, 30-40, 40-50, 50-60, 60-70, 70-80, $>$ 80). This is an observational dataset and there is no treatment information. The minimum of the censoring time and the times to each of the following events were measured: discharge, intubation, ventilation, and death. Time of hospital admission was treated as time zero.

\textit{Ordinal Outcome}. We use the following mutually exclusive ordinal outcome based on the severity of a patient's Covid-19 status: (1) censor or discharge, (2) intubation or ventilation, and (3) death.  Among 40 patients who had been admitted twice, only 14 patients had different outcomes between the two visits (9 patients were classified as 2 during first admission and as 1 in the second admission while 5 patients transited from 1 to 2). Among 3 patients who had been admitted three times, only 1 patient had different ordinal outcomes between 3 visits that he was classified a 1, 2, and 1 respectively. For all patients who had been admitted more than once, there was no death. To deal with duplicated observations for these patients, we only include the observations with a more severe outcome. As a result of the above classification, there are 207 (60\%) censor/discharges, 59 (17\%) intubation or ventilation, and 79 (23\%) deaths. We consider three estimands of treatment effects: difference in mean (DIM), Mann-Whitney (MW), and average log odds ratio (LOR). To estimate the nuisance functions for fully adjusted estimators, we fit a series of polynomial regressions from order 1 to 5 and then select the optimal model based on BIC score for DIM and MW. For LOR, these nuisance functions are estimated by fitting proportional odds models with polynomials of order 1 to 5 and selecting the best model based on BIC. We present the relative efficiency of covariate-adjusted estimators that adjust for all the covariates in Table~\ref{ordinalall}. The estimated efficiency gain is about 7\% for the fully adjusted estimator, whereas for the working-model-based estimators we do not see evidence of a significant efficiency gain. In contrast, adjusting for a single baseline covariate gives an estimated efficiency gain ranging from 1\% to 5\%, and the difference between using fully adjusted and working-model-based estimators is negligible when only adjusting for one covariate. We leave the details to Appendix \ref{additionalsim}.

\textit{Survival Outcome}. We choose the time point of interest to be $t = 350$ hours, where the overall survival is around 70\%, and assess the relative efficiency for survival outcomes. We consider three estimands of treatment effects: risk difference (RD), relative risk (RR), and restricted mean survival time (RMST). We use elastic net \citep{Friedman2010regularization} for variable selection, where the tuning penalty parameter is selected via 5-fold cross validation. In particular, we select those variables with nonzero coefficients. To estimate the nuisance functions, we then fit a sequence of Cox proportional hazards models with polynomials of orders 1 to 7 of the selected variables and select the model with the smallest BIC score. The results are shown in Table~\ref{survival}. Adjusting for a single baseline covariate gives an efficiency gain ranges from 1\% to 9\% in estimating RD or RR, with age being the most prognostic factor. A similar trend is observed for RMST. Using elastic net, we select the following 4 baseline factors: age, CVD, chronic kidney disease, and cholesterol medications. Adjusting for these four factors gives an 11\% efficiency gain in estimating RD or RR and RMST. 

\begin{table}
     \caption{\label{ordinalall} Relative efficiency (95\% CI) of fully adjusted and working-model-based estimators that adjust for all baseline covariates for estimating DIM, MW and LOR in the Covid-19 dataset. ``F" stands for the fully adjusted estimator, and ``W" stands for the working-model-based estimator.}
    \centering
    \begin{tabular}{lrr}
        \hline
          & F & W \\
        \hline
        DIM & 0.93 (0.88, 0.97) & 1.02 (0.95, 1.10) \\
        MW & 0.94 (0.92, 0.97) &  1.05 (0.98,  1.14) \\
        LOR & 0.93 (0.89,  0.98) & 1.01 (0.94, 1.08) \\
        \hline
        \end{tabular}
\end{table}

\begin{table}
    \caption{\label{survival}Relative efficiency (95\% CI) for estimating RD, RR and RMST in time-to-event setting in the Covid-19 dataset. Note under the null, relative efficiency of RD and RR are the same and therefore only the one for RD is presented. Selected variables include age, CVD, chronic kidney disease and cholesterol medications. }
    \centering
    \begin{tabular}{lcc}
        \hline
          & RD & RMST \\
        \hline
        age & 0.91 (0.79, 1.06) & 0.92 (0.87, 0.97)\\
        gender & 1.00 (0.88, 1.14) & 1.00 (1.00, 1.00)\\
        race & 1.00 (0.88, 1.14) & 1.00 (0.99, 1.00)\\
        CVD & 0.98 (0.85, 1.13) & 0.99 (0.97, 1.01)\\
        HTN & 1.00 (0.88, 1.14) & 1.00 (1.00, 1.01)\\
        diabetes & 1.00 (0.88, 1.14) & 1.00 (0.99, 1.00)\\
        kidney disease& 0.96 (0.84, 1.10) & 0.97 (0.93, 1.00)\\
        cholesterol meds & 0.98 (0.86, 1.12) & 0.98 (0.96, 0.99)\\
        HTN meds & 1.00 (0.88, 1.14) & 1.00 (1.00, 1.00)\\
        BMI & 0.99 (0.87, 1.14) &0.99 (0.97, 1.00)\\
        selected & 0.89 (0.76, 1.04) & 0.89 (0.84, 0.95)\\
        \hline
        \end{tabular}
\end{table}

\section{Discussion}\label{sec:discussion}
In this paper, we presented a framework to use external data to infer about the relative efficiency of covariate-adjusted analyses in a future clinical trial. We also exhibited the applicability of our framework for a variety of treatment effect estimands of particular interest. For each of these estimands, we introduced a consistent and asymptotically normal estimator of the relative efficiency and provided an analytic means to develop Wald-type confidence intervals. We also introduced a double bootstrap scheme that enables confidence interval construction in certain problems even when an analytic form for the standard error is not available.

When the outcome is only partially observed, standard unadjusted and adjusted analyses typically provide consistent estimators of the treatment effect under different assumptions on the coarsening mechanism. In our view, the choice between adjusted and unadjusted estimator should first and foremost be based on the plausibility of these assumptions. In settings where both sets of assumptions are plausible, the relative efficiency of the two estimators represents a natural criterion upon which to make this choice. Interestingly, unlike for fully adjusted estimators in uncoarsened settings, it is possible that the unadjusted estimator will, in fact, be more efficient than the adjusted estimator when both estimators are consistent. As a specific example, in the survival setting, our results in Theorem~\ref{corollaryreRD} show that the asymptotic variance of the adjusted estimator is smaller than that of the unadjusted estimator if the covariates are only predictive of the survival time, but is larger if the covariates are only predictive of the censoring time.

The relative efficiency we considered is based on a sharp null setting where the treatment has no effect. As a consequence, we do not need to specify the full distribution of $Y^t|A,W^t$ expected in the trial. Moreover, if the treatment effect estimator is regular, which is the case for all those that we considered, then the relative efficiency at this sharp null also serves as an accurate approximation to the relative efficiency under a variety of local alternatives. Though accurate in such settings, we expect that this approximation may be poor when the treatment is extremely beneficial in some subgroups while being quite harmful in some others. While a subgroup analysis might be able to detect this after the trial is completed, it is not generally possible to know \textit{a priori} whether this kind of subgroup effect exists. An alternative approach would involve specifying a particular alternative distribution that the investigator is interested in. In this case, the relative efficiency under that alternative can be derived and estimated. 
Our framework for estimating relative efficiencies based on external data can be easily modified for this setting. 

Observational settings and clinical trials can be quite different in terms of coarsening, and thus we define relative efficiency for a user-specified coarsening mechanism that approximates that of the future trial. This also extends to the case where the covariate distribution is different between the external data and the future clinical trial due to, for example, trial eligibility criteria. In such cases, a particular covariate distribution for the future trial can be imposed when defining the relative efficiency, and the external data can then be used to estimate the distribution of the outcome conditional on covariates.

\appendix

\section*{Appendix}

This appendix is organized as follows. In Appendix~\ref{sec:survival}, we develop estimators and confidence intervals for the relative efficiency in settings where there are time-to-event outcomes with right censoring. In Appendix~\ref{ordinalproof}, we prove lemmas and theorems in Section~\ref{sec:ordinal} on continuous and ordinal outcomes. In Appendix~\ref{survivalproof}, we prove lemmas and theorems in Appendix~\ref{sec:survival} on time-to-event outcomes with right censoring. In Appendix~\ref{additionalsim}, we show some additional experiment results. In Appendix~\ref{unionset}, we develop a two-step procedure with sample splitting to construct confidence intervals, and show that it achieves nominal coverage. In Appendix~\ref{app:pseudocode}, we give the pseudocode for the double bootstrap scheme presented in Section~\ref{sec:ordinal}.

\section{Estimation of relative efficiencies for time-to-event outcome with right censoring}\label{sec:survival}

We consider three estimands of treatment effect, which are all functionals of the treatment-arm-specific survival function $S_a(t) := P(T^t>t|A=a)$. For the unadjusted analysis, we consider plug-in estimators based on the treatment-arm-specific Kaplan-Meier estimator \citep{kaplan1958nonparametric}, which we denote as $\tilde S_a(t)$. Such plug-in estimators are consistent and asymptotically linear provided that $T^t \perp C^t |A$ \citep[see, for example,][]{diaz2019improved}.

In contrast, the consistency of covariate-adjusted estimators often relies on the assumption that $T^t \perp C^t |(W^t, A)$. In fact, many recently proposed adjusted estimators are based on the efficient influence function of the treatment effect estimand in a model where the only assumption is that $T^t \perp C^t |(W^t, A)$  \citep[e.g.,][]{moore2009increasing,stitelman2012general,diaz2019improved}. Under regularity conditions, these estimators achieve the semiparametric efficiency bound in this model. Constructing these estimators often requires estimation of nuisance functions such as the conditional hazard function $h_a(t,w) = P(T^t=t|T^t \geq t, A=a, W^t=w)$, the conditional survival function $S_a(t,w) = P(T^t > t|A=a,W^t=w)$, the conditional distribution of censoring time $G_a(t,w) := P(C^t \geq t|A=a,W^t=w)$ or the treatment mechanism $\pi(w)=P(A=1|W^t=w)$. We call these estimators ``fully adjusted".

As discussed in Section~\ref{sec:partiallyobserved}, the efficiency of an adjusted estimator relative to that of an unadjusted estimator is relevant only when both estimators are consistent --- as noted earlier, a sufficient condition for this to hold is that the observed data arises from a distribution in the intersection model consisting of all distributions of $(Z^t,C^t)$ for which $T^t \perp C^t |(W^t, A)$ and $T^t \perp C^t |A$. Notably, there is not generally any guarantee that a fully adjusted estimator will be efficient relative to the observed data model consisting of the distributions of $G_a(Z^t,C^t)$ generated by sampling $(Z^t,C^t)$ from a distribution in this intersection model. Stated more plainly, if it is known in advance that both the adjusted and unadjusted survival function estimators are consistent, then, in certain cases, there may exist a more efficient estimator of this survival function.

Unlike the cases of continuous or ordinal outcomes that we considered in Section~\ref{sec:ordinal}, we are not aware of a parametric working model for the conditional distribution of $T^t|A, W^t$ that yields a RAL estimator of $S_0$ and $S_1$ when marginalized over the distribution of the covariate $W^t$. Nevertheless, it is possible to define adjusted estimators based on working models in this setting. To see this, note that many of the aforementioned fully adjusted estimators do have the doubly robust property: they are consistent if either $(S_0,S_1)$ or $(G,\pi)$ is correctly specified, and are efficient if both are correctly specified. This allows us to use potentially misspecified parametric working models to estimate $(S_0,S_1)$ as long as we estimate the distribution of censoring time using a correctly specified semiparametric or nonparametric model --- this is the case, for example, if we estimate the censoring distribution via a correctly specified arm-specific Kaplan-Meier estimator. Such estimators are rarely used in practice. We, therefore, focus on computing the relative efficiency of fully adjusted estimators, which see more use, as compared to that of unadjusted estimators.

\subsection{Estimation of relative efficiency}\label{subsec:survestimation}

As in previous works \citep{moore2009increasing,stitelman2012general,diaz2019improved}, we assume that survival and censoring time are discrete, and take values in $\{t_1,t_2,\ldots,t_K\}$. We let $t_0 =0$ be the baseline time. We expect similar derivations can be done for continuous time, and in the simulation studies we empirically validate the performance of our proposed methods when time is measured on a continuous scale.

The first two estimands we consider focus on survival functions at a specific time point. The \emph{risk difference (RD)} is defined as $S_0(t_k) - S_1(t_k)$ for a time $t_k$ of interest. The \emph{relative risk (RR)} is defined as $\{1-S_1(t_k)\}/\{1-S_0(t_k)\}$ for a time $t_k$ of interest. 
We consider the unadjusted estimator $\tilde S_0(t_k) - \tilde S_1(t_k)$ for RD and $\{1-\tilde S_1(t_k)\}/\{1-\tilde S_0(t_k)\}$ for RR, where $\tilde S_a$ is the Kaplan-Meier estimator within each treatment group. Let $\hat S_a$ denote the efficient adjusted estimator proposed in \citet{moore2009application}. 
For each of the two estimands under consideration, we refer to the estimator that replaces $\tilde S_a$ in the unadjusted estimator with $\hat S_a$ as the fully adjusted estimator.

Recall that $\hat S_a$ is a consistent estimator of $S_a$ when $T^t \perp C^t |(A,W^t)$. Under additional regularity conditions given in Theorem~1 in \citet{moore2009increasing}, for each $a\in \{0,1\}$ and $k \in \{1,\ldots,K\}$, $\hat S_a(t_k)$ is an asymptotically linear estimator of $S_a(t_k)$ with influence function
\begin{align}
    (y^t,\delta^t,\tilde{a},w^t) &\mapsto \sum_{j=1}^{k} -\frac{I\{\tilde{a}=a\}S_a(t_k,w^t)}{\pi_aS_a(t_j,w^t)G_a(t_j,w^t)}\left[\delta^t I\{y^t =t_j\}-I\{y^t \geq t_j\}h_a(t_j,w^t)\right] \nonumber \\
    &\quad +S_a(t_k,w^t)-S_a(t_k). \label{eq:SHatIF}
\end{align}
Moreoever, for each $a\in \{0,1\}$ and $k \in \{1,\ldots,K\}$, $\tilde S_a(t_k)$ is a RAL estimator of $S_a(t_k)$ when $T^t \perp C^t|A$ with influence function \citep[see, e.g.,][]{diaz2019improved}
\begin{align}
    (y^t,\delta^t,\tilde{a},w^t) &\mapsto  \sum_{j=1}^{k}-\frac{I\{\tilde{a}=a\}S_a(t_k)}{S_a(t_j)G_a(t_j)\pi_a}\left[\delta^t I\{y^t = t_j\}-h_a(t_j)I\{y^t \geq t_j\}\right]. \label{eq:STildeIF}
\end{align}
Here, $h_a(t)$ is the hazard corresponding to $S_a$ at time $t$ and $G_a(t) := P(C^t\geq t|A=a)$. The influence functions of the fully adjusted and unadjusted estimators of the treatment effect estimand, which we denote as $D_a$ and $D_u$, respectively, can then be derived via the delta method.

As in Section~\ref{sec:partiallyobserved}, we define the relative efficiency as the ratio between the variances of $D_a$ and $D_u$ under the sharp null. In such cases, the distribution of the observed data in the trial is characterized by the marginal distribution of $A$, denoted by $\Pi$, the joint distribution of $(T^t,W^t)$, denoted by $P$, and the conditional distribution of $C^t$ given $(A,W^t)$. In particular, this implies that $S_1(t,w) = S_0(t,w) = S(t,w)$ for all $(t,w)$, where $S(t,w) := P(T^t>t|W^t=w)$ is the conditional survival function under $P$, and also that $S_1(t) = S_0(t) = S(t)$ for all $t$, where $S(t) := P(T^t>t)$ is the marginal survival function under $P$. To simplify the presentation, we suppose additionally that $C^t\perp A|W^t=w$, and write $G(t,w) := P(C^t \geq t | W^t=w)$ and $G(t):= P(C^t \geq t)$. For given $G$ and $\Pi$, the relative efficiency parameter is a functional of $P$.

Before presenting the form of the relative efficiency, we introduce some additional needed notation. For $(T,W) \sim P$, let $h(t,w) := P(T=t|T\geq t, W=w)$ and $h(t):= P(T=t|T\geq t)$ be the conditional and marginal hazard functions under $P$, respectively. We define the following quantities, which will be useful throughout this section:
\begin{align}\label{definef}
    s_j^{kl} &= \frac{S(t_k)S(t_l)\{S(t_{j-1})-S(t_j)\}}{S(t_j)S(t_{j-1})},\textnormal{ and} \nonumber \\
    f_j^{kl}(w) &= \frac{S(t_k,w)S(t_l,w)\{S(t_{j-1},w)-S(t_j,w)\}}{S(t_j,w)S(t_{j-1},w)}.
\end{align}
Interestingly, in the null case that we consider, the relative efficiencies are the same for the RD and RR estimands.

\begin{lemma}\label{reRD}
Suppose that Conditions~\ref{cond2bounded} and \ref{cond1prime} hold and, in addition, that $S(t_k)<1$. Suppose that $\hat S_a$ and $\tilde{S}_a$ are asymptotically linear with influence functions given in \eqref{eq:SHatIF} and \eqref{eq:STildeIF}, respectively. For both the RR and RD estimand, the relative efficiency of the fully adjusted estimator as compared to the unadjusted estimator is given by $\Phi_a(P) = \sigma_a^2(P)/\sigma_u^2(P)$, where
\begin{align}
    \sigma^2_a(P) = E_P\left[\sum_{j=1}^{k}f_j^{kk}(W)/G(t_j,W)\right], \ \ \sigma^2_u(P) = \sum_{j=1}^{k}s_j^{kk}/G(t_j).
\end{align}
\end{lemma}
In what follows, we will often write $\phi_a$ for $\Phi_a(P)$.

To estimate $\phi_a$ from the external data $(Y,\Delta,W)$, we estimate $\sigma_u^2$ and $\sigma_a^2$ separately. We observe that $\sigma_u^2$ is a transformation of $S(t)$, and hence we construct a plug-in estimator $\hat s_{j}^{kl}$ of $s_{j}^{kl}$ using covariate-adjusted estimator of $S(t)$ given in \citet{moore2009application} and estimate $\sigma_u^2$ by 
$\hat\sigma_u^2 = \sum_{j=1}^{k}\hat s_{j}^{kk}/G(t_j)$. We estimate $\sigma^2_a$ using one-step estimation based on its EIF. Recall that $C$ is the censoring time in the external data. We define $H(t,w) := P(C\geq t|W=w)$. For notational convenience, we define the following function, which appears multiple times in the EIF of $\sigma^2_a$:
\begin{multline*}
    g_j^{kl}(y,\delta,w) = \left\{\frac{S(t_k,w)}{S(t_j,w)}-\frac{S(t_k,w)}{S(t_{j-1},w)}\right\}\tau_l(y,\delta,w) +\left\{\frac{S(t_l,w)}{S(t_j,w)}
    -\frac{S(t_l,w)}{S(t_{j-1},w)}\right\}\tau_k(y,\delta,w) \\-\frac{S(t_k,w)S(t_l,w)}{S^2(t_j,w)}\tau_j(y,\delta,w) + \frac{S(t_k,w)S(t_l,w)}{S^2(t_{j-1},w)}\tau_{j-1}(y,\delta,w),
\end{multline*}
where
\begin{equation}\label{eq:tau}
    \tau_l(y,\delta,w) = \sum_{u\leq l}-\frac{S(t_l,w)}{S(t_u,w)H(t_u,w)}\left[I\{y=t_u,\delta=1\}-h(t_u,w)I\{y \geq t_u\}\right].
\end{equation}
The efficient influence function of $\sigma_a^2$ relative to the observed data model is 
\begin{equation}\label{EIFcondvarsurv}
    \textnormal{IF}_a(y,\delta,w) = \sum_{j=1}^{k} \frac{1}{G(t_j,w)}\left\{g_j^{kk}(y,\delta,w)
    +f_j^{kk}(w)\right\}-\sigma_a^2. 
\end{equation}
The derivation of this expression is deferred to Appendix~\ref{survivalproof}. Let $\hat S(t,w)$ and $\hat h(t,w)$ be estimators of the conditional survival and hazard functions, respectively. Let $\hat H(t,w)$ be an estimator of the conditional censoring distribution. Define $\hat g_{j}^{kk}, \hat f_{j}^{kk}$ with these estimates. We estimate $\sigma_a^2$ with
\begin{equation*}
    \hat\sigma_a^2 = \frac{1}{n}\sum_{i=1}^{n}\sum_{j=1}^{k}\frac{1}{G(t_j,W_i)}\left\{\hat g_j^{kk}(Y_i,\Delta_i,W_i)+\hat f_j^{kk}(W_i)\right\}.
\end{equation*}
We then estimate $\phi_a$ by $\hat\phi = \hat \sigma_a^2/\hat \sigma_u^2$. The properties of $\hat\phi$ are given in the following theorem.

\begin{theorem}\label{RD}
Suppose that (1) Conditions~\ref{cond2bounded} and \ref{cond1prime} hold; (2) $\hat S(t,w)$, $\hat H(t,w)$, $S(t,w)$, $H(t,w)$ and $G(t,w)$ are all uniformly bounded away from 0, and $\hat S(t,w)$, $\hat h(t,w)$ and $h(t,w)$ are uniformly bounded above; (3) for all $t \in \{t_1,\ldots,t_K\}$, the random functions $\hat H(t,\cdot):\mathcal{W} \rightarrow \mathbb{R}$, $\hat S(t,\cdot):\mathcal{W} \rightarrow \mathbb{R}$ and $\hat h(t,\cdot):\mathcal{W} \rightarrow \mathbb{R}$ are such that $\|\hat H(t,\cdot) - H(t,\cdot)\|_{L^2(P_W)} = o_P(n^{-1/4})$, $\|\hat S(t,\cdot) - S(t,\cdot)\|_{L^2(P_W)} = o_P(n^{-1/4})$, $\|\hat h(t,\cdot) - h(t,\cdot)\|_{L^2(P_W)} = o_P(n^{-1/4})$ and they all belong to a certain fixed $Q$-Donsker class $\mathcal{F}$ of functions with probability tending to one. Then, $\hat\phi$ is an efficient estimator of $\phi_a$. 
\end{theorem}
The influence function of $\hat\phi$ is given in Appendix~\ref{survivalproof}.

The final treatment effect estimand we consider is the \emph{restricted mean survival time (RMST)}, defined as $\psi = \sum_{j=1}^{k}\{S_1(t_j)-S_0(t_j)\}$. We again consider two plug-in estimators: the unadjusted KM-based estimator $\hat\psi_u = \sum_{j=1}^{k}\{\tilde S_1(t_j)-\tilde S_0(t_j)\}$ and the fully-adjusted estimator $\hat\psi_a = \sum_{j=1}^{k}\{\hat S_1(t_j)-\hat S_0(t_j)\}$.

\begin{lemma}\label{reRMST}
Suppose that Conditions~\ref{cond2bounded} and \ref{cond1prime} hold. Then, the relative efficiency $\phi_a$ takes the following form
\begin{multline*}
    \phi_a = \sigma_a^2(P)/\sigma_u^2(P), \ \ \textnormal{where} \\
    \sigma_a^2(P) = \sum_{j=1}^{k}\sum_{l=1}^{k}\sum_{u=1}^{\min(j,l)}E_{P}[f_{u}^{jl}(W)/G(t_u,W)], \ \ \sigma_u^2(P) = \sum_{j=1}^{k}\sum_{l=1}^{k}\sum_{u=1}^{\min(j,l)}s_{u}^{jl}/G(t_u).
\end{multline*}
\end{lemma}

We construct $\hat\sigma_u^2$ in a similar way as in the case of the risk difference, namely by plugging in an efficient adjusted estimator of $S(\cdot)$. As for $\sigma_a^2$, its EIF can be derived in a similar fashion as in the case of RD, and we defer the details to Appendix~\ref{survivalproof}. We propose the following estimator 
\begin{equation}
    \hat\sigma_a^2 = \frac{1}{n} \sum_{i=1}^{n}\sum_{j=1}^{k}\sum_{l=1}^{k}\sum_{u=1}^{\min(j,l)}\left\{\hat g_u^{jl}(Y_i,\Delta_i,W_i)+\hat f_u^{jl}(W_i)\right\}.\label{eq:quadrupleSum}
\end{equation}
The relative efficiency is estimated by $\hat\phi = \hat\sigma_a^2/\hat\sigma_u^2$. 

Based on \eqref{eq:quadrupleSum}, it appears that computing the quadruple sum used to define $\hat\sigma_a^2$ will take order $nk^3$ time. As it turns out, these sums can be computed much more efficiently. In Appendix~\ref{survivalproof}, we show that for a given $i$ and given estimates of $\tau$ and $S$, the inner three sums can be computed in $O(k)$ time, resulting in an $O(nk)$ complexity for computing the above quadruple sum.

\begin{theorem}\label{RMST}
Under the same conditions as in Theorem \ref{RD}, $\hat\phi$ is an efficient estimator of $\phi_a$. 
\end{theorem}
Again, the specific form of its influence function is given in Appendix~\ref{survivalproof}.
\begin{remark}
As discussed at the end of Section~\ref{subsec:ordinalestimation}, the influence function of $\hat\phi$ is identically 0 in certain special cases. One such case arises when $T \perp W$ under $P$ and the mapping $G$ does not depend on $W$. In these cases, a two-step procedure can be considered for inference --- see the end of Section~\ref{subsec:ordinalestimation} for a description of such an approach in a similar setting.
\end{remark}

As noted earlier, for it to be interesting to compare the efficiency of the adjusted and unadjusted estimators, it must be the case that both are asymptotically linear. In such settings, we now characterize cases in which the adjusted estimator will be more efficient than will the unadjusted estimator. Moreover, unlike in the uncoarsened data setting, there are also settings where the adjusted estimator may be \textit{less} efficient than the unadjusted estimator. We also characterize these cases.

\begin{theorem}\label{corollaryreRD}
For all three estimands considered:
\begin{enumerate}
    \item If the conditions in Lemma~\ref{reRD} hold and $P$ is such that $T\perp W$, then $\sigma_u^2(P) \leq \sigma_a^2(P)$, that is, the unadjusted estimator is at least as efficient as the adjusted estimator. Moreover, if $\textnormal{var}_P[G(t_j,W)]>0$ for some $j \in \{1,\ldots,k\}$, then the inequality is strict: $\sigma_u^2(P) < \sigma_a^2(P)$.
    \item If the conditions in Lemma~\ref{reRD} hold and $G(\cdot,\cdot)$ is such that $G(t,w_1) = G(t,w_2)$ for all $t \leq t_k$ and $w_1,w_2 \in \mathcal{W}$, 
    then $\sigma_a^2(P) \leq \sigma_u^2(P)$, that is, the adjusted estimator is at least as efficient as the unadjusted estimator. 
\end{enumerate}
\end{theorem}

\section{Proofs of results in the case where the outcome is fully observed}\label{ordinalproof}

\subsection{Supporting lemmas for proofs in Section~\ref{sec:ordinal}}

\begin{lemma}[EIF of mean conditional variance]\label{gradientcondvar}
For a given function $u_P(\cdot)$ that can be written as $u_P(y) = \int h(x,y)dP(x)$ for some function $h$, the canonical gradient of $\sigma_a^2 = E_P[\textnormal{var}(u_P(Y)|W)]$ is given by 
\begin{equation}
    D_a(y,w) = \left\{u_P(y)-f_P(w)\right\}^2 + 2 \int \{u_P(\tilde y)-f_P(\tilde w)\}h(y,\tilde y)dP(\tilde y, \tilde w)-3\sigma_a^2,
\end{equation}
where $f(w) = E_P\left[u_P(Y)|W=w\right]$.
\end{lemma}

\begin{proof}
We prove this lemma by directly applying the definition of a gradient. We consider the one-dimensional submodel $\{P_\epsilon : |\epsilon|\le 1\}$ with density $$p_\epsilon(y,w) = p(y|w)\{1+\epsilon s_1(y|w)\}p(w)\{1+\epsilon s_2(w)\},$$
where the range of $s_1$ and $s_2$ falls in $[-1,1]$ and these functions satisfy $E_P[s_1(Y|W)|W]=0$ $P$-almost surely and $E_P[s_2(W)]=0$.
Let $f_P(w) = E_P[u_P(Y)|W=w]$ and $f_{P_\epsilon}(w) = E_{P_\epsilon}[u_{P_\epsilon}(Y)|W=w]$. We have that
\begin{align*}
    \sigma_a^2(P_\epsilon) &= \int\left\{u_{P_\epsilon}(y)-f_{P_\epsilon}(w)\right\}^2\{1+\epsilon s_1(y|w)\}\{1+\epsilon s_2(w)\}dP(y|w)dP(w) \nonumber \\
    &= \int\left\{u_{P_\epsilon}(y)-f_{P_\epsilon}(w)\right\}^2dP(y,w) \\
    &\quad+ \int\left\{u_{P_\epsilon}(y)-f_{P_\epsilon}(w)\right\}^2\{\epsilon s_1(y|w)+\epsilon s_2(w) +\epsilon^2 s_1(y|w)s_2(w)\} dP(y,w).
\end{align*}
The second term on the right has the following derivative with respect to $\epsilon$ at $\epsilon=0$ 
\begin{equation*}
    \int\left\{u_{P}(y)-f_{P}(w)\right\}^2\{s_1(y|w)+ s_2(w)\} dP(y,w),
\end{equation*}
and so will contribute $\{u_{P}(Y)-f_{P}(W)\}^2-\sigma_a^2$ to the gradient. We now focus on the first term, which re-writes as
\begin{align*}
    &\int \left[\int h(x,y)\{1+\epsilon s_1(x|w)\}\{1+\epsilon s_2(w)\}dP(x|w)dP(w)-f_{P_\epsilon}(w)\right]^2 dP(y,w) \nonumber \\
    &= \int\left[\left\{\int h(x,y)dP(x,w)-f_{P_\epsilon}(w)\right\}^2+c(\epsilon)^2+2c(\epsilon)\left\{\int h(x,y)dP(x,w)-f_{P_\epsilon}(w)\right\}\right]dP(y,w),
\end{align*}
where $c(\epsilon) = \int h(x,y)\{\epsilon s_1(x|w)+\epsilon s_2(w)+\epsilon^2s_1s_2\}dP(x,w)$. The first term in the above display has derivative $0$ with respect to $\epsilon$ at $\epsilon = 0$, as $f_P$ is the true conditional mean. The second term also has derivative 0 at $\epsilon=0$, as it is quadratic in $\epsilon$. At $\epsilon=0$, the third term has derivative
\begin{align*}
    &2\int\left(\left[\int h(x,y)\{s_1(x|w)+s_2(w)\}dP(x,w)\right]\left[\int h(x,y)dP(x,w)-f_{P}(w)\right]\right) dP(y,w) \nonumber \\
    &= 2\int\int\left(\left\{\int h(x,y)dP(x,w)-f_{P}(w)\right\}\left[ h(x,y)\{s_1(x|\tilde w)+s_2(\tilde w)\}\right]\right)dP(x,\tilde w)dP(y,w) \nonumber \\
    &= 2\int\left[\int\left\{u_P(y)-f_P(w)\right\}h(x,y)dP(y,w)\right]\{s_1(x|\tilde w)+s_2(\tilde w)\}dP(x,\tilde w).
\end{align*}
The inner integral has mean 
\begin{equation*}
    \int \int \left\{u_P(y)-f_P(w)\right\}h(x,y)dP(y,w)dP(x,\tilde w) = \int \{u_P(y)-f_P(w)\}u_P(y) dP(y,w) = \sigma_a^2.
\end{equation*}
Therefore the following is a gradient:
\begin{equation*}
    D_a(y,w) = \left\{u_P(y)-f_P(w)\right\}^2 + 2 \int \{u_P(\tilde y)-f_P(\tilde w)\}h(y,\tilde y)dP(\tilde y, \tilde w)-3\sigma_a^2.
\end{equation*}
Since we are working within a locally nonparametric model, the above is also the canonical gradient.
\end{proof}

\begin{lemma}[EIF of mean conditional covariance]\label{gradientcondcov}
Consider a locally nonparametric model of distributions of $(Y,W)$. For given functions $u(\cdot)$ and $v(\cdot)$, the canonical gradient of $\sigma_{uv} = E_P[\text{cov}_P(u(Y),v(Y)|W)]$ is
\begin{equation}
    D_{cov}(y,w) = \{u(y) - f_u(w)\}\{v(y)-f_v(w)\} - \sigma_{uv},
\end{equation}
where $f_u(w) = E_P[u(Y)|W=w]$ and $f_v(w) = E_P[v(Y)|W=w]$.
\end{lemma}

\begin{proof}
As in the proof of Lemma~\ref{gradientcondvar}, we consider a one-dimensional submodel $\{P_{\epsilon} : \epsilon\}$ with density $p_\epsilon(y,w) = p(y,w)\{1+\epsilon h(y,w)\}$. Let $f_{P,u}(w) = E_P[u(Y)|W=w]$ and $f_{P,v}(w) = E_P[v(Y)|W=w]$. Note that
\begin{align*}
    \sigma_{uv}(P_\epsilon) &= \int \left\{u(y)-f_{P_\epsilon,u}(w)\right\}\left\{v(y)-f_{P_\epsilon,v}(w)\right\}\{1+\epsilon h(y,w)\}dP(y,w).
\end{align*}
Also, because
\begin{equation*}
    \int \left[ \left\{u(y)-f_{P,u}(w)\right\}\frac{d}{d\epsilon}f_{P_\epsilon,v}(w)\Big\rvert_{\epsilon=0} + \left\{v(y)-f_{P,v}(w)\right\}\frac{d}{d\epsilon}f_{P_\epsilon,u}(w)\Big\rvert_{\epsilon=0}\right]dP(y,w) = 0,
\end{equation*}
it holds that
\begin{align*}
    \frac{d}{d\epsilon}\sigma_{uv}(P_\epsilon)\Big\rvert_{\epsilon=0} &= \int \left\{u(y)-f_{P,u}(w)\right\}\left\{v(y)-f_{P,v}(w)\right\}h(y,w)dP(y,w).
\end{align*}
This shows that $D_{cov}$ is a gradient, and, because the model is locally nonparametric, $D_{cov}$ must therefore be the canonical gradient.
\end{proof}

Let $X$ be a generic random variate with distribution $P \in \mathcal{M}$ with support in a bounded set $\mathcal{X} \subset \mathbb{R}^d$. Let $U_\alpha: \mathcal{X} \xrightarrow[]{} \mathbb{R}^m$ be a function indexed by $\alpha \in \mathbb{R}^m$. Suppose that $PU_\alpha = 0$ has a unique solution in $\alpha$, and we denote this solution by $\psi_1 = \psi_1(P)$. In general, we can regard $\psi_1(P)$ as a parameter defined implicitly through the estimating equation. The following lemma establishes the pathwise differentiability of this and a related parameter, under appropriate conditions.

\begin{lemma}[Pathwise differentiability of parameters defined via estimating equations]\label{gradientEE}
Let $\psi_1: \mathcal{M} \xrightarrow[]{} \mathbb{R}^m$ be such that $\psi_1(P)$ is the unique solution in $\alpha$ to the estimating equation $PU_\alpha=0$. Suppose that, for each $x\in\mathcal{X}$, $U_\alpha(x)$ is continuously differentiable in $\alpha$, with derivative $\dot U_\alpha(x) := \frac{\partial}{\partial \tilde\alpha} U_{\tilde\alpha}(x)\rvert_{\tilde\alpha=\alpha}$. Suppose in addition that $U_{\psi_1},\dot U_{\psi_1} \in L^2(P)$ and that $P\dot{U}_{\psi_1}$ is invertible. 
Then, $\psi_1$ is pathwise differentiable and its gradient relative to any locally nonparametric model is given by
\begin{equation*}
    D_1(x) = -\left(P\dot{U}_{\psi_1}\right)^{-1} U_{\psi_1}(x).
\end{equation*}
Moreover, for each $\alpha\in\mathbb{R}^m$, let $g_\alpha: \mathcal{X} \xrightarrow[]{} \mathbb{R}$ be a function, and suppose that $\alpha\mapsto g_\alpha(x)$ is differentiable for all $x \in \mathcal{X}$. For each $P\in\mathcal{M}$, define $\psi_2 := Pg_{\psi_1}$. Then $\psi_2 : \mathcal{M} \rightarrow \mathbb{R}^m$ is pathwise differentiable with gradient
\begin{equation*}
    D_2(x) =  g_{\psi_1}(x) + P\left(\frac{\partial}{\partial \alpha}g_\alpha(x)\rvert_{\alpha =\psi_1}\right)^\top D_1(x)-\psi_2.
\end{equation*}
\end{lemma}

\begin{proof}
For $h \in L^2_0(P)$ whose range is contained in $[-1,1]$, consider the one-dimensional submodel $\{P_\epsilon: \epsilon\}$, where each $P_\epsilon$ has density $p_\epsilon(x) =\{1+\epsilon h(x)\}p(x)$. This submodel has score $h$ at $\epsilon = 0$. By definition, we have that $P_\epsilon U_{\psi_1(P_\epsilon)} = \int U_{\psi_1(P_\epsilon)}(x)\{1+\epsilon h(x)\}p(x)dx = 0$. Define a function $f: (\alpha,\epsilon) \mapsto  \int U_\alpha(x)\{1+\epsilon h(x)\}p(x)dx$, which is linear in $\epsilon$. Thus, the continuous differentiability of $U$ as a function of $\alpha$ implies that $f$ is also continuously differentiable. Then the implicit function theorem implies that
\begin{equation*}
    \frac{\partial}{\partial \epsilon}\psi_1(P_\epsilon)\rvert_{\epsilon=0} = -\left\{P\dot{U}_{\psi_1(P)}\right\}^{-1}\int U_{\psi_1(P)}(x)h(x)dP(x) = \langle D_1,h \rangle_P.
\end{equation*}
We note that $D_1 \in L^2_0(P)$. Hence by definition $\psi_1$ is pathwise differentiable with gradient $D_1$, which is also the only gradient in any locally nonparametric model.

Also, noting that $\psi_2(P_\epsilon) = \int g_{\psi_1(P_\epsilon)}(x)\{1+\epsilon h(x)\}dP(x)$, we see that
\begin{align*}
    \frac{\partial}{\partial \epsilon}\psi_2(P_\epsilon)\rvert_{\epsilon=0} &= \int g_{\psi_1}(x)h(x)dP(x) + \left\{\int \frac{\partial}{\partial \alpha}g_\alpha(x)\rvert_{\alpha =\psi_1}dP(x)\right\}^\top \left\{\frac{\partial}{\partial \epsilon}\psi_1(P_\epsilon)\rvert_{\epsilon=0}\right\} \\
    &= \int \left[ g_{\psi_1}(x) + P\left\{\frac{\partial}{\partial \alpha}g_\alpha(x)\rvert_{\alpha =\psi_1}\right\}^\top D_1(x)\right]h(x)dP(x) \\
    &= \langle D_2,h \rangle_P
\end{align*}
Thus, $\psi_2$ is pathwise differentiable and $D_2$ is the gradient in any locally nonparametric model. 
\end{proof}

\subsection{Results in Section~\ref{sec:ordinal}}\label{proofB1}

\begin{proof}[Proof of Lemma \ref{reATE}]
For notational convenience, we define $\mu_a = E[Y^t|A=a]$, for $a \in\{0,1\}$. 

We first establish properties of the working-model-based estimator of $\mu_a$, given by $\hat\mu_a = \hat\alpha_a+\hat\beta_a^\top  \bar{W^t}$. To do this, we note that the fitted coefficients from arm-specific linear regression $(\hat\alpha_a,\hat\beta_a^\top )$ satisfy the following first-order conditions:
\begin{equation*}
    \sum_{i=1}^{n^t}I\{A_i=a\}(Y_i^t-\hat\alpha_a-\hat\beta_a^\top  W_i^t) = 0; \quad \sum_{i=1}^{n}I\{A_i=a\}W_i^t(Y_i^t-\hat\alpha_a-\hat\beta_a^\top W_i^t) = 0.
\end{equation*}
Let $(\alpha_a^*,\beta_a^*)$ be the large-sample limit of $(\hat\alpha_a,\hat\beta_a)$, defined implicitly as the solution to
\begin{equation*}
    E[I\{A=a\}(Y^t-\alpha_a-\beta_a^\top W^t)] = 0; \quad E[I\{A=a\}W^t(Y^t-\alpha_a-\beta_a^\top W^t)] = 0.
\end{equation*}
The efficient influence function of $\mu_a$ when the treatment is randomized is given by 
\begin{equation*}
    D_a^*(y^t,a^t,w^t) = \frac{I\{a^t=a\}}{\pi_a}\{y^t-r_a(w^t)\} + r_a(w^t) - \mu_a,
\end{equation*}
where $\pi_a= \Pi(A=a)$. We claim that $D_a$ defined below is also a gradient of $\mu_a$ in this model.
\begin{equation*}
    D_a(y^t,a^t,w^t) = \frac{I\{a^t=a\}}{\pi_a}\left\{y^t-\alpha_a^*-(\beta_a^*)^\top w^t\right\} + \alpha_a^*+(\beta_a^*)^\top  w^t-\mu_a.
\end{equation*}
To show this, we will show that $D_a^*-D_a$ lies in the orthogonal complement of the tangent space. First, define
\begin{align*}
    L_{0,A}^2(\nu) &= \{s \in L_0^2(v): s(y,a,w) = s(y^\prime,a,w^\prime) \ \forall (y,w),(y^\prime,w^\prime)\};\\
    L_{0,W^t}^2(\nu) &= \{s \in L_0^2(v): s(y,a,w) = s(y^\prime,a^\prime,w) \ \forall (y,a),(y^\prime,a^\prime)\}; \\
    L_{0,Y^t|A,W^t}^2(\nu) &= \{s \in L_0^2(v): E[s(Y^t,A,W^t)|A=a,W^t=w] = 0 \ \forall a,w\}.
\end{align*}
The tangent space decomposes as $L_{0,A}^2(\nu) \bigoplus L_{0,W^t}^2(\nu) \bigoplus L_{0,Y^t|A,W^t}^2(\nu)$. Next,
\begin{equation*}
    D_a(Y^t,A,W^t)-D_a^*(Y^t,A,W^t) = \left[\frac{I\{A=a\}}{\pi_a}-1\right]\{r_a(W^t)-\alpha^*_a-\beta^*_aW^t\}.
\end{equation*}
We note that each individual factor in the above display has mean 0. This, together with the independence between $A$ and $W^t$, implies that $\langle s, D_a-D_a^*\rangle_\nu = 0$ for any $s$ in $L_{0,A}^2(\nu)$, $L_{0,W^t}^2(\nu)$ or $L_{0,Y^t|A,W^t}^2(\nu)$. This implies that the difference is orthogonal to each component of the tangent space, and hence the tangent space itself.

Next, we note that $\hat\mu_a$ re-writes as a one-step estimator based on the gradient $D_a$. Let $\hat\nu$ be a distribution of $(Y^t,A,W^t)$ such that $E_{\hat\nu}[Y^t|A=a,W^t=w] = \hat\alpha_a + \hat\beta_a^\top  w$ and $E_{\hat\nu}[I\{A=a\}|W^t=w] = \hat\pi_a$, the sample proportion of $A=a$. The remainder is given by
\begin{align*}
    \psi(\hat\nu)-\psi(\nu) + P_\nu D_a(\hat\nu) &= P_\nu\left[\frac{I\{A=a\}}{\hat\pi_a}(Y^t-\hat\alpha_a-\hat\beta_a^\top W^t) + \hat\alpha_a+\hat\beta_a^\top W^t - r_a(W^t)\right] \\
    &= \left(\frac{\pi_a}{\hat\pi_a}-1\right)P_\nu\left\{r_a(W^t)-\hat\alpha_a-\hat\beta_a^\top W^t\right\},
\end{align*}
which is $o_P(n^{-1/2})$ as the first term is $O_P(n^{-1/2})$ and the second term is $o_P(1)$. We note that $Y$ and $W$ both have bounded support and $\pi_a \in (0,1)$. Thus, $\|D_a(\hat\nu) - D_a\|_{L^2(\nu)} = o_P(1)$, which follows from the convergence of $\hat\alpha_a$, $\hat\beta_a$ and $\hat\pi_a$ to their population counterparts. Let $\gamma^* := (\pi_a,\alpha_a^*,\beta_a^*,\mu_a)$ be the true parameter value and $\gamma = (\pi,\alpha,\beta,\mu)$ be a generic parameter value, and define a class of functions $\mathcal{F} := \{f_\gamma: (y,\tilde a,w) \mapsto I\{\tilde a=a\}\left(y-\alpha-\beta^\top w\right)/\pi + \alpha+\beta^\top  w-\mu, \gamma \in B(\gamma^*)\}$ where $B(\gamma^*)$ denotes a small neighborhood of $\gamma^*$. We note that $D_a(\hat\nu) \in \mathcal{F}$ with probability tending to 1. The boundedness of $Y^t$, $W^t$ and $\pi_a$ implies that this class of functions is Lipschitz in $\gamma$, and Example 19.7 in \citet{van2000asymptotic} implies that $\mathcal{F}$ is a Donsker class. 

Therefore, $\hat\mu_a$ is asymptotically linear with influence function $D_a$, and so is $\hat\psi_m$ with influence function $D_m = D_1-D_0$. Now focusing on $\nu$ where the sharp null holds, $D_m$ simplifies to 
\begin{equation*}
    D_m(y^t,a^t,w^t) = \left(\frac{a^t}{\pi_1} - \frac{1-a^t}{\pi_0}\right)(y^t-\alpha^*-\beta^*w^t)
\end{equation*}
Due to the independence of $A$ and $(Y^t,W^t)$ under the sharp null, it has variance $\sigma^2_m/(\pi_1\pi_0)$ under the sharp null. 

Now, we derive the asymptotic variances of the unadjusted estimator and the fully adjusted estimator. In doing so, we consider a more general parameter $\psi = E[u(Y^t)|A=1] - E[u(Y^t)|A=0]$ for a given function $u$. The average treatment effect corresponds to the special case of $u$ being the identity function. 

Recall that the unadjusted estimator is given by the difference between arm-specific means
$$\hat\psi_u = \frac{\sum_{i=1}^{n^t}u(Y_i^t)A_i}{\sum_{i=1}^{n^t}A_i}- \frac{\sum_{i=1}^{n^t}u(Y_i^t)(1-A_i)}{\sum_{i=1}^{n^t}(1-A_i)}.$$
Applying the delta method, we see that $\hat\psi_u$ is asymptotically linear with influence function
$$D_u(y^t,a^t,w^t) = \left(\frac{a^t}{\pi_1} - \frac{1-a^t}{\pi_0}\right)\left\{u(y^t)-E[u(Y^t)|A=a^t]\right\}.$$
Under the sharp null, it simplifies to
$$D_u(y^t,a^t,w^t) = \left(\frac{a^t}{\pi_1} - \frac{1-a^t}{\pi_0}\right)\left\{u(y^t)-E[u(Y^t)]\right\},$$
with variance under the sharp null $\textnormal{var}_P(u(Y))/(\pi_1\pi_0)$, due to the independence between $A$ and $(Y^t,W^t)$ under the sharp null.

The fully adjusted estimator we consider is the AIPW estimator given by
$$\hat\psi_a = \frac{1}{n^t}\sum_{i=1}^{n^t}\left[\frac{A_i\{u(Y_i^t) - \hat r_1(W_i^t)\}}{\hat\pi(W_i^t)} -\frac{(1-A_i)\{u(Y^t_i) - \hat r_0(W_i^t)\}}{1-\hat\pi(W_i^t)}  + \hat r_1(W_i^t)-\hat r_0(W_i^t)\right],$$
where $\hat r_1(w)$ and $\hat r_0(w)$ are estimators of $r_1(w) = E[u(Y^t)|A=1,W^t=w]$ and $r_0(w) = E[u(Y^t)|A=0,W^t=w]$. The influence function of this estimator $\hat\psi_a$ is given by
$$D^*(y^t,a^t,w^t) = \frac{\{u(y^t)-r_1(w^t)\}a^t}{\pi_1}+r_1(w^t)-\frac{\{u(y^t)-r_0(w^t)\}(1-a^t)}{\pi_0}-r_0(w^t)-\psi.$$
Under the sharp null, it simplifies to
$$D^*(y^t,a^t,w^t) =\left(\frac{a^t}{\pi_1} - \frac{1-a^t}{\pi_0}\right) \left\{y^t-E[u(Y^t)|W^t=w^t]\right\},$$
with variance under the sharp null $E_P[\textnormal{var}_P(u(Y)|W)]/(\pi_1\pi_0)$, where we again use the independence between $A$ and $(Y^t,W^t)$ under the sharp null.
\end{proof}

\begin{proof}[Proof of Theorem \ref{ATE}] First, we consider the estimation of $\sigma^2_m$. Define a parameter $P\mapsto (\alpha^*(P),\beta^*(P))$, where $(\alpha^*(P),\beta^*(P))$ is the minimizer of $E_P[(Y-\alpha-\beta^\top W)^2]$ in $(\alpha,\beta)$. Define a set of estimating functions $U(\alpha,\beta) = (y-\alpha-\beta^\top w,w(y-\alpha-\beta^\top w))$. The first-order condition implies that $(\alpha^*(P),\beta^*(P))$ is the unique solution to the estimating equation $PU(\alpha,\beta)=0$. This solution can be written as a differentiable transformation of $E_P[Y], E_P[W], E_P[YW]$ and $E_P[WW^\top ]$. Hence, by the chain rule, $(\alpha^*(P),\beta^*(P))$ is pathwise differentiable. 

In the remainder of this proof, we will often write $(\alpha^*(P),\beta^*(P))$ as $(\alpha^*,\beta^*)$. The fitted coefficients $(\hat\alpha,\hat\beta)$ solve the empirical estimating equation $P_nU(\alpha,\beta)=0$, and are asymptotically linear estimators of $(\alpha^*,\beta^*)$ with influence function
\begin{equation*}
    \textnormal{IF}_{ab}(y,w) = \begin{bmatrix} 1 & E[W^\top ] \\ E[W] & E[WW^\top ] \end{bmatrix}^{-1} \begin{bmatrix} y-\alpha^*-w^\top\beta^* \\ w(y-\alpha^*-w^\top\beta^*)\end{bmatrix}.
\end{equation*}
Lemma \ref{gradientEE} implies that this influence function is also the canonical gradient of $(\alpha^*,\beta^*)$ in any locally nonparametric model, and therefore $(\hat\alpha,\hat\beta)$ is also regular \citep[Proposition 2.3.i, ][]{pfanzagl1990estimation}. 
Define the estimating function $U_{ATE}(\alpha,\beta,\sigma^2) := (y-\alpha-\beta^\top w)^2-\sigma^2$. We note that $\sigma_m^2$ is the solution in $\sigma^2$ to the estimating equation $PU_{ATE}(\alpha^*,\beta^*,\sigma^2) = 0$ and $\hat\sigma^2_m$ solves its empirical counterpart $P_nU_{ATE}(\hat\alpha,\hat\beta,\sigma^2)=0$. Hence,
\begin{align*}
    0 &=  P_nU_{ATE}(\hat\alpha,\hat\beta,\hat\sigma_{m}^2)-PU_{ATE}(\alpha^*,\beta^*,\sigma_{m}^2) \\
    &= (P_n-P)U_{ATE}(\alpha^*,\beta^*,\sigma_{m}^2) + P\left\{U_{ATE}(\hat\alpha,\hat\beta,\hat\sigma_{m}^2)-U_{ATE}(\alpha^*,\beta^*,\sigma_{m}^2)\right\} \\
    &\ \ + (P_n-P)\left\{U_{ATE}(\hat\alpha,\hat\beta,\hat\sigma_{m}^2)-U_{ATE}(\alpha^*,\beta^*,\sigma_{m}^2)\right\}.
\end{align*}
The consistency of $(\hat\alpha,\hat\beta)$ implies that $\hat\sigma_m^2$ is also consistent. This together with the bounded support of $W$ implies that $\|U_{ATE}(\hat\alpha,\hat\beta,\hat\sigma_{m}^2)-U_{ATE}(\alpha^*,\beta^*,\sigma_{m}^2)\|_{L^2(P)} = o_P(1)$. Furthermore, we can show that $U_{ATE}$ is a Lipschitz function of $(\alpha,\beta,\sigma^2)$ in a neighborhood of $(\alpha^*,\beta^*,\sigma_{m}^2)$, and thus $U_{ATE}(\hat\alpha,\hat\beta,\hat\sigma_{m}^2)$ belongs to a Donsker class with probability tending to 1 \citep[Example 19.7,][]{van2000asymptotic}. Lemma 19.24 in \citet{van2000asymptotic} implies that the last term in the above display is $o_P(n^{-1/2})$. Applying a Taylor expansion to the second term, we have
\begin{align*}
    \hat\sigma_{m}^2 -\sigma_{m}^2 &= (P_n-P)U_{ATE}(\alpha^*,\beta^*,\sigma_{m}^2) + \left(\frac{\partial}{\partial \alpha}PU_{ATE}\right)|_{(\alpha^*,\beta^*)}(\hat\alpha-\alpha^*)\\
    &\ \ + \left(\frac{\partial}{\partial \beta}PU_{ATE}\right)^\top |_{(\alpha^*,\beta^*)}(\hat\beta-\beta^*) + o_P(n^{-1/2}).
\end{align*}
In particular, we have
\begin{equation*}
    \frac{\partial U_{ATE}}{\partial \alpha}(y,w) = -2(y-\alpha-\beta^\top w), \ \ \frac{\partial U_{ATE}}{\partial \beta}(y,w)= -2w(y-\alpha-\beta^\top w),
\end{equation*}
both of which have mean 0 at $(\alpha^*,\beta^*)$ by the first-order condition of $(\alpha^*,\beta^*)$. This implies that $\hat\sigma^2_m$ is asymptotically linear with influence function $\textnormal{IF}_m(y,w) = (y-\alpha^*-w^\top \beta^*)^2-\sigma^2_m$. Lemma \ref{gradientEE} implies that this is also the canonical gradient of $\sigma_m^2$, and hence $\hat\sigma^2_m$ is also regular.

Next we estimate $\sigma_a^2$. The proposed estimator is a one-step estimator based on the canonical gradient. Applying Lemma \ref{gradientcondvar} with $h(x,y) = y$, we obtain the canonical gradient $\textnormal{IF}_a(y,w) = \{y-r(w)\}^2-\sigma^2_a$ where $r(w) = E_P[Y|W=w]$. Let $\hat P$ be a distribution of $(Y,W)$ such that the conditional mean of $Y$ given $W$ is $\hat r(w)$ and the marginal distribution of $W$ is the empirical distribution of $W$. Then,
\begin{align*}
    \hat\sigma_a^2 - \sigma_a^2 &= \sigma_a^2(\hat P) + P_n \textnormal{IF}_{a}(\hat P)-\sigma_a^2 \\
    &= (P_n-P)\textnormal{IF}_{a}(P) + (P_n-P)\{\textnormal{IF}_{a}(\hat P)-\textnormal{IF}_{a}(P)\} + R(\hat P,P),
\end{align*}
where
\begin{align*}
    R(\hat P,P) &= \sigma^2_a(\hat P)-\sigma^2_a(P)+P\{\textnormal{IF}_{a}(\hat P)\} \\
    &= \int \{\hat r(w)-r(w)\}^2dP(w) = o_P(n^{-1/2}).
\end{align*}
As $\hat r(w)$ is uniformly bounded and belongs to a Donsker class $\mathcal{F}$ with probability tending to 1 and $Y$ has bounded support, Theorem 2.10.6 in \citet{van1996weak} implies that $\textnormal{IF}_a(\hat P) = \{y-\hat r(w)\}^2-\hat\sigma_a^2$ belongs to a Donsker class $\tilde{\mathcal{F}}$ that is also bounded. The difference between $\textnormal{IF}_a(\hat P)$ and $\textnormal{IF}_a(P)$ is given by
\begin{equation*}
    \{y-\hat r(w)\}^2 - \hat\sigma^2_a - \{y-r(w)\}^2+\sigma_a^2 = \{2y-\hat r(w)-r(w)\}\{r(w)-\hat r(w)\}-(\hat\sigma_a^2-\sigma_a^2).
\end{equation*}
Therefore,
\begin{align*}
    \|\textnormal{IF}_a(\hat P) - \textnormal{IF}_a(P)\|_{L^2(P)} &\leq \left[\int\{2y-\hat r(w)-r(w)\}^2\{r(w)-\hat r(w)\}^2 dP\right]^{1/2} + |\hat\sigma_a^2-\sigma_a^2| \\
    & \leq M \|r-\hat r\|_{L^2(P)} + |\hat\sigma_a^2-\sigma_a^2|,
\end{align*}
for some $M$ as the support of $Y$ and $\hat r$ are bounded. The first term in the last line is $o_P(1)$ by the assumption that $\|\hat r-r\|_{L^2(P)} = o_P(n^{-1/4})$. Thus, $ \|\textnormal{IF}_a(\hat P) - \textnormal{IF}_a(P)\|_{L^2(P)} = o_P(1)$ provided that $\hat\sigma_a^2$ is a consistent estimator of $\sigma_a^2$. Suppose for now that this is indeed the case, then Lemma 19.24 in \citet{van2000asymptotic} implies that $(P_n-P)\{\textnormal{IF}_{a}(\hat P)-\textnormal{IF}_{a}(P)\}$ is $o_P(n^{-1/2})$. This shows that $\hat\sigma_a^2$ is regular and asymptotically linear with influence function $\textnormal{IF}_{a}$. 

We now show that $\hat\sigma_a^2$ is indeed a consistent estimator of $\sigma_a^2$. Note that
\begin{align*}
    \hat\sigma_a^2 - \sigma_a^2 &= P_n\left[\{y-\hat r(w)\}^2-\sigma_a^2\right] \\
    &= P_n\left[\{y-r(w)\}^2-\sigma_a^2\right]+P_n\left[\{r(w)-\hat r(w)\}^2+2\{y-r(w)\}\{r(w)-\hat r(w)\}\right] \\
    &\leq P_n\left[\{y-r(w)\}^2-\sigma_a^2\right] + M_1 P_n|r(w)-\hat r(w)|,
\end{align*}
for some constant $M_1$. The first term in the last line is $o_P(1)$ by the law of large number. As for the second term, we rewrite it as
\begin{equation*}
    P_n|r(w)-\hat r(w)| = (P_n-P)|\hat r(w)- r(w)| + P|\hat r(w)- r(w)|.
\end{equation*}
Lemma 19.24 in \citet{van2000asymptotic} implies that $(P_n-P)|\hat r(w)- r(w)| =o_P(1)$ as $|\hat r(w) - r(w)|$ lies in a Donsker class with probability tending to 1 and $\|\hat r - r\|_{L^2(P)} = o_P(1)$. In addition, $P|\hat r(w)- r(w)| \leq \|\hat r-r\|_{L^2(P)} = o_P(1)$. This establishes the consistency of $\hat\sigma_a^2$.

Finally, we estimate $\sigma_u^2$ with the sample variance of $Y$, which is regular and asymptotically linear with influence function $\textnormal{IF}_u(y,w) = \{y-E_P[Y]\}^2-\sigma_u^2$. 

Theorem \ref{DIM} then follows by applying the delta method. Specifically, the influence function of $\hat\phi_a$ is given by $(\textnormal{IF}_a -\phi_a \textnormal{IF}_u)/\sigma^2_u$, and similarly the influence function of $\hat\phi_m$ is $(\textnormal{IF}_m -\phi_m \textnormal{IF}_u)/\sigma^2_u$. Both estimators are efficient as we work in a locally nonparametric model.
\end{proof}

\begin{proof}[Proof of Lemma \ref{ordinalworking}]
First recall that $\theta^*_a(k,w)$ is the best approximation to the true outcome regression $\theta_a(k,w)$ within the proportional odds model. The coefficients $\alpha_a^*$ and $\beta_a^*$ satisfy the following first-order condition
\begin{equation*}
    E\left[\frac{I\{A_i=a\}}{\pi_a}\left[I\{Y_i^t\leq k\} -\theta_{\alpha_a^*,\beta_a^*}(k,W_i^t)\right]\right] = 0,
\end{equation*}
which implies that
\begin{equation*}
   F_a(k) = E\left[\theta_a(k,W^t)\right] = E\left[\theta^*_a(k,W^t)\right].
\end{equation*}
We claim that when the treatment is randomized, the following is a gradient of $F_a(k)$,
\begin{equation*}
    \textnormal{IF}_{F_a(k)}(y^t,a^t,w^t) = \frac{I\{a^t=a\}}{\pi_a}\left[I\{y^t\leq k\}-\theta^*_a(k,w^t)\right]+\theta^*_a(k,w^t)-F_a(k).
\end{equation*}
To show this, we will show that the difference between $\textnormal{IF}_{F_a(k)}$ and the canonical gradient of $F_a(k)$ lies in the orthogonal complement of the tangent space. The canonical gradient is given by 
\begin{equation*}
    D_a^*(y^t,a^t,w^t) = \frac{I\{a^t=a\}}{\pi_a}\left[I\{y^t\leq k\}-\theta_a(k,w^t)\right]+\theta_a(k,w^t)-F_a(k).
\end{equation*}
Hence,
\begin{equation*}
    \textnormal{IF}_{F_a(k)}(Y^t,A,W^t)-D_a^*(Y^t,A,W^t) = \left\{\theta^*_a(k,W^t)-\theta_a(k,W^t)\right\}\left[1-I\{A=a\}/\pi_a\right].
\end{equation*}
As each individual factor has mean 0 and $A \perp W^t$, the difference is orthogonal to each component of the tangent space and hence the tangent space itself.

Finally, $\hat F_a$ re-writes as a one-step estimator based on the gradient $\textnormal{IF}_{F_a(k)}$. In particular, let $\hat\nu$ be a distribution with outcome regression $\theta_{\hat\alpha_a,\hat\beta_a}$ and treatment mechanism $\hat\pi_1$, then
\begin{equation*}
    \hat F_a(k) - F_a(k) 
    = (P_{\hat v}-P_\nu)\textnormal{IF}_{F_a(k)}(\nu) + (P_{\hat\nu}-P_\nu)\{\textnormal{IF}_{F_a(k)}(\hat \nu)-\textnormal{IF}_{F_a(k)}(\nu)\} + R(\hat \nu,\nu),
\end{equation*}
where
\begin{align*}
    R(\hat\nu,\nu) &=  F_a(k)(\hat\nu)-F_a(k)(\nu) + P_\nu \textnormal{IF}_{F_a(k)}(\hat\nu)\\
    &= \frac{\pi_a-\hat\pi_a}{\hat{\pi}_a}P_\nu\left\{\theta_a(k,W^t)-\theta^*_{a}(k,W^t)+\theta^*_{a}(k,W^t)-\theta_{\hat\alpha_a,\hat\beta_a}(k,W^t)\right\} \\
    &= \frac{\pi_a-\hat\pi_a}{\hat{\pi}_a}P_\nu\left\{\theta^*_{a}(k,W^t)-\theta_{\hat\alpha_a,\hat\beta_a}(k,W^t)\right\}.
\end{align*}
First we note that $R(\hat\nu,\nu)$ is $o_P(n^{-1/2})$ as the first term in the last line is $O_P(n^{-1/2})$ and the second term is $o_P(1)$. Also given the bounded support of $W^t$ and the fact that $\pi_1$ and $\pi_0$ are bounded away from 0, we can show that the convergence of $\hat\pi_a$, $\hat\alpha_a$ and $\hat\beta_a$ implies that $\|\textnormal{IF}_{F_a(k)}(\hat \nu)-\textnormal{IF}_{F_a(k)}(\nu)\|_{L_2}(\nu) = o_P(1)$. Example 19.7 in \citet{van2000asymptotic} shows that $\textnormal{IF}_{F_a(k)}(\hat v)$ lies in a Donsker class with probability tending to 1, as $\textnormal{IF}_{F_a(k)}$ is Lipschitz in its indexing parameters in a neighborhood of the true parameter value, again due to the boundedness of $W^t$ and $\pi_a$. Lemma 19.24 in \citet{van2000asymptotic} implies that $(P_{\hat\nu}-P_\nu)\{\textnormal{IF}_{F_a(k)}(\hat \nu)-\textnormal{IF}_{F_a(k)}(\nu)\} = o_P(n^{-1/2})$. 

Thus $\hat F_a(k)$ is asymptotically linear with influence function $\textnormal{IF}_{F_a(k)}$.
\end{proof}

\begin{proof}[Proof of Lemma \ref{reDIM}]
First we consider the variance of the unadjusted and fully adjusted estimators. Recall that, in proving Lemma \ref{reATE}, we considered general functions $u(\cdot)$. Although the outcome is now ordinal, the same arguments as in the proof of Lemma \ref{reATE} applies here, and we can show that $\sigma_u^2 = E_P[u(Y)]$ and $\sigma_a^2 = E_P[\textnormal{var}_P(u(Y)|W)]$.

We now derive the influence function of the adjusted estimator based on the working proportional odds model. For the ease of notation, let $b_k := u(k)-u(k+1)$. Recall that $\hat\psi_{m} = \sum_{k=1}^{K-1}b_k\{\hat F_1(k)-\hat F_0(k)\}$, and thus, by Lemma \ref{ordinalworking}, it is asymptotically linear with influence function
$$D_{m}(y^t,a^t,w^t) = \sum_{k=1}^{K-1}b_k\left\{\textnormal{IF}_{F_1(k)}(y^t,a^t,w^t)-\textnormal{IF}_{F_0(k)}(y^t,a^t,w^t)\right\}-\psi.$$
Under the sharp null, the above display simplifies to 
\begin{equation*}
  D_{m}(y^t,a^t,w^t) = \left(\frac{a^t}{\pi_1}-\frac{1-a^t}{\pi_0}\right)\sum_{k=1}^{K-1}b_k\left[I\{y^t\leq k\}-\theta^*(k,w^t)\right].
\end{equation*}
Under the sharp null, its variance is $E_P\left[\left(\sum_{k=1}^{K-1}b_k\left[I\{Y\leq k\}-\theta^*(k,W)\right]\right)^2\right] / (\pi_1\pi_0)$ due to the independence between $A$ and $(Y^t,W^t)$ under the sharp null.
\end{proof}

\begin{proof}[Proof of Theorem \ref{DIM}]
We first consider estimating $\sigma_m^2$. To start, we show that $(\hat\alpha,\hat\beta)$, the maximizer of the empirical version of \eqref{propoddsEE}, is a RAL estimator of $(\alpha^*,\beta^*)$. The first-order conditions of this maximization imply that $(\hat\alpha,\hat\beta)$ solves a set of estimating equations, as the parameter space is unconstrained. Specifically, define $U(\alpha,\beta) = \left(U_1(\alpha,\beta), \ldots, U_{K-1}(\alpha,\beta), U_K(\alpha,\beta)\right)$ with 
\begin{align*}
    U_k(\alpha,\beta)(y,w) &= I\{y\leq k\}-\frac{\exp(\alpha_k+\beta^\top  w)}{1+\exp(\alpha_k+\beta^\top  w)}, \ k=1,\ldots,K-1,  \\
    U_K(\alpha,\beta)(y,w) &= w\sum_{k=1}^{K-1}\left[I\{y\leq k\}-\frac{\exp(\alpha_k+\beta^\top  w)}{1+\exp(\alpha_k+\beta^\top  w)}\right].
\end{align*}
Then we have that $P_nU_k(\hat\alpha,\hat\beta) = 0$ for all $k=1,\ldots,K$. We can show, through the usual arguments used to study estimating equations \citep[Chapter 5,][]{van2000asymptotic}, that the influence function of $(\hat\alpha,\hat\beta)$ is
\begin{equation*}
    \textnormal{IF}_{ab}(y,w) = -\{P\dot U(\alpha^*,\beta^*)\}^{-1}U(\alpha^*,\beta^*)(y,w).
\end{equation*}
In particular, the derivative matrix $\dot U$ can be partitioned into $\begin{bmatrix}\dot U_{\alpha\alpha} & \dot U_{\alpha\beta}^\top  \\ \dot U_{\alpha\beta} & \dot U_{\beta\beta}\end{bmatrix}$, with
\begin{align*}
    \dot U_{\alpha\alpha}(y,w) &= -\text{Diag}\left[\theta^*(k,w)\{1-\theta^*(k,w)\}\right], \\
    \dot U_{\alpha\beta}(y,w) &= -\left[w\theta^*(1,w)\{1-\theta^*(1,w)\},\cdots,w\theta^*(K-1,w)\{1-\theta^*(K-1,w)\}\right]^\top , \\
    \dot U_{\beta\beta}(y,w) &= -\sum_{k=1}^{K-1}ww^\top  \theta^*(k,w)\{1-\theta^*(k,w)\}.
\end{align*}
Lemma~\ref{gradientEE} implies that $\textnormal{IF}_{ab}$ is the canonical gradient of $(\alpha^*,\beta^*)$ in a locally nonparametric model, and thus $(\hat\alpha,\hat\beta)$ is also regular \citep[Proposition 2.3.i,][]{pfanzagl1990estimation}.

Next, we define the following estimating equation:
\begin{equation*}
    U_{DIM}(\alpha,\beta,\sigma^2)(y,w) = \left(\sum_{k=1}^{K-1}b_k\left[I\{y \leq k\}-\theta_{\alpha,\beta}(k,w)\right]\right)^2 -\sigma^2.
\end{equation*}
By definition, $\sigma^2_{m}$ is the unique solution in $\sigma^2$ to the equation $PU_{DIM}(\alpha^*,\beta^*,\sigma^2) = 0$, and $\hat\sigma^2_{m}$ solves its empirical counterpart $P_nU_{DIM}(\hat\alpha,\hat\beta,\sigma^2)=0$. Hence, we have
\begin{align*}
    0 &=  P_nU_{DIM}(\hat\alpha,\hat\beta,\hat\sigma_{m}^2)-PU_{DIM}(\alpha^*,\beta^*,\sigma_{m}^2) \\
    &= (P_n-P)U_{DIM}(\alpha^*,\beta^*,\sigma_{m}^2) + P\left\{U_{DIM}(\hat\alpha,\hat\beta,\hat\sigma_{m}^2)-U_{DIM}(\alpha^*,\beta^*,\sigma_{m}^2)\right\} \\
    &\ \ + (P_n-P)\left\{U_{DIM}(\hat\alpha,\hat\beta,\hat\sigma_{m}^2)-U_{DIM}(\alpha^*,\beta^*,\sigma_{m}^2)\right\}.
\end{align*}
The consistency of $(\hat\alpha,\hat\beta)$ implies that $\hat\sigma^2_m$ is also consistent. Combining this with the bounded support of $W$, we see that $\|U_{DIM}(\hat\alpha,\hat\beta,\hat\sigma_{m}^2)-U_{DIM}(\alpha^*,\beta^*,\sigma_{m}^2)\|_{L^2(P)} = o_P(1)$. Furthermore, it can be shown that $U_{DIM}$ is a Lipschitz transformation of $(\alpha,\beta,\sigma^2)$ in a neighborhood of $(\alpha^*,\beta^*,\sigma_{m}^2)$, and thus that $U_{DIM}(\hat\alpha,\hat\beta,\hat\sigma_{m}^2)$ belongs to a Donsker class with probability tending to 1 \citep[Example 19.7,][]{van2000asymptotic}. Lemma 19.24 in \citet{van2000asymptotic} implies that the last term in the above display is $o_P(n^{-1/2})$. Thus,
\begin{align*}
    \hat\sigma_{m}^2 -\sigma_{m}^2 &= (P_n-P)U_{DIM}(\alpha^*,\beta^*,\sigma_{m}^2) + \left.\left(\frac{\partial}{\partial \alpha}PU_{DIM}\right)\right|_{(\alpha^*,\beta^*)}(\hat\alpha-\alpha^*)\\
    &\ \ + \left.\left(\frac{\partial}{\partial \beta}PU_{DIM}\right)\right|_{(\alpha^*,\beta^*)}(\hat\beta-\beta^*) + o_P(n^{-1/2}).
\end{align*}
The partial derivatives are given by
\begin{align*}
    \frac{\partial}{\partial \alpha_k}U_{DIM}|_{(\alpha^*,\beta^*)}(y,w) &= -2b_k\left(\sum_{l=1}^{K-1}b_k[I\{y \leq l\}-\theta^*(l,w)]\right)\theta^*(k,w)\{1-\theta^*(k,w)\}, \\
    \frac{\partial}{\partial \beta}U_{DIM}|_{(\alpha^*,\beta^*)}(y,w) &= -2\left(\sum_{k=1}^{K-1}b_k[I\{y \leq k\}-\theta^*(k,w)]\right)\left[\sum_{k=1}^{K-1}b_k w \theta^*(k,w)\{1-\theta^*(k,w)\}\right].
\end{align*}

Combining all the results above, we see that $\hat\sigma^2_m$ is an asymptotically linear estimator of $\sigma_m^2$ with influence function $\textnormal{IF}_m$, where
\begin{align*}
    \textnormal{IF}_m(y,w) &= U_{DIM}(\alpha^*,\beta^*,\sigma_{m}^2)(y,w) + \left\{\frac{\partial PU_{DIM}}{\partial(\alpha,\beta)}|_{(\alpha^*,\beta^*)}\right\}^\top  \textnormal{IF}_{ab}(y,w).
\end{align*}
Lemma~\ref{gradientEE} implies that $\textnormal{IF}_m$ is the canonical gradient of $\sigma^2_m$ in any locally nonparametric model, and thus that $\hat\sigma_m^2$ is regular \citep[Proposition 2.3.i,][]{pfanzagl1990estimation}.

We now consider the estimation of $\sigma_a^2$. The proposed estimator is a one-step estimator based on the canonical gradient, and the proof is very similar to that of Theorem \ref{ATE}. Applying Lemma \ref{gradientcondvar} with $h(x,y) = u(y)$, we obtain the canonical gradient $\textnormal{IF}_a(y,w) = \{u(y)-r(w)\}^2-\sigma^2_a$. 
Let $\hat P$ be a distribution of $(Y,W)$ such that the conditional mean of $u(Y)$ given $W$ is $\hat r(w)$, and that the marginal distribution of $W$ is the empirical distribution of $W$. We have that
\begin{align*}
    \hat\sigma_a^2 - \sigma_a^2 &= \sigma_a^2(\hat P) + P_n \textnormal{IF}_{a}(\hat P)-\sigma_a^2 \\
    &= (P_n-P)\textnormal{IF}_{a}(P) + (P_n-P)\{\textnormal{IF}_{a}(\hat P)-\textnormal{IF}_{a}(P)\} + R(\hat P,P),
\end{align*}
where
\begin{align*}
    R(\hat P,P) &= \sigma^2_a(\hat P)-\sigma^2_a(P)+P\{\textnormal{IF}_{a}(\hat P)\} \\
    &= \int \{\hat r(w)-r(w)\}^2dP(w) = o_P(n^{-1/2}),
\end{align*}
where the latter equality holds by assumption. 
As $\hat r(w)$ is uniformly bounded and belongs to a Donsker class $\mathcal{F}$ with probability tending to 1 and $Y$ has bounded support, Theorem 2.10.6 in \citet{van1996weak} implies that $\textnormal{IF}_a(\hat P) = \{u(y)-\hat r(w)\}^2-\hat\sigma_a^2$ belongs to a transformed Donsker class $\tilde{\mathcal{F}}$ that is also bounded. The difference between $\textnormal{IF}_a(\hat P)$ and $\textnormal{IF}_a(P)$ is given by
\begin{equation*}
    \{u(y)-\hat r(w)\}^2 - \hat\sigma^2_a - \{u(y)-r(w)\}^2+\sigma_a^2 = \{2u(y)-\hat r(w)-r(w)\}\{r(w)-\hat r(w)\}-(\hat\sigma_a^2-\sigma_a^2).
\end{equation*}
Therefore,
\begin{align*}
    \|\textnormal{IF}_a(\hat P) - \textnormal{IF}_a(P)\|_{L^2(P)} &\leq \left[\int\{2u(y)-\hat r(w)-r(w)\}^2\{r(w)-\hat r(w)\}^2 dP\right]^{1/2} + |\hat\sigma_a^2-\sigma_a^2| \\
    & \leq M \|r-\hat r\|_{L^2(P)} + |\hat\sigma_a^2-\sigma_a^2|,
\end{align*}
for some $M$ as the support of $Y$ and $\hat r$ are bounded. The first term in the last line is $o_P(1)$ by the assumption that $\|\hat r-r\|_{L^2(P)} = o_P(n^{-1/4})$. Thus, $ \|\textnormal{IF}_a(\hat P) - \textnormal{IF}_a(P)\|_{L^2(P)} = o_P(1)$ provided that $\hat\sigma_a^2$ is a consistent estimator of $\sigma_a^2$. Suppose for now that this is indeed the case, then Lemma 19.24 in \citet{van2000asymptotic} implies that $(P_n-P)\{\textnormal{IF}_{a}(\hat P)-\textnormal{IF}_{a}(P)\}$ is $o_P(n^{-1/2})$. Hence, if we show that $\hat\sigma_a^2$ is a consistent estimator of $\sigma_a^2$, then we will have shown that $\hat\sigma_a^2$ is regular and asymptotically linear with influence function $\textnormal{IF}_{a}$. 

We now show that $\hat\sigma_a^2$ is indeed a consistent estimator of $\sigma_a^2$. Note that
\begin{align*}
    \hat\sigma_a^2 - \sigma_a^2 &= P_n\left[\{u(Y)-\hat r(W)\}^2-\sigma_a^2\right] \\
    &= P_n\left[\{u(Y)-r(W)\}^2-\sigma_a^2\right]+P_n\left[\{r(W)-\hat r(W)\}^2+2\{u(Y)-r(W)\}\{r(W)-\hat r(W)\}\right] \\
    &\leq P_n\left[\{u(Y)-r(W)\}^2-\sigma_a^2\right] + M_1 P_n|r(W)-\hat r(W)|,
\end{align*}
for some constant $M_1$, where we used the fact that both $Y$ and $W$ have bounded support and $\hat r$ is uniformly bounded. The first term in the last line is $o_P(1)$ by the weak law of large numbers. As for the second term, we see that it is equal to
\begin{equation*}
    P_n|r(W)-\hat r(W)| = (P_n-P)|\hat r(W)- r(W)| + P|\hat r(W)- r(W)|.
\end{equation*}
Lemma 19.24 in \citet{van2000asymptotic} implies that $(P_n-P)|\hat r- r| =o_P(1)$ as $|\hat r - r|$ lies in a Donsker class with probability tending to 1 and $\|\hat r - r\|_{L^2(P)} = o_P(1)$. In addition, $P|\hat r- r| \leq \|\hat r-r\|_{L^2(P)} = o_P(1)$. This establishes the consistency of $\hat\sigma_a^2$.

Finally, we estimate $\sigma_u^2$ with the sample variance of $u(Y)$, which is regular and asymptotically linear with influence function $\textnormal{IF}_u = \{u(y)-E[u(Y)]\}^2-\sigma_u^2$. 

Theorem \ref{DIM} then follows by applying the delta method. Specifically, the influence function of $\hat\phi_a$ is given by $(\textnormal{IF}_a -\phi_a \textnormal{IF}_u)/\sigma^2_u$, and, similarly, the influence function of $\hat\phi_m$ is $(\textnormal{IF}_m -\phi_m \textnormal{IF}_u)/\sigma^2_u$.
\end{proof}

\begin{proof}[Proof of Lemma \ref{reMW}]
Recall that the unadjusted estimator is $$\hat\psi_u = \left\{\sum_{i=1}^{n^t}\sum_{j=1}^{n^t}A_i(1-A_j)h(Y_i^t,Y_j^t)\right\}\bigg/ \left\{(\sum_{i=1}^{n^t}A_i)(n^t-\sum_{j=1}^{n^t}A_j)\right\}.$$
First we introduce some notation. Let $\nu_n$ be the empirical distribution of $X^t$ in the future trial data. Define the functions $\eta_a(\cdot)$, $a\in\{0,1\}$, analogously to $\eta(\cdot)$ but within each treatment arm as $\eta_a(k) := P_a(Y<k) + P_a(Y=k)/2$ for $a\in\{0,1\}$, and define $\hat\psi_{u1} := \sum_{i=1}^{n^t}\sum_{j=1}^{n^t}A_i(1-A_j)h(Y_i^t,Y_j^t)/(n^t)^2$. We note that $\hat\psi_{u1}$ is a V-statistic with symmetric kernel $\tilde h(X_1^t,X_2^t) := \{A_1(1-A_2)h(Y_1^t,Y_2^t)+A_2(1-A_1)h(Y_2^t,Y_1^t)\}/2$. For a generic distribution $Q$ of $X^t$, we define $Q^2\tilde h := \int\int\tilde h(x_i^t,x_j^t) dQ(x_i^t)dQ(x_j^t)$. With this notation, $\hat\psi_{u1}=\nu_n^2\tilde h$. Note that
\begin{align*}
    \hat\psi_{u1} -\pi_1\pi_0\psi &= \nu_n^2\tilde h - \nu^2\tilde h = 2(\nu_n-\nu)(\nu\tilde h)+(\nu_n-\nu)^2 \tilde h.
\end{align*}
To establish the asymptotic linearity of $\hat\psi_{u1}$, we first show that $(\nu_n-\nu)^2 \tilde h$ is $o_P(n^{-1/2})$. To start, define a class of functions $\tilde{\mathcal{H}} := \{x \mapsto \tilde h(x,x_2): x_2 \in \mathcal{X}\}$ where $\mathcal{X}$ is the support of $X^t$. Each function in $\tilde{\mathcal{H}}$ is a weighted sum of 4 binary terms, with weights being either 1 or $1/2$. Each term is indexed by $a_2$ and $y_2^t$, and can be computed with 3 arithmetic operations and 1 comparison. Theorem 8.4 in \citet{anthony2009neural} implies that each binary term belongs to a VC-class with VC dimension at most 64, and Lemma 19.15 in \citet{van2000asymptotic} in turn implies that this class is Donsker. Theorem 2.10.6 in \citet{van1996weak} then implies that $\tilde{\mathcal{H}}$ is a Donsker class (hence also Glivenko-Cantelli). Define $\tilde h_{1n}: x \mapsto \int \tilde h(x_1,x)d(\nu_n-\nu)(x_1)$, then $\nu \tilde h_{1n}^2 \leq \{\sup_{x \in \mathcal{X}}|\tilde h_{1n}(x)|\}^2 = o_P(1)$. Next, note that $(\nu_n-\nu)^2 \tilde h = (\nu_n-\nu) \tilde h_{1n}$. The function $x \mapsto \int \tilde h(x,x_2)d\nu_n (x_2)$ is in the closure of the convex hull of the Donsker class $\tilde{\mathcal{H}}$, and $x \mapsto \int \tilde h(x,x_2)d\nu (x_2)$ is a fixed function. This together with the symmetry of $\tilde h$ implies that $\tilde h_{1n}$ lies in a Donsker class. Lemma 19.24 then implies that $(\nu_n-\nu)^2 \tilde h = (\nu_n-\nu) \tilde h_{1n} = o_P(n^{-1/2})$. 

Next we note that $\nu \tilde h(x^t,\cdot): x^t \mapsto [a\pi_0\eta_0(y^t) + (1-a)\pi_1\{1-\eta_1(y^t)\}]/2$. Combining this with the previous results, we have that  
\begin{equation*}
    \hat\psi_{u1} -\pi_1\pi_0\psi = \frac{1}{n^t}\sum_{i=1}^{n^t}\left[A_i\pi_0\eta_0(Y_i^t) 
    +(1-A_i)\pi_1\{1-\eta_1(Y_i^t)\}-2\pi_1\pi_0\psi\right] + o_P(n^{-1/2}).
\end{equation*}
Applying the delta method, we see that $\hat\psi_u$ is asymptotically linear with influence function
\begin{equation*}
    D_u(y^t,a^t,w^t) = \frac{a^t\{\eta_0(y^t)-\psi\}}{\pi_1}+\frac{(1-a^t)\{1-\eta_1(y^t)-\psi\}}{\pi_0},
\end{equation*}
which, under the sharp null, simplifies to $(\frac{a^t}{\pi_1}-\frac{1-a^t}{\pi_0})\{\eta(y^t)-\frac{1}{2}\}$. The variance of $\eta(Y^t)$ can be calculated exactly, and equals to $(1-\sum_{k=1}^{K}p_k^3)/12$.

Next we look at the fully adjusted estimator. The efficient influence function of $\psi$ was given in, for example, \citet{mao2018causal},
\begin{align*}
    D^*(y^t,a^t,w^t) &= 1-2\psi + \frac{a^t}{\pi_1}\left\{\eta_0(y^t)-E\left[\eta_0(Y^t)|A=1,W^t=w^t\right]\right\}\\
    &-\frac{1-a^t}{\pi_0}\left\{\eta_1(y^t)-E\left[\eta_1(Y^t)|A=0,W^t=w^t\right]\right\} \\
    &+ E\left[\eta_0(Y^t)|A=1,W^t=w^t\right]-E\left[\eta_1(Y^t)|A=0,W^t=w^t\right],
\end{align*}
which, under the sharp null, simplifies to $\left(\frac{a^t}{\pi}-\frac{1-a^t}{\pi_0}\right)\left\{\eta(y^t)-E\left[\eta(Y^t)|W^t=w^t\right]\right\}$ with variance $\sigma_a^2/(\pi_1\pi_0)$, where we use the independence between $A$ and $(Y^t,W^t)$ under the sharp null.

Finally we consider the estimator based on proportional odds model. Let $\hat{P}_a$ be a distribution with CDF $\hat F_a(k)$ for $a=0,1$. Recall that $\hat\psi_{m} = \int\int h(x,y)d\hat{P}_1(x)d\hat{P}_0(y)$. Since $\hat F$'s are asymptotically linear, we have
\begin{align*}
    \hat\psi_{m} - \psi &= \int\int h(x,y)dP_1(x)d(\hat{P}_0-P_0)(y)+\int\int h(x,y) dP_0(y)d(\hat{P}_1-P_1)(x) + o_P(n^{-1/2}) \\
    &= \int\{1-\eta_1(y)\}d(\hat{P}_0-P_0)(y)+\int \eta_0(x)d(\hat{P}_1-P_1)(x) + o_P(n^{-1/2}).
\end{align*}
The first term can be alternatively written as $\sum_{k=1}^{K-1}-b_1(k)\{\hat{F}_0(k)-F_0(k)\}$ where $b_1(k) = \eta_1(k)-\eta_1(k+1)$, similarly for the second term. Then Lemma \ref{ordinalworking} implies that $\hat\psi_{m}$ is asymptotically linear with influence function
\begin{equation*}
    D_{m}(y^t,a^t,w^t)= \sum_{k=1}^{K-1}\left\{-b_1(k)\textnormal{IF}_{F_0(k)}(y^t,a^t,w^t)
    + b_0(k)\textnormal{IF}_{F_1(k)}(y^t,a^t,w^t)\right\},
\end{equation*}
which, under the sharp null, simplifies to
\begin{equation*}
    \left(\frac{a^t}{\pi_1}-\frac{1-a^t}{\pi_0}\right)\sum_{k=1}^{K-1}\{\eta(k)-\eta(k+1)\}\left[I\{y^t\leq k\}-\theta^*(k,w^t)\right].
\end{equation*}
The variance of $D_m$ under the sharp null is $\sigma^2_m/(\pi_1\pi_0)$, where we used the independence between $A$ and $(Y^t,W^t)$ under the sharp null. Lemma \ref{reMW} follows by the definition of relative efficiency.
\end{proof}

\begin{proof}[Proof of Theorem \ref{MW}]
We first study $\hat\sigma^2_{m}$ based on estimating equations. The proof is similar to that of Theorem \ref{DIM} except that we need to estimate the marginal distribution of $Y$ in addition. We use the sample proportion $\hat p_k$, which is asymptotically linear with influence function $\textnormal{IF}_{p_k}(y,w) = I\{y=k\}-p_k$.

Consider the following estimating equation
\begin{equation*}
    U_{MW}(\alpha,\beta,\boldsymbol{\tilde p},\sigma^2)(y,w) = \left(\sum_{k=1}^{K-1}-\frac{1}{2}(\tilde p_k+\tilde p_{k+1})[I\{y\leq k\}-\theta_{\alpha,\beta}(k,w)]\right)^2-\sigma^2.
\end{equation*}
By definition, $\sigma^2_{m}$ is the unique solution to the equation $PU_{MW}(\alpha^*,\beta^*,\boldsymbol{p},\sigma^2)=0$, and $\hat\sigma^2_{m}$ solves its empirical counterpart. Hence, we have
\begin{align*}
    0 &=  P_nU_{MW}(\hat\alpha,\hat\beta,\boldsymbol{\hat p},\hat\sigma_{m}^2)-PU_{MW}(\alpha^*,\beta^*,\boldsymbol{p},\sigma_{m}^2) \\
    &= (P_n-P)U_{MW}(\alpha^*,\beta^*,\boldsymbol{p},\sigma_{m}^2) + P\left\{U_{MW}(\hat\alpha,\hat\beta,\boldsymbol{\hat p},\hat\sigma_{m}^2)-U_{MW}(\alpha^*,\beta^*,\boldsymbol{p},\sigma_{m}^2)\right\} \\
    &\ \ + (P_n-P)\left\{U_{MW}(\hat\alpha,\hat\beta,\boldsymbol{\hat p},\hat\sigma_{m}^2)-U_{MW}(\alpha^*,\beta^*,\boldsymbol{p},\sigma_{m}^2)\right\}.
\end{align*}
Consistency of $(\hat\alpha,\hat\beta)$ and $\boldsymbol{\hat p}$ implies that $\hat\sigma^2_m$ is also consistent. This together with the bounded support of $W$ implies that $\|U_{MW}(\hat\alpha,\hat\beta,\boldsymbol{\hat p},\hat\sigma_{m}^2)-U_{MW}(\alpha^*,\beta^*,\boldsymbol{p},\sigma_{m}^2)\|_{L^2(P)} = o_P(1)$. Furthermore, we can show that $U_{MW}$ is a Lipschitz function of $(\alpha,\beta,\boldsymbol{\tilde p},\sigma^2)$ in a neighborhood of $(\alpha^*,\beta^*,\boldsymbol{p},\sigma_{m}^2)$, and thus $U_{MW}(\hat\alpha,\hat\beta,\boldsymbol{\hat p},\hat\sigma_{m}^2)$ belongs to a Donsker class with probability tending to 1 \citep[Example 19.7,][]{van2000asymptotic}. Lemma 19.24 in \citet{van2000asymptotic} implies that the last term in the above display is $o_P(n^{-1/2})$. 

Applying a Taylor expansion, we have that
\begin{align*}
    \hat\sigma^2_{m} - \sigma^2_{m} &= (P_n-P)U_{MW}(\alpha^*,\beta^*,\boldsymbol{p},\sigma^2_{m})+\left(\frac{\partial}{\partial\alpha}PU_{MW}\right)\rvert_{(\alpha^*,\beta^*,\boldsymbol{p})}(\hat\alpha-\alpha^*) \\
    &+ \left(\frac{\partial}{\partial \beta}PU_{MW}\right)\rvert_{(\alpha^*,\beta^*,\boldsymbol{p})}(\hat\beta-\beta^*) + \left(\frac{\partial}{\partial \boldsymbol{\tilde p}}PU_{MW}\right)\rvert_{(\alpha^*,\beta^*,\boldsymbol{p})}^\top (\boldsymbol{\hat p}-\boldsymbol{p}) + o_P(n^{-1/2}).
\end{align*}
Note that
\begin{align*}
    \frac{\partial}{\partial \alpha_k}U_{MW}|_{(\alpha^*,\beta^*,\boldsymbol{p})}(y,w) &= -\frac{1}{2}\left\{\sum_{l=1}^{K-1}(p_l+p_{l+1})\left[I\{y\leq l\}-\theta^*(l,w)\right]\right\} \\
    &\quad\quad\quad\times(p_k+p_{k+1})\theta^*(k,w)\left\{1-\theta^*(k,w)\right\} \\
    \frac{\partial}{\partial \beta}U_{MW}|_{(\alpha^*,\beta^*,\boldsymbol{p})}(y,w) &= -\frac{1}{2}\left\{\sum_{l=1}^{K-1}(p_l+p_{l+1})\left[I\{y\leq l\}-\theta^*(l,w)\right]\right\} \\
    &\quad\quad\quad\times\left[\sum_{k=1}^{K-1}(p_k+p_{k+1})W\theta^*(k,w)\left\{1-\theta^*(k,w)\right\}\right] \\
    \frac{\partial}{\partial \tilde p_k}U_{MW}|_{(\alpha^*,\beta^*,\boldsymbol{p})}(y,w) &= \frac{1}{2}\left\{\sum_{l=1}^{K-1}(p_l+p_{l+1})\left[I\{y\leq l\}-\theta^*(l,w)\right]\right\} \\
    &\quad\quad\times\left[I\{y\leq k\}-\theta^*(k,w)+I\{Y\leq k-1\}-\theta^*(k-1,w)\right].
\end{align*}
Hence, $\hat\sigma_m^2$ satisfies $\hat\sigma_{m}^2 - \sigma_{m}^2 = (P_n-P)\textnormal{IF}_{m} + o_P(n^{-1/2})$, where
\begin{align*}
    \textnormal{IF}_m(y,w) &= U_{MW}(\alpha^*,\beta^*,\boldsymbol{p},\sigma_{m}^2)(y,w)+\left\{\frac{\partial PU_{MW}}{\partial (\alpha,\beta)}|_{(\alpha^*,\beta^*,\boldsymbol{p})}\right\}^\top  \textnormal{IF}_{ab}(y,w) \nonumber \\
    &\quad+\sum_{k=1}^{K}\frac{\partial PU_{MW}}{\partial \tilde p_k}|_{(\alpha^*,\beta^*,\boldsymbol{p})}\textnormal{IF}_{p_k}(y,w).
\end{align*}
We note that using similar arguments as Lemma \ref{gradientEE}, we can show that $\textnormal{IF}_m$ is the canonical gradient of $\sigma^2_m$, and therefore $\hat\sigma^2_m$ is a regular estimator \citep[Proposition 2.3.i,][]{pfanzagl1990estimation}.

We estimate the unadjusted variance by plugging in $\hat p_k$. By the delta method, $\hat\sigma^2_u$ is asymptotically linear with influence function $\textnormal{IF}_u = -\sum_{k=1}^Kp_k^2 \textnormal{IF}_{p_k}/4$.

Next we establish the asymptotic linearity of $\hat\sigma^2_a$. Applying Lemma \ref{gradientcondvar} with $h(\tilde y,y) = I\{\tilde y < y\} + \frac{1}{2}\{\tilde y=y\}$, we obtain the EIF of $\sigma^2_a$ as follows
\begin{equation*}
    \textnormal{IF}_a(y,w) = \left\{\eta(y)-r(w)\right\}^2 + 2\int \left\{\eta(\tilde y)-r(\tilde w)\right\}h(y,\tilde y)dP(\tilde y,\tilde w) - 3\sigma^2_a.
\end{equation*}
We show through direct linearization that $\hat\sigma^2_a$ is indeed asymptotically linear with influence function $\textnormal{IF}_a$. To start, we note that $\hat\sigma_a^2 = P_n(\hat\eta -\hat r)^2$ and $\sigma_a^2 = P(\eta-r)^2$. Thus,
\begin{align*}
    \hat\sigma^2_a - \sigma^2_a &= P_n(\hat\eta -\hat r)^2 - P(\eta-r)^2 \\
    &= \underbrace{P_n(\hat\eta -\hat r)^2-P_n(\eta -\hat r)^2}_{\text{term 1}} + \underbrace{P_n(\eta - \hat r)^2-P(\eta -r)^2}_{\text{term 2}}.
\end{align*}
We analyze term 2 first. 
\begin{align}
    &P_n(\eta - \hat r)^2-P(\eta -r)^2  \label{eq:etarterm2}\\
    &\quad= (P_n-P)(\eta-r)^2+(P_n-P)\left\{(\eta-\hat r)^2-(\eta-r)^2\right\} + P\left\{(\eta-\hat r)^2-(\eta -r)^2\right\} \nonumber
\end{align}
The first term on the right-hand side is already linear and contributes the first term to the influence function. By the assumption that $\|\hat r-r\|_{L^2(P_W)} = o_P(n^{-1/4})$, the third term is negligible because
\begin{align*}
    \left|P\left\{(\eta-\hat r)^2-(\eta -r)^2\right\}\right| &=  \left|P\left\{2(\eta-r)(r-\hat r)+(r-\hat r)^2\right\}\right| = P(r-\hat r)^2 = o_P(n^{-1/2}).
\end{align*}
We now turn to the second term on the right-hand side of \eqref{eq:etarterm2}. By Theorem 2.10.6 in \citet{van1996weak}, the fact that $\hat r$ belongs to a bounded $P$-Donsker class with probability tending to 1 implies that $(\eta-\hat r)^2-(\eta-r)^2$ also belongs to a $P$-Donsker class with probability tending to 1, as $\eta$ and $r$ are both fixed and bounded functions. Furthermore,
\begin{align*}
    \|(\eta-\hat r)^2-(\eta-r)^2\|_{L^2(P)} &= \|(2\eta-\hat r-r)(\hat r-r)\|_{L^2(P)} \leq M\|\hat r-r\|_{L^2(P)} = o_P(1),
\end{align*}
for some constant $M$. Lemma 19.24 in \citet{van2000asymptotic} implies that the second term is also $o_P(n^{-1/2})$.

We now analyze term 1. Note that
\begin{align*}
    \textnormal{term 1} &= (P_n-P)\left\{(\hat\eta -\hat r)^2-(\eta -\hat r)^2\right\}+P\left\{(\hat\eta -\hat r)^2-(\eta -\hat r)^2\right\} \\
    &= (P_n-P)\left\{(\hat\eta -\hat r)^2-(\eta -\hat r)^2\right\} + P\left\{(\hat\eta+\eta -2\hat r)(\hat\eta -\eta)\right\} \\
    &= P\left\{(2\eta -2r)(\hat\eta -\eta)\right\}+P\left\{(\hat\eta-\eta +2r-2\hat r)(\hat\eta -\eta)\right\}\\
    &\quad +(P_n-P)\left\{(\hat\eta -\hat r)^2-(\eta -\hat r)^2\right\}.
\end{align*}
We note that $\hat\eta$ belongs to a bounded $P$-Donsker class, as the empirical distribution function lies in the closure of the convex hull of a $P$-Donsker class. Theorem 2.10.6 in \citet{van1996weak} again implies that $(\hat\eta -\hat r)^2-(\eta -\hat r)^2$ belongs to a $P$-Donsker class with probability tending to 1. Moreover,
\begin{align*}
    \|(\hat\eta -\hat r)^2-(\eta -\hat r)^2\|_{L^2(P)} &= \|(\hat\eta -\eta)(\hat\eta+\eta -2\hat r)\|_{L^2(P)} \leq M\|\hat\eta -\eta\|_{L^2(P)} = o_P(1).
\end{align*}
Thus the third term is $o_P(n^{-1/2})$. The second term is also $o_P(n^{-1/2})$ by our assumption on the convergence rate of $\hat r$ and the convergence of $\hat \eta$. Finally, the first term is the linear term that contributes to the rest of $\textnormal{IF}_a$. To see this, we write it in the integral form.
\begin{align*}
    &\int 2\{\eta(\tilde y)-r(\tilde w)\}\{\hat\eta(\tilde y)-\eta(\tilde y)\} dP(\tilde y,\tilde w) \\
    &= \int 2\{\eta(\tilde y)-r(\tilde w)\}\left\{\int h(y,\tilde y)d(P_n-P)(y)\right\} dP(\tilde y,\tilde w) \\
    &= 2(P_n-P)\int\{\eta(\tilde y)-r(\tilde w)\}h(\cdot,\tilde y) dP(\tilde y,\tilde w).
\end{align*}
Note that Lemma~\ref{gradientcondvar} implies that $P\int\{\eta(\tilde y)-r(\tilde w)\}h(\cdot,\tilde y) dP(\tilde y,\tilde w) = \sigma_a^2$.

Theorem \ref{MW} now follows by applying the delta method. Specifically, the influence function of $\hat\phi_a$ is given by $(\textnormal{IF}_a -\phi_a \textnormal{IF}_u)/\sigma^2_u$, and similarly the influence function of $\hat\phi_m$ is $(\textnormal{IF}_m -\phi_m \textnormal{IF}_u)/\sigma^2_u$. These estimators are RAL in any locally nonparametric model.
\end{proof}

\begin{proof}[Proof of Lemma \ref{reLOR}]
Recall that the unadjusted estimator is given by
$$\hat\psi_u = \frac{1}{K-1}\sum_{k=1}^{K-1} \left\{\text{logit }\tilde F_1(k)-\text{logit }\tilde F_0(k)\right\}, \ \ \tilde F_a(k) = \frac{\sum_{i=1}^{n^t}I\{Y_i^t \leq k,A_i=a\}}{\sum_{i=1}^{n^t}I\{A_i=a\}}.$$
Applying the delta method shows that its influence function is given by
\begin{equation*}
    D_u(y^t,a^t,w^t) = \frac{1}{K-1}\sum_{k=1}^{K-1}\left(\frac{a^t\left[I\{y^t\leq k\}-F_1(k)\right]}{\pi_1F_1(k)\{1-F_1(k)\}}-\frac{(1-a^t)\left[I\{y^t \leq k\}-F_0(k)\right]}{\pi_0F_0(k)\{1-F_0(k)\}}\right),
\end{equation*}
which, under the sharp null, simplifies to
\begin{equation*}
    \frac{1}{K-1}\sum_{k=1}^{K-1} \left(\frac{a^t}{\pi_1}-\frac{1-a^t}{\pi_0}\right)\left[\frac{I\{y^t\leq k\}-F(k)}{F(k)\{1-F(k)\}}\right].
\end{equation*}
Due to the independence of $A$ and $(Y^t,W^t)$ under the sharp null, the variance of $D_u(Y^t,A,W^t)$ under the sharp null is $\sigma_u^2/(\pi_1\pi_0)$.

The working-model-based adjusted estimator $\hat\psi_{m}$ replaces $\tilde F_a$ with $\hat F_a$. By Lemma \ref{ordinalworking} and the delta method, the influence function of $\hat\psi_{m}$ is
\begin{equation*}
    D_{m}(y^t,a^t,w^t) = \frac{1}{K-1}\sum_{k=1}^{K-1}\left[ \frac{\textnormal{IF}_{F_1(k)}(y^t,a^t,w^t)}{F_1(k)\{1-F_1(k)\}}-\frac{\textnormal{IF}_{F_0(k)}(y^t,a^t,w^t)}{F_0(k)\{1-F_0(k)\}}\right],
\end{equation*}
which under the sharp null becomes
\begin{equation*}
    D_{m}(y^t,a^t,w^t) = \left(\frac{a^t}{\pi_1}-\frac{1-a^t}{\pi_0}\right)\frac{1}{K-1}\sum_{k=1}^{K-1}\frac{I\{y^t\leq k\}-\theta^*(k,w^t)}{F(k)\{1-F(k)\}},
\end{equation*}
The variance of $D_m(Y^t,A,W^t)$ under the sharp null is $\sigma_{m}^2/(\pi_1\pi_0)$.

Finally the efficient influence function can be obtained by projecting $D_u$ onto the tangent space.
\begin{multline*}
    D^*(y^t,a^t,w^t) = \frac{1}{K-1}\sum_{k=1}^{K-1}\Big(\frac{a^t[I\{y^t\leq k\}-\theta_1(k,w^t)]/\pi_1+\theta_1(k,w^t)-F_1(k)}{F_1(k)\{1-F_1(k)\}} \\
    -\frac{(1-a^t)[I\{y^t\leq k\}-\theta_0(k,w^t)]/\pi_0+\theta_0(k,w^t)-F_0(k)}{F_0(k)\{1-F_0(k)\}}\Big).
\end{multline*}
Under the sharp null, it simplifies to
\begin{equation*}
    D^*(y^t,a^t,w^t) = \frac{1}{K-1}\sum_{k=1}^{K-1}\left[\left(\frac{a^t}{\pi}-\frac{1-a^t}{\pi_0}\right)\frac{I\{y^t \leq k\}-\theta(k,w^t)}{F(k)\{1-F(k)\}}\right],
\end{equation*}
The variance of $D^*(Y^t,A,W^t)$ under the sharp null is $\sigma_a^2/(\pi_1\pi_0)$. 
\end{proof}

\begin{proof}[Proof of Theorem \ref{LOR}]
We first consider estimating $\sigma_u^2$. Define the following estimating equation
\begin{equation}
    U_{LOR,u}(\boldsymbol{\check F},\sigma^2)(y,w) = \left[\frac{1}{K-1}\sum_{k=1}^{K-1}\frac{I\{y\leq k\}-\check F(k)}{\check F(k)\{1-\check F(k)\}}\right]^2-\sigma^2.
\end{equation}
Then, $\sigma_u^2$ is the unique solution in $\sigma^2$ to the equation $PU_{LOR,u}(\boldsymbol{F},\sigma^2) = 0$; and $\hat\sigma_u^2$ solves $P_nU_{LOR,u}(\boldsymbol{\tilde F},\sigma^2)=0$, with $\boldsymbol{\tilde F}$ being the CDF of the empirical distribution of $Y$. We can apply similar estimating equation arguments as in the proof of Theorems~\ref{DIM} and \ref{MW} to show that $\hat\sigma_u^2 - \sigma_u^2 =(P_n-P)\textnormal{IF}_u+o_P(n^{-1/2})$, where
\begin{equation*}
    \textnormal{IF}_u(y,w) = U_{LOR,u}(\boldsymbol{F},\sigma_u^2)(y,w)+\sum_{k=1}^{K}\frac{\partial PU_{LOR,u}}{\partial \check F(k)}\Big\rvert_{\boldsymbol{F}}\left[I\{y\leq k\}-F(k)\right].
\end{equation*}
This implies that $\hat\sigma_u^2$ is asymptotically linear with influence function $\textnormal{IF}_u$. In particular,
\begin{equation*}
    \frac{\partial U_{LOR,u}}{\partial \check F(k)}\Big\rvert_{\boldsymbol{F}}(y,w) = -\frac{2}{(K-1)^2}\left[\sum_{l=1}^{K-1}\frac{I\{y\leq l\}-F(l)}{F(l)\{1-F(l)\}}\right]\left[\frac{I\{y\leq k\}-F(k)}{F(k)\{1-F(k)\}}\right]^2.
\end{equation*}

Next we consider the estimation of $\sigma_{m}^2$. Define the following estimating equation
\begin{equation*}
    U_{LOR}(\alpha,\beta,\boldsymbol{\check F},\sigma^2)(y,w) = \left[\frac{1}{K-1}\sum_{k=1}^{K-1}\frac{I\{y\leq k\}-\theta_{\alpha,\beta}(k,w)}{\check F(k)\{1-\check F(k)\}}\right]^2 - \sigma^2.
\end{equation*}
Then, $\sigma_{m}^2$ is the solution in $\sigma^2$ to $PU_{LOR}(\alpha^*,\beta^*,\boldsymbol{F},\sigma^2)=0$, while $\hat\sigma_{m}^2$ is the solution to $P_n U_{LOR}(\hat\alpha,\hat\beta,\boldsymbol{\tilde F},\sigma^2)=0$. We can again apply estimating equation arguments as in the proof of Theorems~\ref{DIM} and \ref{MW} to show that
\begin{equation*}
    \hat\sigma^2_{m} - \sigma^2_{m} = (P_n-P)\textnormal{IF}_m(y,w) + o_P(n^{-1/2}),
\end{equation*}
where
\begin{multline*}
    \textnormal{IF}_m(y,w) = U_{LOR}(\alpha^*,\beta^*,\boldsymbol{F},\sigma_{m}^2)(y,w) + \left\{\frac{\partial PU_{LOR}}{\partial (\alpha,\beta)}\Big\rvert_{(\alpha^*,\beta^*,\boldsymbol{F})}\right\}^\top  \textnormal{IF}_{ab}(y,w) \\
    +\sum_{k=1}^{K}\frac{\partial PU_{LOR}}{\partial \check F(k)}\Big\rvert_{(\alpha^*,\beta^*,\boldsymbol{F})}\left[I\{y\leq k\}- F(k)\right].
\end{multline*}
This implies that $\hat\sigma^2_m$ is asymptotically linear with influence function $\textnormal{IF}_m$. Here,
\begin{align*}
    \frac{\partial U_{LOR}}{\partial \alpha_k}\Big\rvert_{(\alpha^*,\beta^*,\boldsymbol{F})}(y,w) &= -\frac{2}{(K-1)^2}\left[\sum_{l=1}^{K-1}\frac{I\{y\leq l\}-\theta^*(l,w)}{ F(l)\{1- F(l)\}}\right]\frac{\theta^*(k,w)\{1-\theta^*(k,w)\}}{ F(k)\{1- F(k)\}}, \\
    \frac{\partial U_{LOR}}{\partial \beta}\Big\rvert_{(\alpha^*,\beta^*,\boldsymbol{F})}(y,w) &= -\frac{2}{(K-1)^2}\left[\sum_{l=1}^{K-1}\frac{I\{y\leq l\}-\theta^*(l,w)}{ F(l)\{1- F(l)\}}\right]\left[\sum_{l=1}^{K}\frac{w\theta^*(l,w)\{1-\theta^*(l,w)\}}{ F(l)\{1- F(l)\}}\right], \\
    \frac{\partial U_{LOR}}{\partial \check F(k)}\Big\rvert_{(\alpha^*,\beta^*,\boldsymbol{F})}(y,w) &= -\frac{2}{(K-1)^2}\left[\sum_{l=1}^{K-1}\frac{I\{y\leq l\}-\theta^*(l,w)}{ F(l)\{1- F(l)\}}\right]\frac{[I\{y\leq k\}- F(k)]\{1-2 F(k)\}}{ F^2(k)\{1- F(k)\}^2}.
\end{align*}

Finally we consider the estimation of $\sigma_a^2$. We define $\sigma_{kl} := E_P\left[\text{cov}(I\{Y\leq k\},I\{Y\leq l\}|W)\right]$. Then $\sigma^2_a$ can be equivalently written as
\begin{equation*}
    \sigma_a^2(P) = \frac{1}{(K-1)^2}\sum_{k=1}^{K-1}\sum_{l=1}^{K-1}\frac{\sigma_{kl}}{ F(k)F(l)\{1- F(k)\}\{1-F(l)\}}.
\end{equation*}
We estimate $ F(k)$ with $\tilde F(k)$, which is an asymptotically linear estimator. By Lemma \ref{gradientcondcov}, the EIF of $\sigma_{kl}$ is given by
\begin{equation*}
    D_{kl}(y,w) = \left[I\{y\leq k\}-\theta(k,w)\right]\left[I\{y\leq l\}-\theta(l,w)\right] - \sigma_{kl}.
\end{equation*}

We show that the estimator $\hat\sigma_{kl} = n^{-1}\sum_{i=1}^{n}[I\{Y_i\leq k\}-\hat \theta(k,W_i)][I\{Y_i\leq l\}-\hat \theta(l,W_i)]$ is asymptotically linear with influence function $D_{kl}$. For the ease of notation, we use $I_k$ as shorthand for the function $y \mapsto I\{y \leq k\}$, $\theta_k$ for the function $w \mapsto \theta(k,w)$ and $\hat\theta_k$ for the function $w \mapsto \hat\theta(k,w)$, for $k \in \{1,\ldots,K-1\}$. Then we have that $\hat\sigma_{kl} = P_n\{(I_k-\hat\theta_k)(I_l-\hat\theta_l)\}$ and $\sigma_{kl} = P\{(I_k-\theta_k)(I_l-\theta_l)\}$. Therefore,
\begin{align*}
    \hat\sigma_{kl}-\sigma_{kl} &= P_n\left\{(I_k-\hat\theta_k)(I_l-\hat\theta_l)\right\}-P\left\{(I_k-\theta_k)(I_l-\theta_l)\right\} \\
    &= (P_n-P)\left\{(I_k-\theta_k)(I_l-\theta_l)\right\} + P\left\{(I_k-\hat\theta_k)(I_l-\hat\theta_l)-(I_k-\theta_k)(I_l-\theta_l)\right\}\\
    &\ \ + (P_n-P)\left\{(I_k-\hat\theta_k)(I_l-\hat\theta_l)-(I_k-\theta_k)(I_l-\theta_l)\right\}.
\end{align*}
The first term is exactly $(P_n-P)D_{kl}$. For the second term, we note that
\begin{multline*}
    \left|P\left\{(I_l-\theta_l)(\theta_k-\hat\theta_k)+(I_k-\theta_k)(\theta_l-\hat\theta_l)+(\theta_k-\hat\theta_k)(\theta_l-\hat\theta_l)\right\}\right| \\
    =\left|P\left\{(\theta_k-\hat\theta_k)(\theta_l-\hat\theta_l)\right\}\right| \leq \|\hat\theta_k-\theta_k\|_{L^2(P_W)}\|\hat\theta_l-\theta_l\|_{L^2(P_W)} = o_P(n^{-1/2}).
\end{multline*}
Now we turn to the third term. By Theorem 2.10.6 in \citet{van1996weak}, the fact that $\hat\theta_k$ and $\hat\theta_l$ belong to fixed $P$-Donsker classes with probability tending to 1 implies that $(I_k-\hat\theta_k)(I_l-\hat\theta_l)-(I_k-\theta_k)(I_l-\theta_l)$ also belongs to a fixed P-Donsker with probability tending to 1, as the support of $W$ is bounded and $I_k, I_l$, $\theta_k$ and $\theta_l$ are all fixed and bounded functions. In addition, $\|(I_k-\hat\theta_k)(I_l-\hat\theta_l)-(I_k-\theta_k)(I_l-\theta_l)\|_{L^2(P)}=o_P(1)$. Thus, Lemma 19.24 implies that the third term is $o_P(n^{-1/2})$.
Hence $\hat\sigma_{kl}$ has influence function $D_{kl}$. 

Applying the delta method, we see that $\hat\sigma^2_a$ is asymptotically linear with influence function
\begin{multline*}
    \textnormal{IF}_a(y,w) = \frac{1}{(K-1)^2}\sum_{k=1}^{K-1}\sum_{l=1}^{K-1} \Bigg(\frac{D_{kl}(y,w)}{ F(k)F(l)\{1- F(k)\}\{1-F(l)\}} \\
    - \frac{\sigma_{kl}\{1-2 F(k)\}[I\{y\leq k\}- F(k)]}{ F^2(k)\{1- F(k)\}^2F(l)\{1-F(l)\}}- \frac{\sigma_{kl}\{1-2F(l)\}[I\{y\leq l\}-F(l)]}{F^2(l)\{1-F(l)\}^2 F(k)\{1- F(k)\}}\Bigg).
\end{multline*}

Theorem \ref{LOR} then follows by applying the delta method. Specifically, the influence function of $\hat\phi_a$ is given by $(\textnormal{IF}_a -\phi_a \textnormal{IF}_u)/\sigma^2_u$, and similarly the influence function of $\hat\phi_m$ is $(\textnormal{IF}_m -\phi_m \textnormal{IF}_u)/\sigma^2_u$. Regularity of $\hat\sigma_u^2$ and $\hat\sigma_m^2$ can be established using arguments similar to Lemma \ref{gradientEE}. Finally, as we are working within a locally nonparametric model, all of these estimators are efficient.
\end{proof}

\begin{proof}[Proof of Theorem \ref{bootstrap}]
First we introduce some notations. Let $\mathbb{P}_n^*$ denote the empirical distribution of a generic first-layer bootstrap resample from the external data $X$. Let $\tilde X$ denote a generic second-layer sample from $\mathbb{P}_n^*\Pi$, where $\Pi$ is the (known) distribution of the treatment. In what follows we consider a generic estimator such that, for any distribution $\mathbb{Q}$ of $(Y,W)$, the estimator is asymptotically linear with influence function $D_{\mathbb{Q}\Pi}$ in the trial with distribution $\mathbb{Q}\Pi$. We define $\sigma^2(\mathbb{Q}) := {\rm var}_{\mathbb{Q}\Pi}[D_{\mathbb{Q}\Pi}(Y,A,W)] $, and we recall that $\tilde \sigma^2(\mathbb{P}_n^*)$ denotes $N\textnormal{var}_{\mathbb{P}_n^*}[\hat\psi(\tilde X)]$. Because $\Pi$ is fixed, we omit the dependences of $\sigma^2$ and $\tilde \sigma^2$ on this quantity.

The proof below is a modification of the proof of Theorem 23.9 in \citet{van2000asymptotic}. Let $BL_1(\mathbb{R})$ be the set of all functions $h: \mathbb{R} \xrightarrow{} [-1,1]$ that are uniformly Lipschitz. We use subscript $M$ to denote taking expectation conditionally on the external data $X_1,X_2, \ldots$. Let $(\sigma^2)^{\prime}$ be the G{\^a}teaux derivative of the functional $\sigma^2 : \mathbb{Q} \mapsto \sigma^2(\mathbb{Q}) = {\rm var}_{\mathbb{Q}\Pi}[D_{\mathbb{Q}\Pi}(Y,A,W)] $. To start, we note that 

\begin{align*}
    &\sup_{h \in BL_1(\mathbb{R})}\left|E_{M}h\left(\sqrt{n}\left\{\tilde\sigma^2(\mathbb{P}_n^*)-\sigma^2(\mathbb{P}_n)\right\}\right)-Eh\left((\sigma^2)'(\mathbb{G})\right)\right| \nonumber \\
    &\leq \underbrace{\sup_{h \in BL_1(\mathbb{R})}\left|E_{M}h\left(\sqrt{n}\left\{\sigma^2(\mathbb{P}_n^*)-\sigma^2(\mathbb{P}_n)\right\}\right)-Eh\left((\sigma^2)'(\mathbb{G})\right)\right|}_{\text{term 1}} \nonumber \\
    &\quad+ \underbrace{\sup_{h \in BL_1(\mathbb{R})}\left|E_{M}h\left(\sqrt{n}\left\{\tilde\sigma^2(\mathbb{P}_n^*)-\sigma^2(\mathbb{P}_n)\right\}\right)-E_{M}h\left(\sqrt{n}\left\{\sigma^2(\mathbb{P}_n^*)-\sigma^2(\mathbb{P}_n)\right\}\right)\right|}_{\text{term 2}}.
\end{align*}

We study term 2 first. By the Lipschitz property of $h$, term 2 is bounded by
\begin{equation*}
    \sup_{h\in BL_1(\mathbb{R})}E_M\left|\sqrt{n}\left\{\tilde\sigma^2(\mathbb{P}_n^*)-\sigma^2(\mathbb{P}_n^*)\right\}\right| = E_M\left[\sqrt{n}\left|\tilde\sigma^2(\mathbb{P}_n^*)-\sigma^2(\mathbb{P}_n^*)\right|\right].
\end{equation*}
We now show that $E_M\left[\sqrt{n}\left|\tilde\sigma^2(\mathbb{P}_n^*)-\sigma^2(\mathbb{P}_n^*)\right|\right] \overset{P}{\rightarrow} 0$. Recall the asymptotic linear expansion \eqref{ALexpansion}, 
\begin{equation*}
    \sqrt{N}\left\{\hat\psi(\tilde X)-\psi(\mathbb{P}_{n}^*\Pi)\right\}=\frac{1}{\sqrt{N}}\sum_{l=1}^{N}D_{\mathbb{P}_{n}^*\Pi}(\tilde X_{l})+ \sqrt{N}\text{Rem}.
\end{equation*}
It then follows that
\begin{equation*}
    \tilde \sigma^2(\mathbb{P}_n^*) = \sigma^2(\mathbb{P}_n^*) 
    + 2\text{cov}_{\mathbb{P}_n^*}\left(\frac{1}{\sqrt{N}}\sum_{l=1}^{N}D_{\mathbb{P}_{n}^*\Pi}(\tilde X_{l}),\sqrt{N}\text{Rem}\right)+N\textnormal{var}_{\mathbb{P}_n^*}(\text{Rem}).
\end{equation*}
We note that
$N\textnormal{var}_{\mathbb{P}_n^*}(\text{Rem})=N^{1-2\gamma}\textnormal{var}_{\mathbb{P}_n^*}(N^\gamma \text{Rem})$. Therefore, Condition \ref{Eremaindervar} implies that $E\left[N\textnormal{var}_{\mathbb{P}_n^*}(\text{Rem})\right] \leq LN^{1-2\gamma}$ for some constant $L$. In addition, the boundedness of the support of $X$ implies that $\textnormal{var}_{\mathbb{P}_n^*}[D_{\mathbb{P}_{n}^*\Pi}(\tilde X)]$ is also bounded. Thus, by the Cauchy Schwartz inequality and Jensen's inequality,
\begin{align*}
     E&\left[\left|\text{cov}_{\mathbb{P}_n^*}\left(\frac{1}{\sqrt{N}}\sum_{l=1}^{N}D_{\mathbb{P}_{n}^*\Pi}(\tilde X_{l}),\sqrt{N}\text{Rem}\right)\right|\right] \\
    &\leq E\left[\sqrt{\textnormal{var}_{\mathbb{P}_n^*}\left[D_{\mathbb{P}_{n}^*\Pi}(\tilde X)\right]\left\{N^{1-2\gamma}\textnormal{var}_{\mathbb{P}_n^*}(N^\gamma \text{Rem})\right\}}\right] \leq \sqrt{L_1N^{1-2\gamma}},
\end{align*}
for some constant $L_1$. As a result, we have that $E\left[\sqrt{n}\left|\tilde\sigma^2(\mathbb{P}_n^*)-\sigma^2(\mathbb{P}_n^*)\right|\right] \leq \sqrt{L_2nN^{1-2\gamma}}$ for some constant $L_2$. Or equivalently, $E\left[E_M\left[\sqrt{n}\left|\tilde\sigma^2(\mathbb{P}_n^*)-\sigma^2(\mathbb{P}_n^*)\right|\right]\right] \leq \sqrt{L_2nN^{1-2\gamma}}$, where the outer expectation is over $X_1, X_2, \ldots$ Markov's inequality and Condition \ref{Nbiggern} then imply that, for all $\epsilon>0$,
\begin{equation*}
    P(E_M\left[\sqrt{n}\left|\tilde\sigma^2(\mathbb{P}_n^*)-\sigma^2(\mathbb{P}_n^*)\right|\right] > \epsilon) \leq \frac{\sqrt{L_2nN^{1-2\gamma}}}{\epsilon} \rightarrow 0, \quad \text{as} \ n \rightarrow \infty.
\end{equation*}
This shows that $E_M\left[\sqrt{n}\left|\tilde\sigma^2(\mathbb{P}_n^*)-\sigma^2(\mathbb{P}_n^*)\right|\right] \overset{P}{\rightarrow} 0$.

We now turn to term 1, which further decomposes into
\begin{align*}
    & \ \ \underbrace{\sup_{h \in BL_1(\mathbb{R})}\left|E_{M}h\left(\sqrt{n}\left\{\sigma^2(\mathbb{P}_n^*)-\sigma^2(\mathbb{P}_n)\right\}\right)-E_Mh\left((\sigma^2)'(\sqrt{n}(\mathbb{P}_n^*-\mathbb{P}_n))\right)\right|}_{\text{term 1.1}} \\
    &+ \underbrace{\sup_{h \in BL_1(\mathbb{R})}\left|E_Mh\left((\sigma^2)'(\sqrt{n}(\mathbb{P}_n^*-\mathbb{P}_n))\right)-Eh\left((\sigma^2)'(\mathbb{G})\right)\right|}_{\text{term 1.2}}.
\end{align*}
 
Theorem 23.7 and Equation 23.8 in \citet{van2000asymptotic} imply that term 1.2 converges to 0 in probability. For term 1.1, the same argument as was used in the proof of Theorem 23.9 in \citet{van2000asymptotic} shows that this term converges to 0 in probability.

Combining these steps as in the proof of Theorem 23.9 in \citet{van2000asymptotic}, we see that $\sqrt{n}\{\tilde\sigma^2(\mathbb{P}_n^*)-\sigma^2(\mathbb{P}_n)\}$ converges conditionally in distribution to $(\sigma^2)^\prime(\mathbb{G})$, given $X_1,X_2,\ldots$, in probability. 

In particular, we can apply the above argument to the variance of the unadjusted estimator $\sigma_u^2$ and the variance of the working-model-based estimator $\sigma_m^2$ to show that (i) $\sqrt{n}\{\tilde\sigma_u^2(\mathbb{P}_n^*)-\sigma_u^2(\mathbb{P}_n)\}$ converges conditionally in distribution to $(\sigma_u^2)^\prime(\mathbb{G})$, and (ii) $\sqrt{n}\{\tilde\sigma_m^2(\mathbb{P}_n^*)-\sigma_m^2(\mathbb{P}_n)\}$ converges conditionally in distribution to $(\sigma_m^2)^\prime(\mathbb{G})$, both given $X_1,X_2,\ldots$, in probability. The delta method implies that $\sqrt{n}\{\tilde\sigma_m^2(\mathbb{P}_n^*)/\tilde\sigma_u^2(\mathbb{P}_n^*)-\phi_m(\mathbb{P}_n)\}$ converges conditionally in distribution to $(\phi_m)^\prime(\mathbb{G})$, given $X_1,X_2,\ldots$, in probability. Finally, Theorem 3.10 follows by Condition \ref{B2} and Slutsky's theorem.
\end{proof}

\section{Proofs of results for time-to-event outcomes with right censoring}\label{survivalproof}

All three estimands of the treatment effect we consider are transformations of the arm-specific survival functions $S_1(t)$ and $S_0(t)$. Recall that for the unadjusted analysis we plug in $\tilde S_1(t)$ and $\tilde S_0(t)$, the arm-specific KM estimators, and, for the adjusted analysis, we plug in $\hat S_1(t)$ and $\hat S_0(t)$, the efficient adjusted estimators for the arm-specific survival function. 

The influence function of the KM estimator was derived, for example, in \citet{reid1981influence}. In particular, if $T^t\perp C^t|A$, then $\tilde S_a(t_k)$ is an asymptotically linear estimator of $S_a(t_k)$ with influence function
\begin{align*}
    \eta_{a,k}(y^t,\delta^t,a^t,w^t) &=  \sum_{j=1}^{k}-\frac{I\{a^t=a\}S_a(t_k)}{S_a(t_j)G_a(t_j)\pi_a}\left[I\{y^t = t_j,\delta^t=1\}-h_a(t_j)I\{y^t \geq t_j\}\right],
\end{align*}
for $a= 0$ or 1, and $k \in \{1,\ldots,K\}$. Here $h_a(t)$ is the hazard corresponding to $S_a$ at time $t$, $G_a(t,w) := P(C^t\geq t|A=a,W^t=w)$ and $G_a(t) := P(C^t\geq t|A=a)$.

Under the assumption that $T^t \perp C^t |(A,W^t)$ and other regularity conditions, $\hat S_a(t_k)$ is asymptotically linear with influence function given in \citet{moore2009application}
\begin{align*}
    \lambda_{a,k}(y^t,\delta^t,a^t,w^t) &= \sum_{j=1}^{k} -\frac{I\{a^t=a\}S_a(t_k,w^t)}{\pi_aS_a(t_j,w^t)G_a(t_j,w^t)}\left[I\{y^t =t_j,\delta^t=1\}-I\{y^t \geq t_j\}h_a(t_j,w^t)\right] \nonumber \\
    & \ \ \ +S_a(t_k,w^t)-S_a(t_k), 
\end{align*}
for $a= 0$ or 1, and $k \in \{1,\ldots,K\}$.

\subsection{Supporting lemmas for proofs in Section~\ref{proofC2}}
\begin{lemma}\label{eifsurvvar}
(EIF of the variance of the fully-adjusted estimators) For $(k,l)\in\{1,\ldots,K\}^2$ and $j\in \{1,\ldots,\min(k,l)\}$, let $D_{j}^{kl}$ be defined as in \eqref{Djkl}. When we measure the treatment effect with the risk difference $S_0(t_k) - S_1(t_k)$ or the relative risk $\{1-S_1(t_k)\}/\{1-S_0(t_k)\}$, the EIF of $\sigma_a^2$ is $\sum_{j=1}^{k}D_{j}^{kk}$. When we use restricted mean survival time $\sum_{j=1}^{k}\{S_1(t_j)-S_0(t_j)\}$ as the treatment effect estimand, the EIF of $\sigma_a^2$ is $\sum_{j=1}^{k}\sum_{l=1}^{k}\sum_{u=1}^{\min(j,l)}D_{u}^{jl}$. 
\end{lemma}

\begin{proof}
Recall that we define $f_{j}^{kl}(\cdot)$ as
\begin{equation*}
    f_j^{kl}(w) = S(t_k,w)S(t_l,w)\left\{\frac{1}{S(t_j,w)}-\frac{1}{S(t_{j-1},w)}\right\}.
\end{equation*}
When we wish to emphasize the dependence of $f_j^{kl}$ on $P$ through its survival function, we instead denote this function by $f_{j,P}^{kl}$.
First we define a parameter $\sigma_{j}^{kl}: \mathcal{M} \rightarrow \mathbb{R}$ such that $\sigma_{j}^{kl}(P) := E_P[f_{j,P}^{kl}(W)/G(t_j,W)]$. We consider the efficient influence function of $\sigma_{j}^{kl}$ in the full data model, that is, the model where we observe $(T,W)$ and there is no censoring. Let $p(t,w) = p(t|w)p(w)$ be the density of the joint distribution.

We consider the one-dimensional submodel $\{P_\epsilon : |\epsilon|\le 1\}$ with density $p(t|w)\{1+\epsilon s_1(t|w)\}p(w)\{1+\epsilon s_2(w)\}$, where the range of $s_1$ and $s_2$ falls in $[-1,1]$ and these functions satisfy $E_P[s_1(Y|W)|W]=0$ $P$-almost surely and $E_P[s_2(W)]=0$. Let $S_\epsilon$ be the survival function corresponding to $P_\epsilon$ and define $f_{j,\epsilon}^{kl}$ similarly to $f_{j}^{kl}$ but with $S$ replaced by $S_\epsilon$. We omit the subscript ``$j$" and superscript ``$kl$" below when it is clear from the context that we are focusing on $f_j^{kl}$. 

Note that
\begin{align*}
    \frac{d}{d\epsilon} S_\epsilon(t_k,w) &= \frac{d}{d\epsilon} \int I\{t>t_k\}\{1+\epsilon s_1(t|w)\}dP(t|w) \\
    &= \int [I\{t>t_k\}-S(t,w)] s_1(t|w)dP(t|w) \\
    &= \int [I\{t>t_k\}-S(t,w)] \{s_1(t|w)+s_2(w)\}dP(t|w).
\end{align*}
Because
\begin{equation*}
    \sigma_j^{kl}(P_\epsilon)
    = \int \frac{S_\epsilon(t_k,w)S_\epsilon(t_l,w)}{G(t_j,w)}\left\{\frac{1}{S_\epsilon(t_j,w)}-\frac{1}{S_\epsilon(t_{j-1},w)}\right\}p(w)\{1+\epsilon s_2(w)\}dw,
\end{equation*}
we then see that
\begin{equation*}
    \frac{d}{d\epsilon}\sigma_j^{kl}(P_\epsilon)\big\rvert_{\epsilon = 0} = \int \frac{f(w)}{G(t_j,w)}s_2(w)dP(w) + \int \frac{1}{G(t_j,w)}\frac{d}{d\epsilon}f_\epsilon(w)\big\rvert_{\epsilon=0}dP(w).
\end{equation*}

Let $\tau_{full,k}(t,w) = I\{t>t_k\}-S(t_k,w)$. By definition, the gradient is given by
\begin{multline*}
    D_{full}(t,w) = \Big\{f(w) + d_l(w)\tau_{full,l}(t,w)+d_k(w)\tau_{full,k}(t,w)+d_j(w) \tau_{full,j}(t,w)\\
    +d_{j-1}(w)\tau_{full,(j-1)}(t,w)\Big\}/G(t_j,w) - E[f(w)/G(t_j,w)],
\end{multline*}
where the partial derivatives $d_l$, $d_k$, $d_j$, and $d_{j-1}$ are given by
\begin{align*}
    d_l(w) &= \frac{S(t_k,w)}{S(t_j,w)}-\frac{S(t_k,w)}{S(t_{j-1},w)}, \ \ d_k(w) = \frac{S(t_l,w)}{S(t_j,w)}-\frac{S(t_l,w)}{S(t_{j-1},w)}, \\
    d_j(w) &= -\frac{S(t_k,w)S(t_l,w)}{S^2(t_j,w)}, \ \ d_{j-1}(w) = \frac{S(t_k,w)S(t_l,w)}{S^2(t_{j-1},w)}.
\end{align*}
This is the EIF in the full data model, as we work with a locally nonparametric model.

The observed data unit (with censoring) is $(Y,\Delta,W)$. To find the EIF in the observed data model, we apply Theorem 10.1 of \citet{tsiatis2007semiparametric} to show that the following is an observed data influence function of $\sigma_j^{kl}$: 
\begin{align}\label{Djkl}
    D_{j}^{kl}(y,\delta,w) &= \Big\{f_{j}^{kl}(w) + d_l(w)\tau_{l}(y,\delta,w)+d_k(w)\tau_{k}(y,\delta,w)+d_j(w) \tau_{j}(y,\delta,w) \nonumber \\
    &\ \ \ \ \ +d_{j-1}(w)\tau_{j-1}(y,\delta,w)\Big\}/G(t_j,w) - E[f_{j}^{kl}(w)/G(t_j,w)] \nonumber \\
    &= \left\{g_j^{kl}(y,\delta,w)
    +f_j^{kl}(w)\right\}/G(t_j,w)- E[f_{j}^{kl}(w)/G(t_j,w)].
\end{align}
Moreover, as our observed data model is locally nonparametric, $D_{j}^{kl}$ is the efficient observed data influence function of $\sigma_j^{kl}$. 
Lemma \ref{eifsurvvar} then follows since the variances are linear combinations of $E[f_j^{kl}(W)/G(t_j,W)]$.
\end{proof}

\begin{lemma}[Computation time of $\hat\sigma^2_a$ with RMST]
For given $\hat\tau$ and $\hat S$, the function
\begin{equation*}
    (y,\delta,w)\mapsto \sum_{j=1}^{k}\sum_{l=1}^{k}\sum_{u=1}^{\min(j,l)}\left\{\hat g_u^{jl}(y,\delta,w)+\hat f_u^{jl}(w)\right\}
\end{equation*}
can be computed in $O(k)$ time. 
\end{lemma}

\begin{proof}
There are 5 terms inside the sums, and we show that the sum of each term can be computed in $O(k)$ time.

To start, we note that
\begin{align*}
    \sum_{j=1}^{k}\sum_{l=1}^{k}\sum_{u=1}^{\min(j,l)}\hat f_u^{jl}(w) &= \sum_{j=1}^{k}\sum_{l=1}^{k}\sum_{u=1}^{\min(j,l)}\hat S(t_j,w)\hat S(t_l,w)\left\{\frac{1}{\hat S(t_u,w)}-\frac{1}{\hat S(t_{u-1},w)}\right\} \\
    &= \sum_{u=1}^{k}\left[\left\{\frac{1}{\hat S(t_u,w)}-\frac{1}{\hat S(t_{u-1},w)}\right\} \left\{\sum_{j\geq u}\hat S(t_j,w)\right\} \left\{\sum_{l \geq u}\hat S(t_l,w)\right\}\right].
\end{align*}
By taking a cumulative sum, $\{\sum_{l \geq u}\hat S(t_l,w)\}_{u=1}^k$ can be pre-computed in $O(k)$ time. Thus, the above display can be computed in $O(k)$ time as we sum over $u$.

The terms in $g$ are of two types. The first of these is
\begin{align*}
    &\ \ \ \sum_{j=1}^{k}\sum_{l=1}^{k}\sum_{u=1}^{\min(j,l)}\hat S(t_j,w)\left\{\frac{1}{\hat S(t_u,w)}-\frac{1}{\hat S(t_{u-1},w)}\right\}\hat \tau_l(y,\delta,w) \\
    &=  \sum_{u=1}^{k}\left[\left\{\frac{1}{\hat S(t_u,w)}-\frac{1}{\hat S(t_{u-1},w)}\right\}\left\{\sum_{l\geq u}\hat \tau_l(y,\delta,w)\right\}\left\{\sum_{j \geq u}\hat S(t_j,w)\right\}\right].
\end{align*}
By taking a cumulative sum, $(\sum_{l \geq u}\hat\tau_l)_{u=1}^k$ can be pre-computed in $O(k)$ time and summing over $u$ takes another $O(k)$ steps. The second type of term is
\begin{multline*}
    \sum_{j=1}^{k}\sum_{l=1}^{k}\sum_{u=1}^{\min(j,l)}\frac{\hat S(t_j,w) \hat S(t_l,w)}{\hat S^2(t_u,w)}\hat\tau_u(y,\delta,w) = \\
    \sum_{u=1}^{k}\left[\left\{\frac{\hat\tau_u(y,\delta,w)}{\hat S^2(t_u,w)}\right\}\left\{\sum_{l\geq u}\hat S(t_l,w)\right\}\left\{\sum_{j \geq u}\hat S(t_j,w)\right\}\right].
\end{multline*}
Again with the sum over $j$ and $l$ pre-computed for all $u$, the summation over $u$ takes $O(k)$ time.
\end{proof}

\subsection{Results in Appendix~\ref{sec:survival}}\label{proofC2}

\begin{proof}[Proof of Lemma~\ref{reRD}]
First we note that the (conditional) independencies $A \perp W^t$, $T^t\perp A|W^t$, $C^t \perp T^t|A$ and $C^t \perp T^t|(A,W^t)$ together imply that $T^t \perp C^t$ and $T^t\perp C^t|W^t$. This result will be useful when we compute the variances of the adjusted and unadjusted estimators.

We consider the risk difference first. The unadjusted estimator, namely $\hat\psi_u = \tilde S_0(t_k) - \tilde S_1(t_k)$, has influence function $D_{u,k} =\eta_{0,k} - \eta_{1,k}$. Under the sharp null where the treatment has no effect and the assumption that $C^t\perp A|W^t$, the influence function simplifies to 
\begin{align}\label{IFnull}
    D_{u,k}(y^t,\delta^t,a^t,w^t) &= \left(\frac{a^t}{\pi_1} - \frac{1-a^t}{\pi_0}\right)\sum_{j=1}^{k} \frac{S(t_k)}{S(t_j)G(t_j)}I_j(y^t,\delta^t), \ \ \textnormal{where} \\
    I_j(y^t,\delta^t) &= I\{y^t = t_j,\delta^t=1\}-h(t_j)I\{y^t \geq t_j\}. \nonumber
\end{align}
Noting that
\begin{equation*}
    E\left[I_j(Y^t,\Delta^t)I_l(Y^t,\Delta^t)\right]
    = \delta_{jl}P(Y=t_j)G(t_j)\{1-h(t_j)\},  \ \ \textnormal{where} \ \ \delta_{jl} = I\{j = l\},
\end{equation*}
we see that
\begin{align*}
    \text{var}(D_{u,k}) &= \frac{1}{\pi_1\pi_0}\sum_{j=1}^{k}\frac{S^2(t_k)S(t_{j-1})h(t_j)\{1-h(t_j)\}}{S^2(t_j)G(t_j)} \nonumber \\
    &= \frac{1}{\pi_1\pi_0}\sum_{j=1}^{k}\frac{S^2(t_k)\{S(t_{j-1})-S(t_j)\}}{S(t_j)G(t_j)S(t_{j-1})} = \frac{\sigma_u^2}{\pi_1\pi_0}.
\end{align*}
The adjusted estimator, namely $\hat\psi_a = \hat S_0(t_k)-\hat S_1(t_k)$, has influence function $D_{a,k} = \lambda_{0,k}-\lambda_{1,k}$. Under the sharp null and the assumption that $C^t\perp A|W^t$, the influence function becomes 
\begin{multline}\label{Dak}
    D_{a,k}(y^t,\delta^t,a^t,w^t) = \left(\frac{a^t}{\pi_1}-\frac{1-a^t}{\pi_0}\right)\times \\
    \sum_{j=1}^{k} \frac{S(t_k,w^t)}{S(t_j,w^t)G(t_j,w^t)}\left[I\{y^t =t_j,\delta^t =1\}-I\{y^t \geq t_j\}h(t_j,w^t)\right].
\end{multline}
The variance of $D_{a,k}$ can be calculated in a similar way as was done for the unadjusted estimator, except that in taking expectation of the indicators, we condition on $W^t$ first. 

\begin{equation*}
    \text{var}(D_{a,k}) = \frac{1}{\pi_1\pi_0}E_P\left[\sum_{j=1}^{k}\frac{S^2(t_k,W)\{S(t_{j-1},W)-S(t_j,W)\}}{S(t_j,W)S(t_{j-1},W)G(t_j,W)}\right] = \frac{\sigma_a^2}{\pi_1\pi_0}.
\end{equation*}
Hence the relative efficiency is given by $\sigma_a^2/\sigma_u^2$, which depends only on the distribution of survival time and the covariate, for a user-specified mapping $G$.

Next we consider the relative risk. The unadjusted estimator, namely $\hat\psi_u = \frac{1-\tilde S_1(t_k)}{1-\tilde S_0(t_k)}$, has influence function $\frac{1}{1-S_0(t_k)}(-\eta_{1,k}+\psi \eta_{0,k})$, which becomes $D_{u,k}/\{1-S(t_k)\}$ under the sharp null. Similarly, the influence function of the adjusted estimator $\hat\psi_a = \frac{1-\hat S_1(t_k)}{1-\hat S_0(t_k)}$ simplifies to $D_{a,k}/\{1-S(t_k)\}$ under the sharp null and the assumption that $C^t\perp A|W^t$. Hence the relative efficiency is again $\sigma_a^2/\sigma_u^2$.
\end{proof}

\begin{proof}[Proof of Theorem \ref{RD}]
Recall that $\hat S(t_k)$ is the efficient adjusted estimator of the survival probability at time $t_k$ using the external data. Hence, this estimator is asymptotically linear with influence function 
\begin{equation*}
    \textnormal{IF}_k(y,\delta,w) = \tau_k(y,\delta,w) +S(t_k,w)-S(t_k),
\end{equation*}
where $\tau$ is defined in \eqref{eq:tau}. Furthermore, define $\textnormal{IF}_0(y,\delta,w) := 0$. Recall that $$\hat s_{u}^{jl} = \frac{\hat S(t_j)\hat S(t_l)\{\hat S(t_{u-1})-\hat S(t_{u})\}}{\hat S(t_u)\hat S(t_{u-1})}.$$
Applying the delta method, we have that $\hat s_{u}^{jl}$ is an asymptotic linear estimator of $s_{u}^{jl}$ with influence function
\begin{multline*}
    \xi_u^{jl}(y,\delta,w) = \left\{\frac{S(t_j)}{S(t_u)}-\frac{S(t_j)}{S(t_{u-1})}\right\}\textnormal{IF}_l(y,\delta,w) +\left\{\frac{S(t_l)}{S(t_u)}
    -\frac{S(t_l)}{S(t_{u-1})}\right\}\textnormal{IF}_j(y,\delta,w) \\-\frac{S(t_j)S(t_l)}{S^2(t_u)}\textnormal{IF}_u(y,\delta,w) + \frac{S(t_j)S(t_l)}{S^2(t_{u-1})}\textnormal{IF}_{u-1}(y,\delta,w).
\end{multline*}
Also, $\hat\sigma_u^2$ is a linear combination of $\hat s_{j}^{kk}$, and its influence function is given by 
\begin{equation*}
    \textnormal{IF}_u(y,\delta,w) = \sum_{j=1}^{k}\xi_{j}^{kk}(y,\delta,w)/G(t_j).
\end{equation*}

We now consider estimating $\sigma_a^2$. Its efficient influence function is derived in Lemma \ref{eifsurvvar}, and also given in \eqref{EIFcondvarsurv}, which is of the form $\textnormal{IF}_a = \sum_{j=1}^{k}D_j^{kk}$. The proposed $\hat\sigma_a^2$ is a one-step estimator based on the EIF. 

Let $Q$ be the distribution of the observed data unit $(Y,\Delta,W)$ in the external dataset, induced by the joint distribution $P$ of $(T,W)$, the conditional distribution of the censoring time $H$ and the function $\Gamma$. Let $\hat P$ be a joint distribution of $(T,W)$ such that the condtional hazard function is given by $\hat h(t,w)$, the conditional survival function is given by $\hat S(t,w)$ and the distribution of $W$ is given by its empirical distribution. Let $\hat Q$ be the observed data distribution induced by $\hat P$, $\hat H$ and $\Gamma$. Then, we have that $\hat\sigma_a^2 = Q_n\textnormal{IF}_a(\hat Q) + \sigma^2_a(\hat P)$, where $Q_n$ is the empirical distribution of $(Y,\Delta,W)$ in the external dataset. Hence, 
\begin{multline*}
    \hat\sigma_a^2 - \sigma^2_a = (Q_n-Q)\textnormal{IF}_a(Q) + (Q_n-Q)\{\textnormal{IF}_a(\hat Q)-\textnormal{IF}_a(Q)\} + R(\hat Q, Q), \\
    \text{where } R(\hat Q,Q) = \sigma_a^2(\hat P)-\sigma_a^2(P) + Q\{\textnormal{IF}_a(\hat Q)\}.
\end{multline*}

Our first step is to show that the remainder term $R(\hat Q,Q)$ is $o_P(n^{-1/2})$. As the influence function $\textnormal{IF}_a$ and the variance $\sigma_a^2$ itself can both be written as a sum of $k$ terms, it is easy to see that we can write $R(\hat Q,Q)$ as $\sum_{j=1}^{k}R_j(\hat Q,Q)$, where
\begin{align*}
    R_j(\hat Q, Q) &= Q\left[\frac{1}{G(t_j,w)} \left\{\hat g_j(y,\delta,w) +\hat f_j(w) - f_j(w)\right\}\right].
\end{align*}
Here we omit the superscript $``kk"$ in $f$ and $g$. We examine the terms coming from $f$ and $g$ separately. First, we study
\begin{multline*}
    Q\left\{\frac{\hat f_j(w)-f_j(w)}{G(t_j,w)}\right\} 
    \\= Q\left(\frac{1}{G(t_j,w)}\left[\left\{\frac{\hat S^2(t_k,w)}{\hat S(t_j,w)}-\frac{\hat S^2(t_k,w)}{\hat S(t_{j-1},w)}\right\}- \left\{\frac{S^2(t_k,w)}{S(t_j,w)}-\frac{S^2(t_k,w)}{ S(t_{j-1},w)}\right\}\right]\right). 
\end{multline*}
Define
\begin{align*}
    R_{f1} &= Q\left[\frac{1}{G(t_j,w)}\left\{\frac{2S(t_k,w)}{S(t_j,w)}-\frac{2S(t_k,w)}{S(t_{j-1},w)}\right\}\{\hat S(t_k,w)-S(t_k,w)\}\right],\\
    R_{f2} &= Q\left[-\frac{1}{G(t_j,w)} \frac{S^2(t_k,w)}{S^2(t_j,w)}\{\hat S(t_j,w)-S(t_j,w)\}\right], \\
    R_{f3} &= Q\left[\frac{1}{G(t_j,w)}\frac{S^2(t_k,w)}{S^2(t_{j-1},w)}\{\hat S(t_{j-1},w)-S(t_{j-1},w)\}\right].
\end{align*}
Then, we apply a second-order Taylor expansion to the function $(x,y,z) \mapsto x^2(1/y-1/z)$. The relevant second-order derivatives are bounded by some constant $M$ when $\hat S(t,w)$, $S(t,w)$ and $G(t,w)$ are all uniformly bounded away from 0. Then,
\begin{align*}
   Q\left\{\frac{\hat f_j(w)-f_j(w)}{G(t_j,w)}\right\}  &\leq R_{f1}+R_{f2}+R_{f3} +M\Big\{\|\hat S(t_k,\cdot) - S(t_k,\cdot)\|_{L^2(P_W)}^2 \\
    &\ \ \ +\|\hat S(t_j,\cdot) - S(t_j,\cdot)\|_{L^2(P_W)}^2+\|\hat S(t_{j-1},\cdot) - S(t_{j-1},\cdot)\|_{L^2(P_W)}^2 \Big\}\\
    &= R_{f1}+R_{f2}+R_{f3}+o_P(n^{-1/2}),
\end{align*}
since $\|\hat S(t,\cdot) - S(t,\cdot)\|_{L^2(P_W)} = o_P(n^{-1/4})$ for all $t$. 

Next, we study the terms in the remainder resulting from $g$. To do this, we define $\delta_h^l(w) := S(t_{l-1},w)\{h(t_l,w)-\hat h(t_l,w)\}/\hat S(t_l,w)$. Then,
\begin{equation*}
    Q \left\{\hat g_j(y,\delta,w)/G(t_j,w)\right\} = R_{g1} + R_{g2} + R_{g3},
\end{equation*}
where
\begin{align*}
    R_{g1} &= Q\left[\frac{1}{G(t_j,w)}\left\{\frac{2\hat S(t_k,w)}{\hat S(t_j,w)}-\frac{2\hat S(t_k,w)}{\hat S(t_{j-1},w)}\right\}\left\{\sum_{l \leq k}-\frac{\hat S(t_k,w) H(t_l,w)}{\hat  H(t_l,w)}\delta_h^l(w)\right\}\right], \\
    R_{g2} &= Q\left[-\frac{1}{G(t_j,w)}\frac{\hat S^2(t_k,w)}{\hat S^2(t_j,w)}\left\{\sum_{l \leq j}-\frac{\hat S(t_j,w) H(t_l,w)}{\hat  H(t_l,w)}\delta_h^l(w)\right\}\right], \\
    R_{g3} &= Q\left[\frac{1}{G(t_j,w)}\frac{\hat S^2(t_k,w)}{\hat S^2(t_{j-1},w)}\left\{\sum_{l \leq j-1}-\frac{\hat S(t_{j-1},w) H(t_l,w)}{\hat  H(t_l,w)}\delta_h^l(w)\right\}\right].
\end{align*}
Noting that $S(t_k,w) = \prod_{l \leq k}\{1-h(t_l,w)\}$ and $\hat S(t_k,w) = \prod_{l \leq k}\{1-\hat h(t_l,w)\}$, we can write the difference between $\hat S(t_k,w)$ and $S(t_k,w)$ as
$$\hat S(t_k,w)-S(t_k,w) = \sum_{l \leq k}\hat S(t_k,w)\delta_h^l(w).$$
Hence,
\begin{multline*}
    R_{f1}+R_{g1} = Q\Bigg(\frac{1}{G(t_j,w)}\Bigg[\left\{\frac{2S(t_k,w)}{S(t_j,w)}-\frac{2S(t_k,w)}{S(t_{j-1},w)}\right\}\sum_{l \leq k}\frac{\hat  H(t_l,w)- H(t_l,w)}{\hat  H(t_l,w)}\hat S(t_k,w)\delta_h^l(w) \\
    + \left\{\frac{2\hat S(t_k,w)}{\hat S(t_j,w)}-\frac{2\hat S(t_k,w)}{\hat S(t_{j-1},w)}-\frac{2S(t_k,w)}{S(t_j,w)}+\frac{2S(t_k,w)}{S(t_{j-1},w)}\right\}\sum_{l \leq k}-\frac{\hat S(t_k,w) H(t_l,w)}{\hat  H(t_l,w)}\delta_h^l(w)\Bigg]\Bigg)
\end{multline*}
The term in the first line on the right-hand side is $o_P(n^{-1/2})$ because $G$, $S$, $\hat S$ and $\hat H$ are uniformly bounded away from 0; $\hat S$ is uniformly bounded above; and $\|\{\hat h(t,\cdot)-h(t,\cdot)\}\{\hat H(t,\cdot)- H(t,\cdot)\}\|_{L^2(P_W)} = o_P(n^{-1/2})$ for all $t$. 
The term in the second line is also $o_P(n^{-1/2})$. To see this, we apply a first-order Taylor expansion to the first factor, which is very similar to the second-order Taylor expansion we studied earlier and the derivative is again bounded by some constant $M$. In addition, we have that $\|\{\hat S(t,\cdot)-S(t,\cdot)\}\{\hat h(\tilde t,\cdot)-h(\tilde t,\cdot)\}\|_{L^2(P_W)} = o_P(n^{-1/2})$ for all $(t,\tilde t)$. 

Similarly, we can show that $R_{f2}+R_{g2}$ and $R_{f3}+R_{g3}$ are both $o_P(n^{-1/2})$, and so is $R_j(\hat Q,Q)$ and consequently so is $R(\hat Q,Q)$. 

Our second step is to show that $(Q_n-Q)\{\textnormal{IF}_a(\hat Q)-\textnormal{IF}_a(Q)\}$ is $o_P(n^{-1/2})$. To do this, we again use Lemma 19.24 in \citet{van2000asymptotic}. We need to verify the following two conditions: (1) $\|\textnormal{IF}_a(\hat Q)-\textnormal{IF}_a(Q)\|_{L^2(Q)} = o_P(1)$, and (2) $\textnormal{IF}_a(\hat Q)$ lies in a fixed $Q$-Donsker class with probability tending to 1.

We first establish condition (1).
\begin{multline*}
    \textnormal{IF}_a(\hat Q)(y,\delta,w)-\textnormal{IF}_a(Q)(y,\delta,w) = \sigma_a^2(\hat Q)-\sigma_a^2(Q) \\
    + \sum_{j=1}^{k}\frac{1}{G(t_j,w)}\left\{\hat g_j(y,\delta,w)-g_j(y,\delta,w)+\hat f_j(w)-f_j(w)\right\}.
\end{multline*}
Using the triangle inequality, it suffices to show that for each $j$, $\|\{\hat g_j-g_j\}/G(t_j,\cdot)\|_{L^2(Q)}=o_P(1)$ and $\|\{\hat f_j-f_j\}/G(t_j,\cdot)\|_{L^2(Q)}=o_P(1)$, and that $\sigma_a^2(\hat Q)-\sigma_a^2(Q) = o_P(1)$. As $G$ is uniformly bounded away from 0, it suffices to show that $\|\hat g_j-g_j\|_{L^2(Q)}=o_P(1)$ and $\|\hat f_j-f_j\|_{L^2(Q)}=o_P(1)$. 
\begin{equation*}
    \hat f_j(w)-f_j(w)
    = \left\{\frac{\hat S^2(t_k,w)}{\hat S(t_j,w)}-\frac{\hat S^2(t_k,w)}{\hat S(t_{j-1},w)}\right\}- \left\{\frac{S^2(t_k,w)}{S(t_j,w)}-\frac{S^2(t_k,w)}{ S(t_{j-1},w)}\right\}. 
\end{equation*}
Recall that when studying the remainder term, we applied a second-order Taylor expansion to the function $(x,y,z) \mapsto x^2(1/y-1/z)$. Here a first-order Taylor expansion suffices. As $S$ is bounded away from 0, there exists some constant $M$ such that the first derivatives are bounded by $M$. Then,
\begin{multline*}
    \|\hat f_j-f_j\|_{L^2(Q)} \leq M\big\{\|\hat S(t_k,\cdot)-S(t_k,\cdot)\|_{L^2(P_W)} \\
    +\|\hat S(t_j,\cdot)-S(t_j,\cdot)\|_{L^2(P_W)}+\|\hat S(t_{j-1},\cdot)-S(t_{j-1},\cdot)\|_{L^2(P_W)}\big\} = o_P(1).
\end{multline*}
We note that $\hat g_j-g_j$ consists of three terms that are of similar forms. We study one of them in details and similar arguments apply to the other two, and we can then apply the triangle inequality to conclude that $\|\hat g_j-g_j\|_{L^2(Q)} = o_P(1)$. For notational convenience, we define
$$\hat\tau_k(y,\delta,w) = \sum_{l\leq k}-\frac{\hat S(t_k,w)}{\hat S(t_l,w)\hat H(t_l,w)}\left[I\{y=t_l,\delta=1\}-\hat h(t_l,w)I\{y \geq t_l\}\right].$$
We focus on the term in $\hat g_j-g_j$ that is given by
\begin{align*}
   &\ \left\{\frac{2\hat S(t_k,w)}{\hat S(t_j,w)}-\frac{2\hat S(t_k,w)}{\hat S(t_{j-1},w)}\right\}\hat\tau_k(y,\delta,w)-\left\{\frac{2S(t_k,w)}{S(t_j,w)}-\frac{2S(t_k,w)}{S(t_{j-1},w)}\right\}\tau_k(y,\delta,w) \\
   &= \underbrace{\left\{\frac{2\hat S(t_k,w)}{\hat S(t_j,w)}-\frac{2\hat S(t_k,w)}{\hat S(t_{j-1},w)}-\frac{2S(t_k,w)}{S(t_j,w)}+\frac{2S(t_k,w)}{S(t_{j-1},w)}\right\}\hat\tau_k(y,\delta,w)}_\text{term 1} \\
   &\quad+ \underbrace{\left\{\frac{2S(t_k,w)}{S(t_j,w)}-\frac{2S(t_k,w)}{S(t_{j-1},w)}\right\}\left\{\hat\tau_k(y,\delta,w)-\tau_k(y,\delta,w)\right\}}_\text{term 2}.
\end{align*}
The triangle inequality allows us to bound each term separately. Since $\hat S$ and $\hat H$ are uniformly bounded away from 0 and $\hat S$ and $\hat h$ are uniformly bounded above, there exists some constant $M$ such that $|\hat\tau_k| \leq M$. Therefore, in term 1 it suffices to upper bound $\|\frac{2\hat S(t_k,\cdot)}{\hat S(t_j,\cdot)}-\frac{2\hat S(t_k,\cdot)}{\hat S(t_{j-1},\cdot)}-\frac{2S(t_k,\cdot)}{S(t_j,\cdot)}+\frac{2S(t_k,\cdot)}{S(t_{j-1},\cdot)}\|_{L^2(P_W)}$. As in our analysis of $\|\hat f_j - f_j\|_{L^2(P_W)}$, we apply a first-order Taylor expansion, where the first order derivatives are bounded by some constant $M^{\prime}$. Then,
\begin{multline*}
    \left\|\frac{2\hat S(t_k,\cdot)}{\hat S(t_j,\cdot)}-\frac{2\hat S(t_k,\cdot)}{\hat S(t_{j-1},\cdot)}-\frac{2S(t_k,\cdot)}{S(t_j,\cdot)}+\frac{2S(t_k,\cdot)}{S(t_{j-1},\cdot)}\right\|_{L^2(P_W)} \leq M^{\prime}\big\{\|\hat S(t_k,\cdot)-S(t_k,\cdot)\|_{L^2(P_W)} \\
    +\|\hat S(t_j,\cdot)-S(t_j,\cdot)\|_{L^2(P_W)}+\|\hat S(t_{j-1},\cdot)-S(t_{j-1},\cdot)\|_{L^2(P_W)}\big\} = o_P(1).
\end{multline*}
Thus, term 1 is indeed $o_P(1)$. As $S$ is uniformly bounded away from 0, $\frac{2S(t_k,w)}{S(t_j,w)}-\frac{2S(t_k,w)}{S(t_{j-1},w)}$ is bounded above. Therefore, it suffices to show that $\|\hat\tau_k-\tau_k\|_{L^2(Q)} = o_P(1)$. Note that
\begin{align*}
    &\hat\tau_k(y,\delta,w)-\tau_k(y,\delta,w) \\
    &=  \underbrace{\sum_{l\leq k}-\frac{\hat S(t_k,w)I\{y \geq t_l\}}{\hat S(t_l,w)\hat H(t_l,w)}\left\{h(t_l,w)-\hat h(t_l,w)\right\}}_\text{term 2.1} \\
    &\quad+ \underbrace{\sum_{l\leq k}\left\{\frac{S(t_k,w)}{ S(t_l,w) H(t_l,w)}-\frac{\hat S(t_k,w)}{\hat S(t_l,w)\hat H(t_l,w)}\right\}\left[I\{y=t_l,\delta=1\}-h(t_l,w)I\{y \geq t_l\}\right]}_\text{term 2.2}.
\end{align*}
To see that Term 2.1 is $o_P(1)$, note that the first factor is bounded above and $\|\hat h(t,\cdot)-h(t,\cdot)\|_{L^2(P_W)} = o_P(1)$. To see that Term 2.2 is $o_P(1)$, note that the second factor is bounded above and the first factor is $o_P(1)$, which can be shown again by a Taylor expansion. This argument shows one of the three terms in $\hat g_j- g_j$ is $o_P(1)$. A similar argument can be applied to show the other two terms are $o_P(1)$ as well. Finally we show that $\sigma_a^2(\hat Q) -\sigma^2_a(Q) = o_P(1)$. Note that
\begin{align*}
    \sigma_a^2(\hat Q) -\sigma^2_a(Q) &= \sum_{j=1}^{k}\left(Q_n\hat f_j - Qf_j\right) \\
    &= \sum_{j=1}^{k}\left\{(Q_n-Q)f_j + Q(\hat f_j-f_j) + (Q_n-Q)(\hat f_j-f_j)\right\}.
\end{align*}
The term $(Q_n-Q)f_j$ is $o_P(1)$ by the law of large numbers. We have shown that $\|\hat f_j-f_j\|_{L^2(Q)}$ is $o_P(1)$, which provided an upper bound on $Q(\hat f_j-f_j)$. As we will show momentarily, $\hat f_j$ lies in a $Q$-Donsker class with probability tending to 1, so Lemma 19.24 in \citet{van2000asymptotic} implies that $(Q_n-Q)(\hat f_j-f_j) = o_P(1)$. Combining these results, we have $\sigma_a^2(\hat Q) -\sigma^2_a(Q) = o_P(1)$.

Next we establish condition (2), which says that $\textnormal{IF}_a(\hat Q)$ lies in a $Q$-Donsker class with probability tending to 1. By Theorem 2.10.6 in \citet{van1996weak}, it suffices to show that $\hat g_j$ and $\hat f_j$ both lie in $Q$-Donsker classes with probability tending to 1, since $G$ is bounded away from 0. This can be shown again by Theorem 2.10.6 in \citet{van1996weak}, as by assumption $\hat S$, $\hat H$ and $\hat h$ all belong to fixed $Q$-Donsker classes with probability tending to 1 and all the functions involved in $\hat g_j$ and $\hat f_j$ are uniformly bounded away from 0 and also bounded above. 

Conditions (1) and (2) allows us to apply Lemma 19.24 in \citet{van2000asymptotic} to conclude that $(Q_n-Q)\{\textnormal{IF}_a(\hat Q)-\textnormal{IF}_a(Q)\}$ is $o_P(n^{-1/2})$.

Now, combining step 1, which showed that $R(\hat Q,Q)$ is $o_P(n^{-1/2})$, and step 2, which showed that $(Q_n-Q)\{\textnormal{IF}_a(\hat Q)-\textnormal{IF}_a(Q)\}$ is $o_P(n^{-1/2})$, we see that $\hat\sigma_a^2$ is asymptotically linear with influence function $\textnormal{IF}_a$. Theorem \ref{RD} then follows by the delta method. Specifically, the influence function of $\hat\phi$ is given by $(\textnormal{IF}_a - \phi_a \textnormal{IF}_u)/\sigma_u^2$. This estimator is efficient as its influence function agrees with the EIF of $\phi_a$.
\end{proof}

\begin{proof}[Proof of Lemma \ref{reRMST}]
The unadjusted estimator, namely $\hat\psi_u = \sum_{j=1}^{k}\{\tilde S_1(t_j) - \tilde S_0(t_j)\}$, has influence function $\sum_{j=1}^{k}( \eta_{1,j} - \eta_{0,j})$. Under the sharp null, the form of the influence function simplifies to $\sum_{j=1}^{k}D_{u,j}$, where the definition of $D_{u,j}$ is given in \eqref{IFnull}.
Since
\begin{align*}
    E[D_{u,j}D_{u,l}] &= \frac{1}{\pi_1\pi_0}S(t_j)S(t_l)\sum_{u=1}^{\min(j,l)}\frac{S(t_{u-1})-S(t_u)}{S(t_{u-1})S(t_u)G(t_u)},
\end{align*}
we have that
\begin{equation*}
    \text{var}\left(\sum_{j=1}^{k}D_{u,j}\right) = \frac{1}{\pi_1\pi_0}\sum_{j=1}^{k}\sum_{l=1}^{k}\sum_{u=1}^{\min(j,l)}\frac{S(t_j)S(t_l)\{S(t_{u-1})-S(t_u)\}}{S(t_{u-1})S(t_u)G(t_u)}.
\end{equation*}

Now we consider the fully adjusted estimator $\hat\psi_a = \sum_{j=1}^{k}\{\hat S_1(t_j) - \hat S_0(t_j)\}$. Under the sharp null, the influence function of $\hat\psi_a$ is given by $\sum_{j=1}^k D_{a,j}$, where $D_{a,j}$ is defined in \eqref{Dak}. As in the proof of Lemma \ref{reRD}, the variance of $\sum_{j=1}^k D_{a,j}$ can be derived in the same way, except that we condition on $W$ first when taking expectations.
\begin{equation*}
    \text{var}\left(\sum_{j=1}^{k}D_{a,j}\right) = \frac{1}{\pi_1\pi_0}\sum_{j=1}^{k}\sum_{l=1}^{k}\sum_{u=1}^{\min(j,l)}E_P\left[\frac{S(t_j,W)S(t_l,W)\{S(t_{u-1},W)-S(t_u,W)\}}{S(t_{u-1},W)S(t_u,W)G(t_u,W)}\right].
\end{equation*}
\end{proof}

\begin{proof}[Proof of Theorem \ref{RMST}]
$\sigma_u^2$ is a linear combination of $s_{u}^{jl}$. The asymptotic linearity of $\hat s_{u}^{jl}$ implies that $\hat\sigma_u^2$ is asymptotically linear with influence function 
\begin{equation*}
    \textnormal{IF}_u(y,\delta,w) = \sum_{j=1}^{k}\sum_{l=1}^{k}\sum_{u}^{\min(j,l)}\xi_u^{jl}(y,\delta,w)/G(t_u).
\end{equation*}
Moreover, $\hat\sigma_a^2$ is again a one-step estimator based on its EIF given in Lemma \ref{eifsurvvar}. Using the same approach as was used in the proof of Theorem \ref{RD}, it can be shown that the remainder term $R(\hat Q, Q)$ is $o_P(n^{-1/2})$, and $(Q_n-Q)\{\textnormal{IF}_a(\hat Q)-\textnormal{IF}_a(Q)\} = o_P(n^{-1/2})$. Due to their close similarity to earlier arguments, we omit the details here. We can then conclude that $\hat\sigma_a^2$ has influence function $\textnormal{IF}_a = \sum_{j=1}^{k}\sum_{l=1}^{k}\sum_{u=1}^{\min(j,l)}D_{u}^{jl}$, where the definition of $D_{u}^{jl}$ is given in \eqref{Djkl} in Lemma~B.1. 

Applying the delta method, the influence function of $\hat\phi$ is given by $(\textnormal{IF}_a - \phi_a \textnormal{IF}_u)/\sigma_u^2$. This estimator is efficient as its influence function agrees with the EIF of $\phi_a$.
\end{proof}

\begin{proof}[Proof of Theorem~\ref{corollaryreRD}]
We prove the first claim by showing that $E_P\left[f_{u}^{jl}(W)/G(t_u,W)\right] \geq s_u^{jl}/G(t_u)$, for all $(u,j,l)$ such that $\max\{j,l\} \leq k$ and $u\leq\min(j,l)$. To start, we note that $T \perp W$ under $P$ implies that $f_u^{jl}(w) = s_u^{jl}$ for all $w \in \mathcal{W}$. Thus, it suffices to show that $E_P[1/G(t_u,W)] \geq 1/G(t_u)$, which follows from the convexity of the function $a \mapsto 1/a$ for $a>0$ and Jensen's inequality. Strict inequality holds when $\textnormal{var}_P[G(t_u,W)] >0$ for some $u$.

To prove the second claim, we note that $C \perp W$ implies that $G(t_j,w) = G(t_j)$ for all $w \in \mathcal{W}$ and $j \in \{1,\ldots,k\}$. We focus on the case of RD and RR first. Consider a bivariate function $(a,b) \mapsto a^2/b$ for $(a,b) \in (0,1)^2$. The Hessian matrix is given by $\begin{bmatrix} 2/b & -2a/b^2 \\ -2a/b^2 & 2a^2/b^3\end{bmatrix}$ with eigenvalues 0 and $2(a^2+b^2)/b^3$, both of which are non-negative. Therefore, this function is convex for $(a,b) \in (0,1)^2$. We note that
\begin{align*}
    \sigma_a^2 &= E_P\left[\sum_{j=1}^{k}S^2(t_k,W)\left\{\frac{1}{S(t_j,W)}-\frac{1}{S(t_{j-1},W)}\right\}\frac{1}{G(t_j)}\right] \\
    &= E_P\left[\frac{S(t_k,W)}{G(t_k)}-\sum_{j=1}^{k-1}\frac{S^2(t_k,W)}{S(t_j,W)}\left\{\frac{1}{G(t_{j+1})}-\frac{1}{G(t_j)}\right\}-\frac{S^2(t_k,W)}{S(t_0,W)G(t_1)}\right] \\
    &= \frac{S(t_k)}{G(t_k)}-\sum_{j=1}^{k-1}\left\{\frac{1}{G(t_{j+1})}-\frac{1}{G(t_j)}\right\}E_P\left[\frac{S^2(t_k,W)}{S(t_j,W)}\right]-\frac{1}{G(t_1)}E_P\left[\frac{S^2(t_k,W)}{S(t_0,W)}\right] \\
    &\leq \frac{S(t_k)}{G(t_k)}-\sum_{j=1}^{k-1}\left\{\frac{1}{G(t_{j+1})}-\frac{1}{G(t_j)}\right\}\frac{S^2(t_k)}{S(t_j)}-\frac{S^2(t_k)}{S(t_0)G(t_1)} \\
    &= \sum_{j=1}^{k}S^2(t_k)\left\{\frac{1}{S(t_j)}-\frac{1}{S(t_{j-1})}\right\}\frac{1}{G(t_j)} = \sigma_u^2,
\end{align*}
where the inequality follows from Jensen's inequality on the function $(a,b) \mapsto a^2/b$ for $(a,b) \in (0,1)^2$. 

For the case of RMST, we consider a $(q+1)$-variate function that maps $(a_1,\ldots,a_q,b)$ to $(\sum_{i=1}^{q}a_i)^2/b$ for $a_i \in (0,1)$ and $b \in (0,1)$. The only non-zero eigenvalue of its Hessian matrix is $2\{(\sum_{i=1}^{q}a_i)^2+qb^2\}/b^3$, which is positive. Therefore, this function is convex for any $q\geq 1$. In the following argument, we will take $q \in \{1,\ldots,k\}$. 
\begin{align*}
    \sigma_a^2 &= \sum_{j=1}^{k}\sum_{l=1}^{k}\sum_{u=1}^{\min(j,l)}E_{P}[f_{u}^{jl}(W)/G(t_u,W)]\\
    &= E_P\left[\sum_{j=1}^{k}\sum_{l=1}^{k}\sum_{u=1}^{\min\{j,l\}}S(t_j,W)S(t_l,W)\left\{\frac{1}{S(t_u,W)}-\frac{1}{S(t_{u-1},W)}\right\}\frac{1}{G(t_u)}\right] \\
    &= E_P\Bigg[\sum_{j=1}^{k}\sum_{l=1}^{k}\Bigg\{\frac{S(t_{\max\{j,l\}},W)}{G(t_{\min\{j,l\}})}-\sum_{u=1}^{\min\{j,l\}-1}\frac{S(t_j,W)S(t_l,W)}{S(t_u,W)}\left\{\frac{1}{G(t_{u+1})}-\frac{1}{G(t_u)}\right\}\\
    &\hspace{7.5em}-\frac{S(t_j,W)S(t_l,W)}{S(t_0,W)G(t_1)}\Bigg\}\Bigg] \\
    &= \sum_{j=1}^{k}\sum_{l=1}^{k}\frac{S(t_{\max\{j,l\}})}{G(t_{\min\{j,l\}})}-E_P\left[\sum_{u=1}^{k-1}\sum_{j=u+1}^{k}\sum_{l=u+1}^{k}\frac{S(t_j,W)S(t_l,W)}{S(t_u,W)}\left\{\frac{1}{G(t_{u+1})}-\frac{1}{G(t_u)}\right\}\right] \\
    &\ \ \ -E_P\left[\sum_{j=1}^{k}\sum_{l=1}^{k}\frac{S(t_j,W)S(t_l,W)}{S(t_0,W)G(t_1)}\right] \\
    &= \sum_{j=1}^{k}\sum_{l=1}^{k}\frac{S(t_{\max\{j,l\}})}{G(t_{\min\{j,l\}})}-E_P\left[\sum_{u=1}^{k-1}\frac{\left\{\sum_{j=u+1}^{k}S(t_j,W)\right\}^2}{S(t_u,W)}\left\{\frac{1}{G(t_{u+1})}-\frac{1}{G(t_u)}\right\}\right]\\
    &\quad-E_P\left[\frac{\left\{\sum_{j=1}^{k}S(t_j,W)\right\}^2}{S(t_0,W)G(t_1)}\right] \\
    &\leq \sum_{j=1}^{k}\sum_{l=1}^{k}\frac{S(t_{\max\{j,l\}})}{G(t_{\min\{j,l\}})}-\left[\sum_{u=1}^{k-1}\frac{\left\{\sum_{j=u+1}^{k}S(t_j)\right\}^2}{S(t_u)}\left\{\frac{1}{G(t_{u+1})}-\frac{1}{G(t_u)}\right\}\right]-\left[\frac{\left\{\sum_{j=1}^{k}S(t_j)\right\}^2}{S(t_0)G(t_1)}\right] \\
    &= \sum_{j=1}^{k}\sum_{l=1}^{k}\sum_{u=1}^{\min\{j,l\}}S(t_j)S(t_l)\left\{\frac{1}{S(t_u)}-\frac{1}{S(t_{u-1})}\right\}\frac{1}{G(t_u)} = \sigma_u^2,
\end{align*}
where the inequality follows from Jensen's inequality.

\end{proof}

\section{Additional Results from the Numerical Experiments}\label{additionalsim}

\subsection{Additional simulation results}
In this section, we present additional simulation results for sample sizes $n=200$ and $n=500$. The simulation set-up is otherwise the same as described in Section~\ref{sec:experiment}. The results for ordinal outcomes are presented in Table~\ref{addsimordinal}, and the results for survival outcomes are presented in Table~\ref{addsimsurvival}.

\begin{table}
\caption{\label{addsimordinal}Simulation results for ordinal outcome. We consider relative efficiency of fully adjusted and working-model-based estimators for DIM, MW and LOR. In the bootstrap approach, we take $B_1 = 100$ and $B_2 = 500$. Results are based on 1000 replications for analytic approach; 200 for bootstrap.}
    \centering
    \begin{tabular}{rrrrrrrr}
    \hline
    \multicolumn{8}{c}{$n=200$} \\
    \hline
         & truth & method & bias & MSE & $\%$RMSE & coverage & CI width \\
        \hline
        DIM (F) & 0.837 & analytic & 0.008 & 0.002 & 0.060 & 0.972 & 0.194 \\
        DIM (P) & 0.840 & analytic & -0.008 & 0.002 & 0.057 & 0.950 & 0.183 \\
        & & bootstrap & -0.001 & 0.003 & 0.068 & 0.945 & 0.236 \\
        \hline
        MW (F) & 0.842 & analytic & 0.011 & 0.003 & 0.060 & 0.977 & 0.200\\
        MW (P) & 0.845 & analytic & -0.004 & 0.002 & 0.056 & 0.951 & 0.186 \\
        && bootstrap & 0.000 & 0.003 & 0.069 & 0.945 & 0.244 \\
        \hline
        LOR (F) & 0.838 & analytic & 0.013 & 0.003 & 0.062 & 0.973 & 0.197\\
        LOR (P) & 0.842 & analytic & -0.002 & 0.002 & 0.054 & 0.961 & 0.182 \\
        && bootstrap & -0.001 & 0.003 & 0.066 & 0.945 & 0.230 \\
        \hline
    \multicolumn{8}{c}{$n=500$} \\
    \hline
         & truth & method & bias & MSE & $\%$RMSE & coverage & CI width \\
        \hline
        DIM (F) & 0.837 & analytic & 0.003 & 0.001 & 0.037 & 0.946 & 0.119 \\
        DIM (P) & 0.840 & analytic & -0.004 & 0.001 & 0.036 & 0.946 & 0.115 \\
        & & bootstrap & -0.004 & 0.002 & 0.051 & 0.965 & 0.177 \\
        \hline
        MW (F) & 0.842 & analytic & 0.008 & 0.001 & 0.038 & 0.945 & 0.120\\
        MW (P) & 0.845 & analytic & 0.001 & 0.001 & 0.036 & 0.943 & 0.117 \\
        & & bootstrap & -0.004 & 0.002 & 0.053 & 0.945 & 0.184 \\
        \hline
        LOR (F) & 0.838 & analytic & 0.006 & 0.001 & 0.038 & 0.957 & 0.120\\
        LOR (P) & 0.842 & analytic & -0.000 & 0.001 & 0.035 & 0.953 & 0.114 \\
        & & bootstrap & -0.003 & 0.002 & 0.049 & 0.955 & 0.170\\
        \hline
    \end{tabular}
\end{table}

\begin{table}
   
    \caption{\label{addsimsurvival}Simulation results for survival outcomes. We only consider the analytical approach and consider relative efficiency of fully adjusted estimators for RD at time 1, 2 and 3 (the relative efficiency is the same for RR) and RMST at time 3. Sample size is 200 or 500, and results are based on 1000 replications.}
    \centering
    \begin{tabular}{rrrrrrr}
    \hline
    \multicolumn{7}{c}{$n=200$} \\
    \hline
     & truth & bias & MSE & $\%$RMSE & coverage & mean width \\
    \hline
    RD ($t=1$) & 0.903 & 0.003 & 0.002 & 0.045 & 0.948 & 0.165\\
    RD ($t=2$) & 0.847 &  0.004 & 0.003 & 0.061 & 0.933 & 0.196 \\
    RD ($t=3$) & 0.819 & 0.007 & 0.004 & 0.074 & 0.908 & 0.212 \\
    RMST ($t=3$) & 0.820 & 0.001 & 0.003 & 0.061 & 0.936 & 0.202\\
    \hline
    \multicolumn{7}{c}{$n=500$} \\
    \hline
     & truth & bias & MSE & $\%$RMSE & coverage & mean width \\
    \hline
    RD ($t=1$) & 0.903 & 0.001 & 0.001 & 0.028 & 0.953 & 0.103\\
    RD ($t=2$) & 0.847 & 0.002 & 0.001 & 0.039 & 0.949 & 0.126  \\
    RD ($t=3$) & 0.819 & 0.004 & 0.002 & 0.048 & 0.916 & 0.140 \\
    RMST ($t=3$) & 0.820 & -0.001 & 0.001 & 0.043 & 0.930 & 0.128\\
    \hline
    \end{tabular}
\end{table}

\subsection{Application to Covid-19 data: ordinal outcomes}
In this section, we present additional results when applying our proposed methods to the Covid-19 dataset with ordinal outcomes. In particular, we estimate the efficiency gain from using the fully adjusted and working-model-based estimators that adjust for one of the covariates, for estimating three treatment effect estimands: DIM (Table \ref{ordinaloneDIM}), MW (Table \ref{ordinaloneMW}) and LOR (Table \ref{ordinaloneLOR}).

\begin{table}
     \caption{\label{ordinaloneDIM} Relative efficiency (95\% CI) of fully adjusted and working-model-based estimators that adjust for one of the covariates for estimating DIM, in the Covid-19 dataset. ``F" stands for the fully adjusted estimator, and ``W" stands for the working-model-based estimator.}
    \centering
    \begin{tabular}{lrr}
        \hline
         & \multicolumn{2}{c}{DIM} \\
          & F & W \\
        \hline
        age & 0.97 (0.94, 1.00) & 0.97 (0.94, 1.00) \\
        gender & 0.99 (0.97, 1.01) & 0.99 (0.97, 1.01) \\
        race & 0.99 (0.98, 1.01) & 0.99 (0.98, 1.01) \\
        CVD & 0.95 (0.92, 0.99) & 0.95 (0.92, 0.99) \\
        HTN & 1.00 (1.00, 1.00) & 1.00 (1.00, 1.00) \\
        diabetes & 0.99 (0.98, 1.01) & 0.99 (0.98, 1.01) \\
        kidney disease & 0.99 (0.96, 1.01) & 0.99 (0.96, 1.01)\\
        cholesterol meds & 1.00 (1.00, 1.00) & 1.00 (1.00, 1.00) \\
        HTN meds & 0.99 (0.98, 1.01) & 0.99 (0.98, 1.01) \\
        BMI & 1.00 (1.00, 1.00) & 1.00 (1.00, 1.00) \\
        \hline
        \end{tabular}
\end{table}

\begin{table}
     \caption{\label{ordinaloneMW} Relative efficiency (95\% CI) of fully adjusted and working-model-based estimators that adjust for one of the covariates for estimating MW, in the Covid-19 dataset. ``F" stands for the fully adjusted estimator, and ``W" stands for the working-model-based estimator.}
    \centering
    \begin{tabular}{lrr}
        \hline
        & \multicolumn{2}{c}{MW} \\
        & F & W \\
        \hline
        age & 0.98 (0.96, 1.01) & 0.98 (0.96, 1.01) \\
        gender & 0.99 (0.97, 1.01) & 0.99 (0.91, 1.01) \\
        race & 0.99 (0.98, 1.01) & 0.99 (0.98, 1.01) \\
        CVD & 0.95 (0.92, 0.99) & 0.95 (0.92, 0.99)\\
        HTN & 1.00 (0.99, 1.00) & 1.00 (1.00, 1.00) \\
        diabetes & 0.99 (0.98, 1.01) & 0.99 (0.98, 1.01) \\
        kidney disease & 0.99 (0.97, 1.01) & 0.99 (0.97, 1.01)\\
        cholesterol meds & 1.00 (0.99, 1.00) & 1.00 (0.99, 1.00) \\
        HTN meds & 0.99 (0.97, 1.01) & 0.99 (0.97, 1.01) \\
        BMI & 1.00 (1.00, 1.00) & 1.00 (1.00, 1.00) \\
        \hline
        \end{tabular}
\end{table}

\begin{table}
     \caption{\label{ordinaloneLOR} Relative efficiency (95\% CI) of fully adjusted and working-model-based estimators that adjust for one of the covariates for estimating LOR, in the Covid-19 dataset. ``F" stands for the fully adjusted estimator, and ``W" stands for the working-model-based estimator.}
    \centering
    \begin{tabular}{lrr}
        \hline
        & \multicolumn{2}{c}{LOR} \\
          & F & W \\
        \hline
        age & 0.97 (0.95, 0.99) & 0.97 (0.93, 1.00)\\
        gender & 0.99 (0.98, 1.02) & 0.99 (0.98, 1.01)\\
        race & 1.00 (0.98, 1.01) & 1.00 (0.98, 1.01)\\
        CVD  & 0.96 (0.92, 0.99) & 0.96 (0.92, 0.99)\\
        HTN & 1.00 (1.00, 1.00) & 1.00 (1.00, 1.00)\\
        diabetes & 0.99 (0.98, 1.01) & 0.99 (0.98, 1.01)\\
        kidney disease & 0.99 (0.97, 1.01) & 0.99 (0.96, 1.01)\\
        cholesterol meds  & 1.00 (1.00, 1.01) & 1.00 (1.00, 1.00)\\
        HTN meds & 0.99 (0.98, 1.01) & 0.99 (0.98, 1.01)\\
        BMI  & 1.00 (1.00, 1.00) & 1.00 (1.00, 1.00)\\
        \hline
        \end{tabular}
\end{table}
\section{Statistical Inference for the Relative Effciency when $Y$ and $W$ are Independent under $P$}\label{unionset}

\subsection{Hypothesis test based on sample splitting}
We split the external data into two subsamples of size $n/2$, denoted by $D_1$ and $D_2$. Let $\hat\sigma_{a,1}^{2}$ be the proposed estimator of the adjusted variance, calculated from observations in $D_1$ only; and let $\hat\sigma_{u,2}^{2}$ be the proposed estimator for the unadjusted variance using $D_2$. By the asymptotic linearity of these estimators, we have that
\begin{equation*}
    \sqrt{n}(\hat\sigma_{a,1}^2 - \sigma^2_a) \xrightarrow[]{d} N(0,2E[\textnormal{IF}_a^2]), \ \ \sqrt{n}(\hat\sigma_{u,2}^2 - \sigma^2_u) \xrightarrow[]{d} N(0,2E[\textnormal{IF}_u^2]).
\end{equation*}
Moreover, as $\hat\sigma_{a,1}^{2}$ and $\hat\sigma_{u,2}^{2}$ are based on different observations, they are independent. Therefore, delta method implies that
\begin{equation*}
    \sqrt{n}\left(\frac{\hat\sigma_{a,1}^{2}}{\hat\sigma_{u,2}^{2}}-\phi_a\right) \xrightarrow[]{d} N\left(0,\frac{2E[\textnormal{IF}_a^2] + 2\phi_a^2E[\textnormal{IF}_u^2]}{\sigma_u^4}\right).
\end{equation*}
A Wald test can be used to test the hypothesis $H_0: \phi_a = 1$, and in fact a Wald confidence interval can also be constructed directly although it might be wider than the one obtained from the proposed two-step procedure. The same argument applies when we consider the working-model-based variance.

The sample splitting approach was also used in \citet{williamson2020efficient} to test a null hypothesis that lies on the boundary of the parameter space.

\subsection{Asymptotic coverage of the confidence set}
Recall that the confidence set, which we denote as $I_{ts}$, is constructed using a two-step procedure. We first test the null hypothesis $H_0: \phi=1$ using a level $\alpha$ test. If it is rejected, we take the confidence set to be $I_{wald}$, the Wald confidence interval; otherwise the confidence set is taken to be $I_{wald} \cup \{1\}$. We argue that the asymptotic coverage of this confidence set is at least $1-\alpha$.

First, consider the case where $\phi \neq 1$ under $P$. This implies that the influence function of $\hat\phi$ without sample splitting is not identically 0. Hence, asymptotic linearity in the form of \eqref{normality} implies that $P(\phi \in I_{wald}) \to 1-\alpha$. In addition, we have that $P(\phi \in I_{wald}) \leq P(\phi \in I_{ts})$. Next, consider $P$ such that $\phi = 1$. 
\begin{align*}
    P(\phi \in I_{ts}) &= P(\phi \in I_{ts} | \text{ reject } H_0)P(\text{ reject } H_0) \\
    &\ \ \ + P(\phi \in I_{ts} | \text{ do not reject } H_0)P(\text{ do not reject } H_0) \\
    &= P(1 \in I_{wald} |\text{ reject } H_0)P(\text{ reject } H_0) \\
    &\ \ \ + P(1 \in I_{wald} \cup \{1\} | \text{ do not reject } H_0)P(\text{ do not reject } H_0) \\
    & \geq P(\text{ do not reject } H_0) \geq 1-\alpha,
\end{align*}
where the last inequality follows because the test is level $\alpha$ and thus the probability of falsely rejecting the null is at most $\alpha$.
\section{Pseudocode for double bootstrap scheme} \label{app:pseudocode}

\begin{algorithm}
\caption{Double bootstrap procedure}\label{algbootstrap}
\hspace*{\algorithmicindent} \textbf{Input:} external data $X$, treatment effect estimators $\hat\psi_u$ and $\hat\psi_m$ \\
 \hspace*{\algorithmicindent} \textbf{Output:} Estimate of and confidence interval for the relative efficiency 
\begin{algorithmic}[1]
\For {$k=1,2,\ldots,B_2$}
    \State Resample $X_k$ of size $N$ from $X$ with replacement
    \State Randomly assign treatment to form $\tilde X_{k}$
    \State Compute $\hat\psi_{u,k} = \hat\psi_u(\tilde X_k), \hat\psi_{m,k} = \hat\psi_m(\tilde X_k)$
\EndFor
\State Let $\tilde\phi(X) = var(\hat\psi_{m,k})/var(\hat\psi_{u,k})$
\For {$i=1,2,\ldots,B_1$}
    \State Resample $X_i^*$ of size $n$ from $X$ with replacement
	    \For {$j=1,2,\ldots,B_2$}
			\State Resample $X_{ij}^{**}$ of size $N$ from $X_i^*$ with replacement
			\State Randomly assign treatment to form $\tilde X_{ij}$
			\State Compute $\hat\psi_u^{ij} = \hat\psi_u(\tilde X_{ij}), \hat\psi_m^{ij} = \hat\psi_m(\tilde X_{ij})$
		\EndFor
	\State $\tilde\phi(X_i^*) = var(\hat\psi_{m}^{ij})/var(\hat\psi_{u}^{ij})$
\EndFor \\
\Return $\tilde\phi(X)$ and $\tilde\phi(X) \pm z_{1-\alpha/2} \text{ sd}(\tilde\phi(X_i^*))$
\end{algorithmic}
\end{algorithm}

%\end{appendices}

\section*{Acknowledgements}

The authors gratefully acknowledge the support of the NIH through award number DP2-LM013340. The content is solely the responsibility of the authors
and does not necessarily represent the official views of the NIH.

\bibliography{output}

\begin{thebibliography}{}

\bibitem[\protect\citeauthoryear{Anthony and Bartlett}{Anthony and
  Bartlett}{2009}]{anthony2009neural}
Anthony, M. and P.~L. Bartlett (2009).
\newblock {\em Neural network learning: Theoretical foundations}.
\newblock cambridge university press.

\bibitem[\protect\citeauthoryear{Austin, Manca, Zwarenstein, Juurlink, and
  Stanbrook}{Austin et~al.}{2010}]{austin2010substantial}
Austin, P.~C., A.~Manca, M.~Zwarenstein, D.~N. Juurlink, and M.~B. Stanbrook
  (2010).
\newblock A substantial and confusing variation exists in handling of baseline
  covariates in randomized controlled trials: a review of trials published in
  leading medical journals.
\newblock {\em Journal of clinical epidemiology\/}~{\em 63\/}(2), 142--153.

\bibitem[\protect\citeauthoryear{Benkeser, Diaz, Luedtke, Segal, Scharfstein,
  and Rosenblum}{Benkeser et~al.}{2020}]{benkeser2020improving}
Benkeser, D., I.~Diaz, A.~Luedtke, J.~Segal, D.~Scharfstein, and M.~Rosenblum
  (2020).
\newblock Improving precision and power in randomized trials for covid-19
  treatments using covariate adjustment, for binary, ordinal, and time-to-event
  outcomes.
\newblock {\em medRxiv\/}.

\bibitem[\protect\citeauthoryear{Bialek, Boundy, Bowen, Chow, Cohn, Dowling,
  Ellington, et~al.}{Bialek et~al.}{2020}]{cdc2020severe}
Bialek, S., E.~Boundy, V.~Bowen, N.~Chow, A.~Cohn, N.~Dowling, S.~Ellington,
  et~al. (2020).
\newblock Severe outcomes among patients with coronavirus disease 2019
  ({COVID}-19)—{U}nited {S}tates, {F}ebruary 12--{M}arch 16, 2020.
\newblock {\em Morbidity and mortality weekly report\/}~{\em 69\/}(12),
  343--346.

\bibitem[\protect\citeauthoryear{Bickel}{Bickel}{1982}]{bickel1982adaptive}
Bickel, P.~J. (1982).
\newblock On adaptive estimation.
\newblock {\em The Annals of Statistics\/}, 647--671.

\bibitem[\protect\citeauthoryear{Bickel, Klaassen, Bickel, Ritov, Klaassen,
  Wellner, and Ritov}{Bickel et~al.}{1993}]{bickel1993efficient}
Bickel, P.~J., C.~A. Klaassen, P.~J. Bickel, Y.~Ritov, J.~Klaassen, J.~A.
  Wellner, and Y.~Ritov (1993).
\newblock {\em Efficient and adaptive estimation for semiparametric models},
  Volume~4.
\newblock Johns Hopkins University Press Baltimore.

\bibitem[\protect\citeauthoryear{Colantuoni and Rosenblum}{Colantuoni and
  Rosenblum}{2015}]{colantuoni2015leveraging}
Colantuoni, E. and M.~Rosenblum (2015).
\newblock Leveraging prognostic baseline variables to gain precision in
  randomized trials.
\newblock {\em Statistics in medicine\/}~{\em 34\/}(18), 2602--2617.

\bibitem[\protect\citeauthoryear{D{\'\i}az, Colantuoni, Hanley, and
  Rosenblum}{D{\'\i}az et~al.}{2019}]{diaz2019improved}
D{\'\i}az, I., E.~Colantuoni, D.~F. Hanley, and M.~Rosenblum (2019).
\newblock Improved precision in the analysis of randomized trials with survival
  outcomes, without assuming proportional hazards.
\newblock {\em Lifetime data analysis\/}~{\em 25\/}(3), 439--468.

\bibitem[\protect\citeauthoryear{D{\'\i}az, Colantuoni, and
  Rosenblum}{D{\'\i}az et~al.}{2016}]{diaz2016enhanced}
D{\'\i}az, I., E.~Colantuoni, and M.~Rosenblum (2016).
\newblock Enhanced precision in the analysis of randomized trials with ordinal
  outcomes.
\newblock {\em Biometrics\/}~{\em 72\/}(2), 422--431.

\bibitem[\protect\citeauthoryear{FDA}{FDA}{2019}]{FDA2019}
FDA (2019).
\newblock Adjusting for covariates in randomized clinical trials for drugs and
  biologics with continuous outcomes. draft guidance for industry.
  https://www.fda.gov/media/123801/download.

\bibitem[\protect\citeauthoryear{Friedman, Hastie, and Tibshirani}{Friedman
  et~al.}{2010}]{Friedman2010regularization}
Friedman, J., T.~Hastie, and R.~Tibshirani (2010).
\newblock Regularization paths for generalized linear models via coordinate
  descent.
\newblock {\em Journal of Statistical Software\/}~{\em 33\/}(1), 1--22.

\bibitem[\protect\citeauthoryear{Gill, Van Der~Laan, and Robins}{Gill
  et~al.}{1997}]{gill1997coarsening}
Gill, R.~D., M.~J. Van Der~Laan, and J.~M. Robins (1997).
\newblock Coarsening at random: Characterizations, conjectures,
  counter-examples.
\newblock In {\em Proceedings of the First Seattle Symposium in Biostatistics},
  pp.\  255--294. Springer.

\bibitem[\protect\citeauthoryear{Heitjan and Rubin}{Heitjan and
  Rubin}{1991}]{heitjan1991ignorability}
Heitjan, D.~F. and D.~B. Rubin (1991).
\newblock Ignorability and coarse data.
\newblock {\em The annals of statistics\/}, 2244--2253.

\bibitem[\protect\citeauthoryear{Hirose et~al.}{Hirose
  et~al.}{2016}]{hirose2016differentiability}
Hirose, Y. et~al. (2016).
\newblock On differentiability of implicitly defined function in
  semi-parametric profile likelihood estimation.
\newblock {\em Bernoulli\/}~{\em 22\/}(1), 589--614.

\bibitem[\protect\citeauthoryear{Ibragimov and Has’minskii}{Ibragimov and
  Has’minskii}{1981}]{ibragimov1981statistical}
Ibragimov, I. and R.~Has’minskii (1981).
\newblock Statistical estimation: asymptotic theory.

\bibitem[\protect\citeauthoryear{Kahan, Jairath, Dor{\'e}, and Morris}{Kahan
  et~al.}{2014}]{kahan2014risks}
Kahan, B.~C., V.~Jairath, C.~J. Dor{\'e}, and T.~P. Morris (2014).
\newblock The risks and rewards of covariate adjustment in randomized trials:
  an assessment of 12 outcomes from 8 studies.
\newblock {\em Trials\/}~{\em 15\/}(1), 139.

\bibitem[\protect\citeauthoryear{Kaplan and Meier}{Kaplan and
  Meier}{1958}]{kaplan1958nonparametric}
Kaplan, E.~L. and P.~Meier (1958).
\newblock Nonparametric estimation from incomplete observations.
\newblock {\em Journal of the American statistical association\/}~{\em
  53\/}(282), 457--481.

\bibitem[\protect\citeauthoryear{Mao}{Mao}{2018}]{mao2018causal}
Mao, L. (2018).
\newblock On causal estimation using-statistics.
\newblock {\em Biometrika\/}~{\em 105\/}(1), 215--220.

\bibitem[\protect\citeauthoryear{McCullagh}{McCullagh}{1980}]{mccullagh1980regression}
McCullagh, P. (1980).
\newblock Regression models for ordinal data.
\newblock {\em Journal of the Royal Statistical Society: Series B
  (Methodological)\/}~{\em 42\/}(2), 109--127.

\bibitem[\protect\citeauthoryear{Moore and van~der Laan}{Moore and van~der
  Laan}{2009a}]{moore2009application}
Moore, K. and M.~J. van~der Laan (2009a).
\newblock Application of time-to-event methods in the assessment of safety in
  clinical trials.
\newblock {\em Design and Analysis of Clinical Trials with Time-to-Event
  Endpoints. Taylor \& Francis\/}, 455--482.

\bibitem[\protect\citeauthoryear{Moore, Neugebauer, Valappil, and van~der
  Laan}{Moore et~al.}{2011}]{moore2011robust}
Moore, K.~L., R.~Neugebauer, T.~Valappil, and M.~J. van~der Laan (2011).
\newblock Robust extraction of covariate information to improve estimation
  efficiency in randomized trials.
\newblock {\em Statistics in medicine\/}~{\em 30\/}(19), 2389--2408.

\bibitem[\protect\citeauthoryear{Moore and van~der Laan}{Moore and van~der
  Laan}{2009b}]{moore2009increasing}
Moore, K.~L. and M.~J. van~der Laan (2009b).
\newblock Increasing power in randomized trials with right censored outcomes
  through covariate adjustment.
\newblock {\em Journal of biopharmaceutical statistics\/}~{\em 19\/}(6),
  1099--1131.

\bibitem[\protect\citeauthoryear{Pfanzagl}{Pfanzagl}{1990}]{pfanzagl1990estimation}
Pfanzagl, J. (1990).
\newblock Estimation in semiparametric models.
\newblock In {\em Estimation in Semiparametric Models}, pp.\  17--22. Springer.

\bibitem[\protect\citeauthoryear{Pfanzagl and Wefelmeyer}{Pfanzagl and
  Wefelmeyer}{1985}]{pfanzagl1985contributions}
Pfanzagl, J. and W.~Wefelmeyer (1985).
\newblock Contributions to a general asymptotic statistical theory.
\newblock {\em Statistics \& Risk Modeling\/}~{\em 3\/}(3-4), 379--388.

\bibitem[\protect\citeauthoryear{Reid}{Reid}{1981}]{reid1981influence}
Reid, N. (1981).
\newblock Influence functions for censored data.
\newblock {\em The Annals of Statistics\/}, 78--92.

\bibitem[\protect\citeauthoryear{Steingrimsson, Hanley, and
  Rosenblum}{Steingrimsson et~al.}{2017}]{steingrimsson2017improving}
Steingrimsson, J.~A., D.~F. Hanley, and M.~Rosenblum (2017).
\newblock Improving precision by adjusting for prognostic baseline variables in
  randomized trials with binary outcomes, without regression model assumptions.
\newblock {\em Contemporary clinical trials\/}~{\em 54}, 18--24.

\bibitem[\protect\citeauthoryear{Stitelman, De~Gruttola, and van~der
  Laan}{Stitelman et~al.}{2012}]{stitelman2012general}
Stitelman, O.~M., V.~De~Gruttola, and M.~J. van~der Laan (2012).
\newblock A general implementation of tmle for longitudinal data applied to
  causal inference in survival analysis.
\newblock {\em The international journal of biostatistics\/}~{\em 8\/}(1).

\bibitem[\protect\citeauthoryear{Tsiatis}{Tsiatis}{2007}]{tsiatis2007semiparametric}
Tsiatis, A. (2007).
\newblock {\em Semiparametric theory and missing data}.
\newblock Springer Science \& Business Media.

\bibitem[\protect\citeauthoryear{Van~der Laan, Laan, and Robins}{Van~der Laan
  et~al.}{2003}]{van2003unified}
Van~der Laan, M.~J., M.~Laan, and J.~M. Robins (2003).
\newblock {\em Unified methods for censored longitudinal data and causality}.
\newblock Springer Science \& Business Media.

\bibitem[\protect\citeauthoryear{Van Der~Laan and Rubin}{Van Der~Laan and
  Rubin}{2006}]{van2006targeted}
Van Der~Laan, M.~J. and D.~Rubin (2006).
\newblock Targeted maximum likelihood learning.
\newblock {\em The international journal of biostatistics\/}~{\em 2\/}(1).

\bibitem[\protect\citeauthoryear{Van~der Vaart}{Van~der
  Vaart}{2000}]{van2000asymptotic}
Van~der Vaart, A.~W. (2000).
\newblock {\em Asymptotic statistics}, Volume~3.
\newblock Cambridge university press.

\bibitem[\protect\citeauthoryear{Van Der~Vaart and Wellner}{Van Der~Vaart and
  Wellner}{1996}]{van1996weak}
Van Der~Vaart, A.~W. and J.~A. Wellner (1996).
\newblock Weak convergence.
\newblock In {\em Weak convergence and empirical processes}, pp.\  16--28.
  Springer.

\bibitem[\protect\citeauthoryear{Vermeulen, Thas, and Vansteelandt}{Vermeulen
  et~al.}{2015}]{vermeulen2015increasing}
Vermeulen, K., O.~Thas, and S.~Vansteelandt (2015).
\newblock Increasing the power of the mann-whitney test in randomized
  experiments through flexible covariate adjustment.
\newblock {\em Statistics in medicine\/}~{\em 34\/}(6), 1012--1030.

\bibitem[\protect\citeauthoryear{Wang, Ogburn, and Rosenblum}{Wang
  et~al.}{2019}]{wang2019analysis}
Wang, B., E.~L. Ogburn, and M.~Rosenblum (2019).
\newblock Analysis of covariance in randomized trials: More precision and valid
  confidence intervals, without model assumptions.
\newblock {\em Biometrics\/}~{\em 75\/}(4), 1391--1400.

\bibitem[\protect\citeauthoryear{Williamson and Feng}{Williamson and
  Feng}{2020}]{williamson2020efficient}
Williamson, B.~D. and J.~Feng (2020).
\newblock Efficient nonparametric statistical inference on population feature
  importance using shapley values.
\newblock {\em arXiv preprint arXiv:2006.09481\/}.

\end{thebibliography}

\end{document}